\documentclass[aps,pra,reprint,groupedaddress,nofootinbib,nopreprintnumbers]{revtex4-2}
\usepackage{amsfonts}
\usepackage{amsmath}
\usepackage{amssymb}
\usepackage{amsthm}
\usepackage{mathtools}
\usepackage{braket}
\usepackage{stackrel}
\usepackage{tensor}
\usepackage{cancel}
\usepackage{enumerate}
\usepackage{esint}
\usepackage{graphicx}
\usepackage{subcaption}
\usepackage{hyperref}
\usepackage{color}

\usepackage[capitalise]{cleveref}

\newcommand{\tr}{\text{tr}}
\newcommand{\sgn}{\text{sgn}}
\newcommand{\ammaG}{\text{\reflectbox{$\Gamma$}}}
\renewcommand{\Re}{\text{Re}}
\renewcommand{\Im}{\text{Im}}

\newtheorem{prop}{Proposition}

\begin{document}

\title{Exact calculation of entanglement negativity for a 1+1D massless scalar field using phase space methods}

\author{Jason Pye}
\email{jason.pye@su.se}
\affiliation{Nordita, Stockholm University and KTH Royal Institute of Technology, Hannes Alfvéns väg 12, SE-106 91 Stockholm, Sweden}

\author{Atharva Hingane}
\email{atharvah@iisc.ac.in}
\affiliation{Department of Physical Sciences, Indian Institute of Science Education and Research (IISER) Mohali, Sector 81, S.A.S. Nagar, Manauli PO 140306, India}
\affiliation{Department of Computational and Data Sciences, Indian Institute of Science, Bangalore 560012, India}

\author{Robert H.~Jonsson}
\email{robert.jonsson@mau.se}
\affiliation{Department of Materials Science and Applied Mathematics, Malmö University, SE-205 06, Malmö, Sweden}
\affiliation{Nordita, Stockholm University and KTH Royal Institute of Technology, Hannes Alfvéns väg 12, SE-106 91 Stockholm, Sweden}

\preprint{NORDITA-2026-064}

\date{29 June 2026}

\begin{abstract}
  Quantum fields exhibit a rich entanglement structure which is still not fully understood. In this work, we study the entanglement structure of the vacuum state of a massless scalar field in (1+1)-dimensions---a paradigmatic case for both high energy and condensed matter physics. We fully characterize the entanglement negativity between two arbitrary compact spacelike-separated regions of the field by calculating the logarithmic negativity along with the modes carrying it, called negativity cores. We achieve this using a framework based on the K\"ahler structure of Gaussian states, wherein we calculate the diagonalization of the operator associated with the partially-transposed restricted linear complex structure. In doing so, we extend the methods of this framework by proposing a basis-independent definition of the transpose operation. The explicit diagonalization we perform is enabled by a reformulation of the eigenvalue problem as a boundary value problem in the complex plane. Our results also suggest extensions to higher dimensions and fermionic fields.
\end{abstract}

\maketitle

\tableofcontents

\section{Introduction}

The vacuum state in quantum field theory (QFT) exhibits a rich structure of entanglement between spacelike-separated regions.
Two central problems are to quantify this entanglement, and moreover to identify the degrees of freedom which are entangled across these subsystems in order to determine how to detect this entanglement.
An important scenario is the case of two spacelike-separated bounded regions, since it corresponds to the entanglement which can be accessed by a pair of local observers interacting with the field.

Here we analytically solve both of these problems for two disjoint spacelike-separated intervals in the vacuum state of a (1+1)-dimensional free massless real scalar field.
The massless scalar field in 1+1D serves as a basic model which is used to describe systems from the electromagnetic field in superconducting circuits to toy models of quantum fields in curved spacetimes.
It is therefore both application-minded and of fundamental theoretical interest to understand its entanglement structure.
Specifically, we compute the logarithmic negativity between arbitrary spacelike-separated intervals (and show its dependence on the geometry through the cross ratio), as well as construct a decomposition of the joint subsystem into pairs of modes carrying the entanglement, called negativity cores.
This provides a characterization of the mixed state entanglement between these two intervals in a manner analogous to the entanglement between a single subsystem of the field and its complement.

The case of a bipartition of the (pure) vacuum state of a field into a local subsystem and its complement---for which the entanglement can be quantified by the von Neumann entropy of the reduced state of the subsystem---has been well-studied in QFT~\cite{calabrese_entanglement_2004,calabrese_entanglement_2009,casini_entanglement_2009,casini_lectures_2022,bianchiEntropySubalgebraObservables2019}.
The dependence of the entanglement entropy on the geometry of a subsystem has been of particular interest, notably in the form of area laws in black hole physics~\cite{sorkin_entropy_1983,bombelli_quantum_1986,srednicki_entropy_1993,callan_geometric_1994,solodukhin_entanglement_2011} and in quantum many-body systems~\cite{eisert_area_2010,amico_entanglement_2008}.
The connection between entanglement and geometry has played a central role in modern developments in holography~\cite{ryu_holographic_2006,swingle_entanglement_2012,almheiri_entropy_2019,penington_entanglement_2020}, as well as proposals that spacetime itself may emerge from entanglement~\cite{vanraamsdonk_building_2010,jacobson_entanglement_2016,cao_space_2017}.

In addition to quantifying the entanglement, one can also identify the degrees of freedom of the field which are entangled across the bipartition.
Generally, when a pure Gaussian state (such as the vacuum) is divided into two subsystems, there exists a basis of local normal modes in both subsystems where each mode is either in a pure state or entangled with exactly one mode in the other subsystem, called its partner mode~\cite{boteroModeWiseEntanglementGaussian2003,boteroSpatialStructuresLocalization2004a,hackl_minimal_2019,agulloCorrelationEntanglementPartners2025}.
In the context of 1+1D Minkowski spacetime, this partner mode structure translates directly into the Unruh effect: when space (the real line) is partitioned into two half-lines, the normal modes on either side are the so-called Rindler modes~\cite{birrell_quantum_1982}.
Each Rindler mode is in a two-mode squeezed state with exactly one mode in the other half-line, thus the reduced state of a single mode in each pair is thermal.
The fact that Rindler modes correspond to plane waves in the natural coordinate frame of a uniformly accelerated observer has the effect that this observer perceives the Minkowski vacuum as a thermal state at the Unruh temperature $T_U = a/2\pi$.
Of course, eternal uniform acceleration is a theoretical idealization, which leads to the question to what extent observers in bounded regions of spacetime can access entanglement in the field.
Since the Unruh effect is closely related to the Hawking effect through the equivalence principle, addressing this question would also open a path towards a more refined analysis of how entanglement is distributed among localized field degrees of freedom in black hole spacetimes.

For entanglement between bounded spacelike-separated regions of a field, early work in algebraic QFT demonstrated maximal violation of Bell inequalities between adjacent regions in dilation-invariant theories (as well as in more general theories between complementary wedge-shaped regions)~\cite{summers_maximal_1987}.
Later, it was shown that the Reeh-Schlieder property of the vacuum state implies that there is distillable entanglement between any two spacelike-separated bounded regions~\cite{verch_distillability_2005}.
The Reeh-Schlieder property is a fundamental feature of the vacuum state, and this demonstrates that vacuum entanglement is intrinsic to the structure of relativistic QFT~\cite{redhead_ado_1995,witten_aps_2018}.
It has also been shown that vacuum entanglement can be extracted into spacelike-separated localized probes interacting with the field~\cite{valentiniNonlocalCorrelationsQuantum1991,reznikEntanglementVacuum2003,steegEntanglingPowerExpanding2009}, which has been called entanglement harvesting (see, e.g.,~\cite{martin-martinezEntanglementCurvedSpacetimes2014,caribeLensingVacuumEntanglement2023} and references therein).

Because the reduced state of the vacuum in two bounded regions is mixed, the von Neumann entropy of the regions (or combinations of entropies, such as the mutual information) does not measure entanglement, as it does not distinguish between classical and quantum correlations.
Instead, a commonly used measure is the logarithmic negativity, which quantifies the extent to which a state violates the positive partial transpose (PPT) criterion, and is an upper bound to the distillable entanglement~\cite{horodecki_quantum_2009,serafini_quantum_2023}.
Unlike the entanglement entropy, it does not exhibit ultraviolet divergences (except in the limit of adjacent intervals).
Logarithmic negativity has been studied in QFT using both analytical~\cite{calabrese_entanglement_2012,calabrese_entanglement_2013} and numerical methods~\cite{marcovitch_critical_2009,klco_entanglement_2021}, as well as with particular families of smearing functions and detector models~\cite{zych_entanglement_2010,agullo_multimode_2025}.
Most recently, the logarithmic negativity for a scalar chiral current in 1+1D was calculated using replica methods in~\cite{arias_entanglement_2026}.
In our work, we both corroborate this result by a direct calculation with phase space methods (where we find agreement in the leading order behavior in the regimes of large and small separation between the intervals), as well as completely characterize the structure of this entanglement by explicitly identifying pairs of field modes across the two intervals which carry the negativity.

In~\cite{klco_entanglement_2023,gao_partial_2024,gao_detecting_2025}, it was shown how phase space methods for calculating the logarithmic negativity of Gaussian states can be used to construct the pair of modes in the two intervals which are the most entangled.
Identifying these modes allows one to construct optimal detection profiles for accessing the negativity in these regions~\cite{gao_detecting_2025}.
Further, it was shown that for certain classes of mixed Gaussian states (including the scalar vacuum), one can simultaneously construct a set of pairs of modes between the intervals which carry all of the negativity, which were called the negativity cores of the system~\cite{klco_entanglement_2023,gao_partial_2024,gao_detecting_2025}.
These works performed numerical calculations of the negativity cores for lattice models.

Here we compute the negativity core decomposition analytically for a 1+1D massless scalar field, as well as construct closed-form solutions for the smearing profiles which can be used to couple to these modes.
Our analytical calculation is performed by adapting the method of~\cite{arias_entropy_2018}, which was used to compute the entropy of two disjoint intervals by formulating the calculation as a boundary value problem in the complex plane.
We cast our derivations in the framework presented in~\cite{hackl_bosonic_2021}, which is centered on the K\"ahler structure induced by Gaussian states on phase space.
This approach, which is developed in parallel for bosons and fermions, has proven fruitful from condensed matter physics and quantum optics, to statistical and black hole physics.
Our work highlights that its structure is also apt for the treatment of the infinite-dimensional phase space of quantum field theory.
Furthermore, as an extension to the framework, we propose a basis-independent characterization of the transpose operation.

While we give a generally accessible presentation of our results in the accompanying paper~\cite{pye_negativity_2026}, here we provide a complete explanation and derivation.
To facilitate our exposition, in~\cref{sec:gaussian} we include a succinct but self-contained review of the phase space formalism.
Readers already familiar with this formalism may want to briefly review this section to familiarize themselves with our notation, especially regarding the distinction between the phase space and co-phase space.
In~\cref{sec:transpose}, we provide a short proof of a basis-independent characterization of the transpose operation, which we will use for our problem.
We also review the relevant results of~\cite{klco_entanglement_2023,gao_partial_2024,gao_detecting_2025} in~\cref{sec:optimal}, which will allow us to easily identify the negativity cores from our computation performed in the framework of~\cite{hackl_bosonic_2021}.
The technical details of our main results are presented in~\cref{sec:qft_phase_space,sec:analytical}.
In~\cref{sec:qft_phase_space}, we construct the phase space of the 1+1D massless scalar field.
\cref{sec:analytical} sets up and solves the eigenvalue problem relevant for the calculation of the logarithmic negativity, as well as the construction of the smearing profiles which identify the negativity cores.
In~\cref{sec:numerical}, we develop a numerical model which we use to corroborate our analytical results.

\section{Review of phase space formalism for Gaussian states}
\label{sec:gaussian}

In this section, we provide a self-contained review of the key concepts of the phase space formalism for Gaussian states (for further details, see, for example, \cite{serafini_quantum_2023,hackl_bosonic_2021}).
We follow, in particular, the approach of~\cite{hackl_bosonic_2021} and carefully distinguish between phase space in which the state is represented through its covariance matrix, and the co-phase space in which the linear observables are represented.
This distinction will prove important in our later treatment of the quantum field case.
For simplicity, in this section we assume that the phase space is finite-dimensional, and defer a discussion of subtleties in infinite dimensions to our concrete case in \cref{sec:qft_phase_space}.

\subsection{Phase space, co-phase space, and symplectic forms}

We begin by defining the classical phase space as a $2N$-dimensional vector space $V$ over $\mathbb{R}$, with elements denoted $\xi^a \in V$.
Note that we employ a superscript as an abstract index to indicate that $\xi^a$ is a vector (it is not the index of the components of $\xi^a$ in some particular basis).
Elements in phase space can be identified with coordinates given by generalized position and momentum quadratures $(q_1,\dots,q_N,p_1,\dots,p_N)$.
We also define the co-phase space $V^\ast$ as the dual space of $V$, i.e., the set of linear functionals $V \to \mathbb{R}$.
Elements in this space are denoted by $f_a \in V^\ast$, and can be associated with linear observables $f_a \xi^a$ (where we use the summation convention).
Note that we use a subscript to denote an abstract covector index.

Quantization is achieved through a linear map $\hat{\xi}^a : V^\ast \to \mathcal{L}(\mathcal{H})$, given by $f_a \mapsto f_a \hat{\xi}^a$, which associates covectors with (unbounded) linear operators on the Hilbert space of the quantum system.
Concretely, this map can be represented by a vector of position and momentum operators, $\hat{\xi}^a \equiv (\hat{q}_1,\dots,\hat{q}_N,\hat{p}_1,\dots,\hat{p}_N)$.
The commutation relations are encoded in
\begin{align}
  i \Omega^{ab} \hat{1} := [ \hat{\xi}^a, \hat{\xi}^b ].
\end{align}
We can view the object $\Omega^{ab}$ as a bilinear form $\Omega^{ab} : V^\ast \times V^\ast \to \mathbb{R}$, which can be used to compute the commutator between any two linear observables.
Note that $\Omega^{ab}$ is antisymmetric and nondegenerate (by assumption).
The latter property implies that there is an induced isomorphism $\Omega^{ab} : V^\ast \to V$, given by $f_a \mapsto \Omega^{ab} f_b$.
We denote its inverse by $\omega_{ab} : V \to V^\ast$ (satisfying $\Omega^{ab} \omega_{bc} = \tensor{\delta}{^a_c}$ and $\omega_{ab} \Omega^{bc} = \tensor{\delta}{_a^c}$), which can also be viewed as a bilinear form $\omega_{ab} : V \times V \to \mathbb{R}$.

These bilinear forms can be represented concretely by matrices.
In a symplectic (canonical) basis, they take the form
\begin{align}
  \Omega^{ab} \equiv \Omega = \begin{bmatrix}
    0 & \mathbb{I}_N \\
    -\mathbb{I}_N & 0
  \end{bmatrix}
  \quad \text{and} \quad
  \omega_{ab} \equiv \omega = \begin{bmatrix}
    0 & -\mathbb{I}_N \\
    \mathbb{I}_N & 0
  \end{bmatrix},
  \label{eq:canonical_form}
\end{align}
where $\mathbb{I}_N$ denotes the $N \times N$ identity matrix.
For ease of notation, it is often useful to drop the indices and write index contractions in matrix language.
For example, $\xi^a \omega_{ab} \chi^b = \xi^T \omega \chi$, where $\chi$ is a column vector of the components of $\chi^b$ in some basis, and $\xi^T$ is a row vector of the components of $\xi^a$ in that basis.
Note that the transpose in $\xi^T$ does not indicate that it is a covector in $V^\ast$.
Similarly, one can write $f_a \Omega^{ab} g_b = f^T \Omega g$.

Symplectic matrices, $S \in \text{Sp}(2N)$, are a group of basis transformations which leave the canonical form \eqref{eq:canonical_form} of the matrix $\Omega$ invariant, $S \Omega S^T = \Omega$ (i.e., $S$ transforms between different symplectic bases).\footnote{Here we only consider passive symplectic transformations, which correspond to basis changes. In general, one also has active symplectic transformations which represent state changes due to Gaussian unitaries in the Hilbert space. However, we will not need these in this work. Also, note that our use of the terms ``passive'' and ``active'' is different than the sense used in quantum optics, which is related to photon number conservation.}
In this work, we will often find it helpful to represent the forms $\Omega^{ab}$ and $\omega_{ab}$ in non-symplectic bases.
Under a general change of basis represented by $M \in \text{GL}(2N)$, the components of a vector $\xi^a \in V$ transform as $\xi \mapsto M \xi$, and the components of a covector $f_a \in V^\ast$ transform as $f \mapsto M^{-T} f$ (so that $f_a \xi^a = f^T \xi$ remains invariant).
Similarly, under this basis change we have $\Omega \mapsto M \Omega M^T$ and $\omega \mapsto M^{-T} \omega M^{-1}$.

\subsection{Modes and subsystems}

A mode is a subsystem which cannot be decomposed into smaller subsystems.
It corresponds to a two-dimensional subspace of $V$ on which $\omega_{ab}$ is nondegenerate.
Often, such a subspace is identified with an element in the complexified phase space, $u^a \in V_\mathbb{C} := \mathbb{C} \otimes V$, which is normalized as $u^{\ast a} \omega_{ab} u^b = i$.
The pair $(u^a,u^{\ast a})$ are called mode functions, and the corresponding mode is the real two-dimensional subspace of $V$ spanned by their real and imaginary parts.
A basis of mode functions $\{ u_n^a, u_n^{\ast a} \}_n$ for $V_\mathbb{C}$ satisfying $u_n^{\ast a} \omega_{ab} u_{n'}^b = i \delta_{nn'}$ and $u_n^a \omega_{ab} u_{n'}^b = 0$ (symplectic orthonormality) yields a symplectic basis $\{ \sqrt{2} \, \Re(u_n^a), \sqrt{2} \, \Im(u_n^a) \}_n$ for $V$.
The elements of the dual basis for $V_\mathbb{C}^\ast$ are given by $v_{na} = i \omega_{ab} u_n^{\ast b}$, and satisfy $v_{na}^\ast \Omega^{ab} v_{n'b} = i \delta_{nn'}$ and $v_{na} \Omega^{ab} v_{n'b} = 0$.
These can be used to construct annihilation and creation operators, $\hat{a}_n := v_{na} \hat{\xi}^a$ and $\hat{a}_n^\dagger = v_{na}^\ast \hat{\xi}^a$, which satisfy $[ \hat{a}_n, \hat{a}_{n'}^\dagger ] = i v_{na} \Omega^{ab} v_{n'b}^\ast = \delta_{nn'}$.
The corresponding quadratures of these modes are $\hat{q}_n = \tfrac{1}{\sqrt{2}} ( \hat{a}_n + \hat{a}_n^\dagger )$ and $\hat{p}_n = \frac{(-i)}{\sqrt{2}} ( \hat{a}_n - \hat{a}_n^\dagger )$.
Using the resolution of identity,
\begin{align}
  \tensor{\delta}{^a_b} = \sum_n ( u_n^{\ast a} v_{nb}^\ast + u_n^a v_{nb} ),
  \label{eq:ROI}
\end{align}
we can construct a mode expansion of $\hat{\xi}^a$ in this basis as
\begin{align}
  \hat{\xi}^a = \sum_n ( \hat{a}_n^\dagger u_n^{\ast a} + \hat{a}_n u_n^a ). \label{eq:mode_expn}
\end{align}

More generally, we define a subsystem $A$ as a subspace $V_A \subset V$ on which the restricted symplectic form $(\omega_A)_{ab} : V_A \times V_A \to \mathbb{R}$ is nondegenerate.
The full phase space $V$ can be decomposed into a direct sum $V = V_A \oplus V_{A^\perp}$, where $V_{A^\perp} := \{ \chi^a \in V : \xi^a \omega_{ab} \chi^b = 0 \; \forall \xi^a \in V_A \}$ is the symplectic complement of $V_A$.
This implies $\omega = \omega_A \oplus \omega_{A^\perp}$.
When we speak of two (disjoint) subsystems $A$ and $B$ (such as spacelike-separated regions in a quantum field theory), we mean two subspaces, $V_A, V_B \subset V$, on each of which the restricted symplectic form is nondegenerate and $\omega_{AB} = \omega_A \oplus \omega_B$ on the joint subspace $V_{AB} = V_A \oplus V_B$.
The latter implies that observables in $A$ and $B$ commute with each other (since we then also have $\Omega_{AB} = \Omega_A \oplus \Omega_B$).

\subsection{Gaussian states}

A Gaussian state can be defined as a ground or thermal state of a quadratic Hamiltonian, that is 
\begin{align}
  \hat{\rho} = \frac1{Z_\beta} e^{-\beta\hat{H}}, \; \; \text{where} \; \; \hat{H} = \frac12 (\hat{\xi} - \xi_0)^a h_{ab} (\hat{\xi} - \xi_0)^b,
\end{align}
and $h_{ab} : V \times V \to \mathbb{R}$ is a bilinear, symmetric, and positive-definite form.
Pure Gaussian states correspond to the limiting case $\hat{\rho} = \lim_{\beta \to \infty} Z_{\beta}^{-1} e^{-\beta \hat{H}}$, and for mixed states we usually absorb $\beta$ into $\hat H$.
$\hat{H}$ is sometimes called the modular or entanglement Hamiltonian of the state \cite{casini_entanglement_2009,casini_lectures_2022}, which may be different from the Hamiltonian which governs the evolution of the system.

Gaussian states have the notable feature that they are completely determined by the mean\footnote{Note that the mean $\xi^a_0$ is exactly given by the linear shift in $\hat H$.} and the covariance matrix of the quadratures,
\begin{align}
  \xi_0^a &= \tr ( \hat{\xi}^a \hat{\rho} ), \\
  G^{ab} &= \tr ( \{ \hat{\xi}^a, \hat{\xi}^b \} \hat{\rho} ) - \xi_0^a \xi_0^b.
\end{align}
The covariance matrix, $G^{ab} : V^\ast \times V^\ast \to \mathbb{R}$, is a bilinear form on the co-phase space which is symmetric and positive-definite.
Therefore, $G^{ab}$ is an inner product on $V^\ast$.
Similarly, its inverse, $g_{ab} : V \times V \to \mathbb{R}$ (satisfying $G^{ab} g_{bc} = \tensor{\delta}{^a_c}$ and $g_{ab} G^{bc} = \tensor{\delta}{_a^c}$), is an inner product on $V$.
The positivity of the density operator $\hat{\rho} \geq 0$ implies the uncertainty relation
\begin{align}
  G^{ab} + i \Omega^{ab} \geq 0.
\end{align}

A simple example of a Gaussian state is a thermal state of a single harmonic oscillator with $\hat{H} = \frac{\lambda}{2} ( \hat{p}^2 + \hat{q}^2 ) = \lambda ( \hat{a}^\dagger \hat{a} + \frac12 )$, where $\hat{a} = \frac{1}{\sqrt{2}}( \hat{q} + i \hat{p} )$ and $\lambda$ will be referred to as the mode frequency.
The density operator is
\begin{align}
  \hat{\rho} = \frac{1}{Z} e^{-\lambda \hat{a}^\dagger \hat{a}}, \label{eq:single_mode_thermal}
\end{align}
with partition function $Z = (1 - e^{-\lambda})^{-1}$.
For this state, we have $\xi_0^a = 0$ and the covariance matrix is determined from the mean occupation number, $G^{ab} = (2 \overline{n} + 1) \delta^{ab}$, where $\overline{n} := \tr(\hat{a}^\dagger \hat{a} \hat{\rho}) = (e^{\lambda}-1)^{-1}$.

The specification of the density operator $\hat{\rho}$ from $\xi_0^a$ and $G^{ab}$ can be seen explicitly by expanding in terms of the orthonormal basis\footnote{I.e., orthonormal with respect to the Hilbert-Schmidt inner product.} of displacement operators, $\hat{D}_\xi = e^{-i \xi^a \omega_{ab} \hat{\xi}^b}$. 
An expansion of a general bounded operator, $\hat{O}$, in terms of displacement operators is given by the Fourier-Weyl relation~\cite{serafini_quantum_2023},
\begin{align}
  \hat{O} = \int \frac{d^{2N}\xi}{(2\pi)^N} \, \tr( \hat{D}_\xi^\dagger \hat{O} ) \, \hat{D}_\xi.
\end{align}
When applied to a density operator $\hat{\rho}$, the set of coefficients $\chi(\xi) := \tr( \hat{D}_\xi^\dagger \hat{\rho} )$ is called the characteristic function of the state.
For a Gaussian state, one can show that it takes the form
\begin{align}
  \chi(\xi) = e^{-\frac14 (\omega \xi)_a G^{ab} (\omega \xi)_b + i \xi^a \omega_{ab} \xi_0^b},
\end{align}
which we see depends only on $G^{ab}$ and $\xi_0^a$.
For example, for the state in \cref{eq:single_mode_thermal} we have $\chi(\xi) = e^{-\frac14 (2 \overline{n} + 1) \| \xi \|_2^2} = e^{-\frac14 (2 \overline{n} + 1) (q^2 + p^2)}$.

From the Fourier-Weyl relation, it is straightforward to obtain the reduced state of $\hat{\rho}$ on a subsystem $A$ by splitting $\omega_{ab}$ and $\xi^a$ according to $V = V_A \oplus V_{A^\perp}$.
By applying the partial trace over $A^\perp$ to the Fourier-Weyl expansion of $\hat{\rho}$, one can show that the reduced density operator on $A$ has a characteristic function which is also Gaussian.
Further, the reduced covariance matrix is simply the restriction of the full covariance matrix to the quadratures associated with modes of subsystem $A$.

In the following, we will assume that $\xi_0^a = 0$, which is sufficient for our case.

\subsection{Symplectic diagonalization}

A key decomposition of Gaussian states arises from the symplectic diagonalization of $G^{ab}$, which can be used to express any Gaussian state as a product of (thermal or pure) single mode states.
This decomposition is guaranteed by Williamson's theorem on the symplectic diagonalization of positive-definite matrices~\cite{serafini_quantum_2023}.
It is constructed through the diagonalization of the operator
\begin{align}
  \tensor{J}{^a_b} := - G^{ac} \omega_{cb} : V \to V.
\end{align}
$\tensor{J}{^a_b}$ is real and antisymmetric with respect to the inner product $g_{ab}$.
Hence, it can be diagonalized over the complexified phase space $V_\mathbb{C} := \mathbb{C} \otimes V$, with eigenvalues and eigenvectors coming in pairs $\tensor{J}{^a_b} u_n^{\ast b} = i \nu_n u_n^{\ast b}$ and $\tensor{J}{^a_b} u_n^b = - i \nu_n u_n^b$.
Further, the eigenvectors of $\tensor{J}{^a_b}$ can be normalized to form a symplectically-orthonormal basis of mode functions, with each mode defined by a pair of eigenvectors $(u_n^a,u_n^{\ast a})$ ~\cite{serafini_quantum_2023}.
Note that the corresponding dual basis vectors, $v_{na} = i \omega_{ab} u_n^{\ast b} \in V_\mathbb{C}^\ast$, are also left eigenvectors of $\tensor{J}{^a_b}$ (i.e., eigenvectors of $\tensor{(J^T)}{_a^b}$).
We can expand $\tensor{J}{^a_b}$ in its eigenbasis as
\begin{align}
  \tensor{J}{^a_b} = \sum_n i \nu_n ( u_n^{\ast a} v_{nb}^\ast - u_n^a v_{nb} ).
\end{align}
Similarly, since $G^{ab} = - \tensor{J}{^a_c} \Omega^{cb}$, we can write
\begin{align}
  G^{ab} = \sum_n \nu_n ( u_n^{\ast a} u_n^b + u_n^a u_n^{\ast b} ).
\end{align}
The set $\{ \nu_n \}_n$ are the symplectic eigenvalues of $G^{ab}$.
The uncertainty principle implies $\nu_n \geq 1$.

The significance of the particular modes which diagonalize $\tensor{J}{^a_b}$ is that they are uncorrelated, i.e., the covariance matrix splits into a direct sum over the subspaces of these modes.
Therefore, the characteristic function, hence $\hat{\rho}$ (through the Fourier-Weyl relation), splits into a product state of these modes.
The state of the individual modes can be shown to be a thermal state of the form \cref{eq:single_mode_thermal}, with mode frequencies $2 \, \text{arccoth}(\nu_n)$.
From this, one can easily calculate, e.g., the entropy of the state from the symplectic eigenvalues.

By comparing with the original form of the state $\hat{\rho} = Z^{-1} e^{-\hat{H}}$, we can also relate this to a decomposition of the modular Hamiltonian.
In general, the explicit relationship between $G^{ab}$ and $h_{ab}$ can be summarized by
\begin{align}
  i \tensor{K}{^a_b} = 2 \, \text{arccoth}(i \tensor{J}{^a_b}),
\end{align}
where $\tensor{K}{^a_b} := \Omega^{ac} h_{cb}$ \cite{hackl_bosonic_2021}.
The right-hand side should be understood as taking a function of the eigenvalues of $\tensor{J}{^a_b}$.
The operator $\tensor{K}{^a_b}$ is the generator of the modular flow, as the Heisenberg-picture evolution of $\hat{\xi}$ generated by $\hat{H}$ is $\frac{d}{d\tau} \hat{\xi}^a(\tau) = \tensor{K}{^a_b} \hat{\xi}^b(\tau)$.

\subsection{Logarithmic negativity}

In this work, we are interested in the entanglement between two subsystems, $A$ and $B$, when the joint state of $AB$ is mixed.
One simple necessary condition for the separability of a state is the positive partial transpose (PPT) criterion.
The idea is to apply the transpose operator to the state of one of the subsystems, $\hat{\rho}_{AB} \mapsto \hat{\rho}_{AB}^\Gamma := (1 \otimes \, ^T) \hat{\rho}_{AB}$.
Since transposition is positive, but not completely positive, $\hat{\rho}_{AB}^\Gamma$ may have negative eigenvalues (hence may not be a physical density operator).
However, this cannot happen if the state is separable (since $\sum_n p_n \, \hat{\rho}_{A,n} \otimes \hat{\rho}_{B,n}$ and $\sum_n \, p_n \hat{\rho}_{A,n} \otimes \hat{\rho}_{B,n}^T$ have the same spectrum).
Thus, $\hat{\rho}_{AB}$ being separable implies it has a PPT, which means that a violation of the PPT criterion is sufficient for $\hat{\rho}_{AB}$ to be entangled.
The converse is not true in general, but for certain systems PPT is known to be both sufficient and necessary for entanglement (e.g., qubit-qubit, qubit-qutrit, $1$ vs. $n$ mode Gaussian states, etc. \cite{horodecki_quantum_2009,serafini_quantum_2023}).

One of the primary tasks of this paper is to compute the logarithmic negativity, which measures the extent to which a partially-transposed (PT) density operator, $\hat{\rho}_{AB}^\Gamma$, violates the PPT criterion.
It is an entanglement monotone (i.e., cannot increase under local operations and classical communication), and is an upper bound for the distillable entanglement \cite{horodecki_quantum_2009,serafini_quantum_2023}.
The logarithmic negativity is given by
\begin{align}
  E_\mathcal{N} = \log_2 \| \hat{\rho}^\Gamma \|_1,
\end{align}
where $\| \hat{\rho}^\Gamma \|_1$ is the sum of the absolute value of the eigenvalues of $\hat{\rho}^\Gamma$.
If $\hat{\rho}^\Gamma$ is PPT, then the eigenvalues are all positive and sum to $1$ (since partial transpose does not change the trace of an operator), hence $E_\mathcal{N} = 0$.
Therefore, $E_\mathcal{N} > 0$ implies the state $\hat{\rho}_{AB}$ is entangled.

Computing the logarithmic negativity of a Gaussian state can be done using a decomposition of $\hat{\rho}_{AB}^\Gamma$ similar to that we reviewed above.
We can apply the partial transpose to $\hat{\rho}_{AB}$ in a basis consisting of the tensor product of Fock bases for $A$ and $B$.
One can then show that this simply modifies the characteristic function by mapping
\begin{align}\label{eq:G_pt}
  G^{ab} \mapsto (G^\Gamma)^{ab} := \tensor{\Gamma}{^a_c} G^{cd} \tensor{(\Gamma^T)}{_d^b},
\end{align}
(also $\xi_0^a \mapsto \tensor{\Gamma}{^a_b} \xi_0^b$),
where $\tensor{\Gamma}{^a_b} : V \to V$ is a linear map which acts by negating the momentum quadratures in $B$ (i.e., $p_{B,n} \mapsto -p_{B,n}$) \cite{serafini_quantum_2023}.
The PT covariance matrix $(G^\Gamma)^{ab} : V^\ast \times V^\ast \to \mathbb{R}$ remains positive-definite, hence can be symplectically-diagonalized by diagonalizing the operator,
\begin{align}
  \tensor{(J^\Gamma)}{^a_b} := - (G^\Gamma)^{ac} \omega_{cb} : V \to V.
  \label{eq:JPT_defn}
\end{align}
The mode decompositions of these operators are of the same general form as the non-PT case, with PT symplectic eigenvalues $\{ \tilde{\nu}_n \}_n$, symplectically-orthonormal eigenvectors $\{ \tilde{u}_n^a, \tilde{u}_n^{\ast a} \}_n$, and dual basis $\{ \tilde{v}_{na}, \tilde{v}_{na}^\ast \}_n$.
The difference with the non-PT case is that one can have PT symplectic eigenvalues with $\tilde{\nu}_n < 1$, due to the fact that the uncertainty principle does not generally hold for $(G^\Gamma)^{ab}$ (since $\hat{\rho}_{AB}^\Gamma$ may not be positive).
It can be shown that the number of PT symplectic eigenvalues with $\tilde{\nu}_n < 1$ cannot be greater than the number of modes in the smaller of the two subsystems $A$ and $B$ \cite{serafini_quantum_2023}.

The modes which violate the uncertainty principle are exactly those which contribute to the negativity.
Using the Fourier-Weyl relation, one can similarly show $\hat{\rho}_{AB}^\Gamma$ factorizes into a product state over these modes.
The states of the individual modes with $\tilde{\nu}_n \geq 1$ are physical thermal states, as before.
Those with $\tilde{\nu}_n < 1$ have the form of a thermal state, but with oscillator frequencies $2 \, \text{arccoth}(\tilde{\nu}_n)$, which have an imaginary component $\pm i \pi$ and yield negative eigenvalues for the state.
From this, it is straightforward to show that
\begin{align}
  E_\mathcal{N} = \sum_{\tilde{\nu}_n < 1} ( - \log_2 \tilde{\nu}_n ).
\end{align}
Therefore, computing the logarithmic negativity reduces to the problem of finding the eigenvalues of $J^\Gamma$ with $\tilde{\nu}_n < 1$.
In \cref{sec:optimal}, we also review how one can use the eigenvectors of $J^\Gamma$ to construct pairs of modes in the subsystems $A$ and $B$ which contain the logarithmic negativity associated with each of the symplectic eigenvalues $\tilde{\nu}_n < 1$.

From the decomposition of $\hat{\rho}_{AB}^\Gamma$, one can also conclude that it can be written in the form $\hat{\rho}_{AB}^\Gamma = \tilde{Z}^{-1} e^{-\hat{H}^\Gamma}$ with $\hat{H}^\Gamma = \tfrac12 \hat{\xi}^a h_{ab}^\Gamma \hat{\xi}^b$, where $\hat{H}^\Gamma$ is called the negativity Hamiltonian \cite{murciano_negativity_2022}.
It is similarly related to the PT covariance matrix through
\begin{align}
  i \tensor{(K^\Gamma)}{^a_b} = 2 \, \text{arccoth}( i \tensor{(J^\Gamma)}{^a_b} ),
  \label{eq:negativity_hamiltonian}
\end{align}
where $\tensor{(K^\Gamma)}{^a_b} := \Omega^{ac} h_{cb}^\Gamma$.

\section{Basis-independent transpose operation}
\label{sec:transpose}

Let $A$ and $B$ be two subsystems described by two subspaces of a symplectic vector space with $\Omega_{AB} = \Omega_A \oplus \Omega_B$.
A partial transpose is a map on phase space of the form $\Gamma = 1_A \oplus T_B$, where $T_B$ is a transpose operator on subsystem $B$.
Typically, this is defined in a canonical basis by negating the momentum variables $p_B \to -p_B$.
Note that there is not a unique definition of (partial) transpose since we could have used any canonical basis.

As we will often be using non-canonical bases in this paper, we will find it convenient to develop a characterization of the transpose operation which does not require to first transform to a canonical basis.
We will show that the above definition of a transpose is equivalent to $T^2 = 1$ and $T \Omega T^T = -\Omega$ (or equivalently $T^T \omega T = -\omega$, where $\omega = \Omega^{-1}$).

\begin{prop}
Let $V$ be a $2N$-dimensional real vector space, and $\omega : V \times V \to \mathbb{R}$ a nondegenerate, antisymmetric, bilinear form.
A linear map $T : V \to V$ satisfies $T^2 = 1$ and $T^T \omega T = -\omega$ if and only if there is a basis for $V$ in which
\begin{align}
  T \equiv \begin{bmatrix}
    \mathbb{I}_N & 0 \\
    0 & -\mathbb{I}_N
  \end{bmatrix}
  \quad \text{and} \quad
  \omega \equiv \begin{bmatrix}
    0 & -\mathbb{I}_N \\
    \mathbb{I}_N & 0
  \end{bmatrix},
\end{align}
where $\mathbb{I}_N$ in these block matrices denotes the $N \times N$ identity operator.
\end{prop}

\begin{proof}
The backward direction is trivial.
For the forward direction, note that $T^2 = 1$ implies that $T$ is diagonalizable with eigenvalues $\pm 1$.
Let $u_\pm$ and $v_\pm$ be any eigenvectors of $T$ with eigenvalue $\pm 1$.
Then
\begin{align}
  u_\pm^T \omega v_\pm = u_\pm^T T^T \omega T v_\pm = - u_\pm^T \omega v_\pm \; \implies \; u_\pm^T \omega v_\pm = 0.
\end{align}
Therefore, in a basis where $T$ is diagonal, we can write
\begin{align}
  T \equiv \begin{bmatrix}
    \mathbb{I}_N & 0 \\
    0 & -\mathbb{I}_N
  \end{bmatrix}
  \quad \text{and} \quad
  \omega \equiv \begin{bmatrix}
    0 &  \omega_{+-} \\
    -\omega_{+-}^T & 0
  \end{bmatrix},
\end{align}
where $\omega_{+-}$ is a matrix representing $\omega$ where the first input is restricted to the positive eigenspace of $T$, and the second to the negative eigenspace.
We also have $\omega_{-+} = -\omega_{+-}^T$ since the matrix representing $\omega$ must be antisymmetric.
Further, since $\omega$ is assumed nondegenerate, then we also have $\omega_{+-}$ invertible.
Hence, $\omega_{+-}$ is also a square matrix, which means that there are an equal number of positive and negative eigenvalues of $T$.
Now we apply the following change of basis to $T$ and $\omega$,
\begin{align}
  M \equiv \begin{bmatrix}
    \mathbb{I}_N & 0 \\
    0 & -\omega_{+-}
  \end{bmatrix}.
\end{align}
The linear map $T$ is invariant under this transformation, because it transforms as $T \mapsto M T M^{-1}$, which is equal to $M T M^{-1} = T$.
The symplectic form transforms as $\omega \mapsto M^{-T} \omega M^{-1}$, hence in this basis we have
\begin{align}
  \omega \equiv \begin{bmatrix}
    0 & -\mathbb{I}_N \\
    \mathbb{I}_N & 0
  \end{bmatrix}.
\end{align}
\end{proof}

\section{Local mode pairs from $J^\Gamma$ eigenfunctions}
\label{sec:optimal}

The eigenvalues of the operator $J^\Gamma$ (defined in \cref{eq:JPT_defn}) can be used to compute the logarithmic negativity between two subsystems $A$ and $B$.
The question as to how much of this negativity can be accessed 
locally was answered in \cite{klco_entanglement_2023,gao_partial_2024,gao_detecting_2025}.
In these works, it was demonstrated how the eigenvectors of $J^\Gamma$ can be used to identify a set of pairs of local modes in $A$ and $B$, such that the logarithmic negativity of each pair of modes corresponds to one of the eigenvalues of $J^\Gamma$.
Thus, in principle one could extract all of the negativity of the $AB$ system by successively swapping out each of these modes.

Here we will review and rederive some of the results of \cite{klco_entanglement_2023,gao_partial_2024,gao_detecting_2025} in the formalism discussed in \cref{sec:gaussian}.
This will allow us to easily identify these modes from our diagonalization of $J^\Gamma$ for the massless scalar field in \cref{sec:analytical}.
First, we show that the logarithmic negativity contained in any single pair of modes in $A$ and $B$ cannot exceed that of the smallest symplectic eigenvalue of $G^\Gamma$.
We then show how to construct a pair of modes in $A$ and $B$ from the corresponding eigenvector of $J^\Gamma$ which has this maximal negativity.
Finally, we demonstrate that, under a certain assumption, one can use the set of eigenvectors of $J^\Gamma$ to identify a simultaneous set of $AB$ mode pairs whose negativities one-by-one correspond to the eigenvalues of $J^\Gamma$.
In \cite{gao_partial_2024}, it was argued that the particular assumption required for this latter construction holds for the scalar field vacuum, and we verify this explicitly for our case in \cref{sec:smearings}.

\subsection{Optimality}

In \cref{sec:gaussian}, we associated subsystems $A$ and $B$ with subspaces $V_A, V_B \subset V$ on each of which $\omega$ is nondegenerate and decomposes as $\omega_A \oplus \omega_B$ on $V_A \oplus V_B$.
Let $G$ be the covariance matrix of a Gaussian state on the joint system $AB$, and let $\Gamma$ represent a partial transpose, which can be written as $\Gamma = 1_A \oplus T_B$.
Recall that $G^\Gamma := \Gamma G \Gamma^T$, can be expanded in terms of its symplectic eigenvectors as
\begin{align}
  G^\Gamma = \sum_n \tilde{\nu}_n ( \tilde{u}_n^\ast \tilde{u}_n^T + \tilde{u}_n \tilde{u}_n^{\ast T} ),
\end{align}
where the $\tilde{u}_n$'s form a set of mode functions with $\tilde{u}_n^{\ast T} \omega \tilde{u}_{n'} = i \delta_{nn'}$ and $\tilde{u}_n^T \omega \tilde{u}_{n'} = 0$.
We can then reconstruct the covariance matrix using $G = \Gamma G^\Gamma \Gamma^T$, since $\Gamma^2 = 1$.

A single mode corresponds to a two-dimensional subspace on which $\omega$ is nondegenerate.
Given any choice of single mode in $A$ and single mode in $B$, the reduced covariance matrix of this four-dimensional subsystem, $S$, can be obtained using a matrix of the form $P_S = P_{S,A} \oplus P_{S,B}$, which extracts the $4 \times 4$ submatrix of correlations in the $AB$ subsystem as $G_S = P_S G P_S^T$.
To compute the logarithmic negativity between $A$ and $B$ in this reduced system, we introduce a partial transpose $\Gamma_S = 1_{S,A} \oplus T_{S,B}$ satisfying $\Gamma_S^2 = 1_S$ (where $1_S$ is the $4 \times 4$ identity matrix) and $T_{S,B} P_{S,B} \Omega P_{S,B}^T T_{S,B}^T = - P_{S,B} \Omega P_{S,B}^T$.
The logarithmic negativity is then obtained from the two symplectic eigenvalues $\tilde{\nu}_{S,1}$ and $\tilde{\nu}_{S,2}$ of $G_S^{\Gamma_S} = \Gamma_S G_S \Gamma_S^T$.
Recall that since $G_S$ is the covariance matrix of two modes, then only one of the symplectic eigenvalues will contribute to the logarithmic negativity.
Let us assume this eigenvalue is $\tilde{\nu}_{S,1}$.
We will demonstrate that $\tilde{\nu}_{S,1} \geq \tilde{\nu}_1$, where $\tilde{\nu}_1$ smallest symplectic eigenvalue of the full PT covariance matrix $G^\Gamma$.
This implies that the logarithmic negativity of any pair of modes in $A$ and $B$ cannot exceed $-\log_2 \tilde{\nu}_1$.

After the symplectic diagonalization of $G_S^{\Gamma_S}$, we can expand it in terms of its normalized mode functions as
\begin{align}
  G_S^{\Gamma_S} = \sum_{n=1}^2 \tilde{\nu}_{S,n} ( \tilde{u}_{S,n}^\ast \tilde{u}_{S,n}^T + \tilde{u}_{S,n} \tilde{u}_{S,n}^{\ast T}).
\end{align}
With the corresponding dual vectors $\tilde{v}_{S,n}$, satisfying $\tilde{v}_{S,n}^{\ast T} \Omega \tilde{v}_{S,n'} = i \delta_{nn'}$, $\tilde{v}_{S,n}^T \tilde{u}_{S,n'} = \delta_{nn'}$, and $\tilde{v}_{S,n}^{\ast T} \tilde{u}_{S,n'} = 0$, we can isolate the smallest symplectic eigenvalue of $G_S^{\Gamma_S}$ by
\begin{align}
  \tilde{\nu}_{S,1} &= \tilde{v}_{S,1}^{\ast T} G_S^{\Gamma_S} \tilde{v}_{S,1} = \tilde{v}_{S,1}^{\ast T} (\Gamma_S P_S \Gamma) G^\Gamma (\Gamma_S P_S \Gamma)^T \tilde{v}_{S,1}.
\end{align}
We note that $(\Gamma_S P_S \Gamma) \Omega (\Gamma_S P_S \Gamma)^T = P_S \Omega P_S^T$, which can be shown by decomposing $\Gamma_S P_S \Gamma = P_{S,A} \oplus T_{S,B} P_{S,B} T_B$ and using $T_{S,B} P_{S,B} \Omega P_{S,B}^T T_{S,B}^T = - P_{S,B} \Omega P_{S,B}^T$ and $T_B \Omega_B T_B^T = -\Omega_B$.
This implies that $\tilde{v}_{S,1}' := (\Gamma_S P_S \Gamma)^T \tilde{v}_{S,1}$ is normalized with respect to $\Omega$, i.e., $\tilde{v}_{S,1}'^{\ast T} \Omega \tilde{v}_{S,1}' = i$.
From the completeness of the set of symplectic eigenvectors of the full PT covariance matrix $G^\Gamma$, we have the resolution of identity
\begin{align}
  1_V = \sum_n ( \tilde{u}_n^\ast \tilde{v}_n^{\ast T} + \tilde{u}_n \tilde{v}_n^T ).
\end{align}
By multiplying on the right by $\Omega$, we also have the expansion
\begin{align}
  \Omega = \sum_n i ( \tilde{u}_n^\ast \tilde{u}_n^T - \tilde{u}_n \tilde{u}_n^{\ast T} ).
\end{align}
From $\tilde{v}_{S,1}'^{\ast T} \Omega \tilde{v}_{S,1}' = i$, we then have
\begin{align}
  \sum_n | \tilde{v}_{S,1}'^T \tilde{u}_n |^2 = 1 + \sum_n | \tilde{v}_{S,1}'^T \tilde{u}_n^\ast |^2.
\end{align}
Therefore,
\begin{align}
  \tilde{\nu}_{S,1} &= \tilde{v}_{S,1}^{\ast T} G_S^{\Gamma_S} \tilde{v}_{S,1} \nonumber\\
  &= \tilde{v}_{S,1}'^{\ast T} G^\Gamma \tilde{v}_{S,1}' \nonumber\\
  &= \sum_n \tilde{\nu}_n ( | \tilde{v}_{S,1}'^T \tilde{u}_n |^2 + | \tilde{v}_{S,1}'^T \tilde{u}_n^\ast |^2 ) \nonumber\\
  &\geq \tilde{\nu}_1 \sum_n ( | \tilde{v}_{S,1}'^T \tilde{u}_n |^2 + | \tilde{v}_{S,1}'^T \tilde{u}_n^\ast |^2 ) \nonumber\\
  &= \tilde{\nu}_1 ( 1 + 2 \sum_n | \tilde{v}_{S,1}'^T \tilde{u}_n^\ast |^2 ) \nonumber\\
  &\geq \tilde{\nu}_1,
\end{align}
which completes the claim.
Note that one can also extend this argument to arbitrary subsystems of $A$ and $B$, i.e., the smallest PT symplectic eigenvalue of any subsystem cannot be smaller than that of the full PT covariance matrix.

\subsection{Achievability}

We saw that any single pair of modes in $A$ and $B$ cannot have a PT symplectic eigenvalue smaller than that of the full $AB$ system, but can this lower bound be saturated?
Here we will demonstrate how a pair of modes can be constructed from the symplectic eigenvectors of $G^\Gamma$ which achieves this bound.

From the expansion $G^\Gamma = \sum_n \tilde{\nu}_n ( \tilde{u}_n^\ast \tilde{u}_n^T + \tilde{u}_n \tilde{u}_n^{\ast T} )$, we see that we also obtain an expansion for the covariance matrix $G = \Gamma G^\Gamma \Gamma^T = \sum_n \tilde{\nu}_n ( u_n^\ast u_n^T + u_n u_n^{\ast T} )$, where $u_n := \Gamma \tilde{u}_n$.
Note that $u_n$ does not define a normalized mode because $\Gamma$ is not symplectic.
We denote by $\tilde{\nu}_1$ the minimum of the set $\{ \tilde{\nu}_n \}_n$.
The expansion of $G$ suggests that one should attempt to find a pair of modes in $A$ and $B$ which isolates the $n=1$ term in the sum.
Therefore, we need to find (a set of) normalized modes which span a subspace containing $(\tilde{u}_1, \tilde{u}_1^\ast)$ after application of the partial transpose.

The choice of pair of local modes is unique (up to single-mode transformations), and can be constructed as follows.
Recall, we assume the phase space of a joint system $AB$ can be decomposed as $V_{AB} = V_A \oplus V_B$, thus we can write $\tilde{u}_1 = \tilde{u}_{1A} + \tilde{u}_{1B}$, where $\tilde{u}_{1A} \in V_A$ and $\tilde{u}_{1B} \in V_B$.
Since $u_1 = \Gamma \tilde{u}_1 = \tilde{u}_{1A} + T \tilde{u}_{1B}$, let us consider the subspace spanned by $(\tilde{u}_{1A},\tilde{u}_{1A}^\ast)$ and $(T\tilde{u}_{1B},T\tilde{u}_{1B}^\ast)$.
Assuming that
\begin{align}
  i \alpha_1 := \tilde{u}_{1A}^{\ast T} \omega \tilde{u}_{1A} \neq 0 \quad \text{and} \quad i \beta_1 := \tilde{u}_{1B}^{\ast T} \omega \tilde{u}_{1B} \neq 0,
\end{align}
(which we verify in \cref{sec:smearings} for our case) we can normalize these to obtain a pair of modes in $A$ and $B$.
Note that because $i = \tilde{u}_1^{\ast T} \omega \tilde{u}_1$, then $\alpha_1 + \beta_1 = 1$.
More explicitly, if $\alpha_1 > 0$ we define $u_{1A} := \frac{1}{\sqrt{\alpha_1}} \tilde{u}_{1A}$, and if $\alpha_1 < 0$ we define $u_{1A} := \frac{1}{\sqrt{-\alpha_1}} \tilde{u}_{1A}^\ast$.
Similarly, if $\beta_1 > 0$ we write $u_{1B} := \frac{1}{\sqrt{\beta_1}} T \tilde{u}_{1B}^\ast$, and if $\beta_1 < 0$ we write $u_{1B} := \frac{1}{\sqrt{-\beta_1}} T \tilde{u}_{1B}$.
One can verify that these give a pair of normalized modes $(u_{1A},u_{1A}^\ast)$ in $A$ and $(u_{1B},u_{1B}^\ast)$ in $B$.
This choice of mode pair is unique (up to single-mode symplectic transformations), since any other two-dimensional subspace corresponding to a mode in $A$ would need to contain $(\tilde{u}_{1A},\tilde{u}_{1A}^\ast)$, and similar for $B$. 

Using the corresponding dual vectors $v_{1A} = i \omega u_{1A}^\ast$ and $v_{1B} = i \omega u_{1B}^\ast$, we construct the projection
\begin{align}
  P_1 = u_{1A} v_{1A}^T + u_{1A}^\ast v_{1A}^{\ast T} + u_{1B} v_{1B}^T + u_{1B}^\ast v_{1B}^{\ast T}.
\end{align}
The reduced covariance matrix on this pair of modes is then $G_1 = P_1 G P_1^T$.
After applying partial transposition, we have $G_1^\Gamma = \Gamma G_1 \Gamma^T = P_1^\Gamma G^\Gamma P_1^\Gamma$, where $P_1^\Gamma := \Gamma P_1 \Gamma$ is also a projection.
One can verify that
\begin{align}
  P_1^\Gamma = \tilde{u}_{1A} \tilde{v}_{1A}^T + \tilde{u}_{1A}^\ast \tilde{v}_{1A}^{\ast T} + \tilde{u}_{1B} \tilde{v}_{1B}^T + \tilde{u}_{1B}^\ast \tilde{v}_{1B}^{\ast T},
\end{align}
where $\tilde{v}_{1A} := \frac{i}{\alpha_1} \omega \tilde{u}_{1A}^\ast$ and $\tilde{v}_{1B} := \frac{i}{\beta_1} \omega \tilde{u}_{1B}^\ast$.
We see that this achieves the aim of finding a pair of modes which contains $\tilde{u}_1$ in their span after applying the partial transpose, since we have $P_1^\Gamma \tilde{u}_1 = \tilde{u}_1$.

Let us verify that $\tilde{\nu}_1$ remains a symplectic eigenvalue of $G_1^\Gamma$.
In the subspace spanned by $\{ \tilde{u}_{1A}, \tilde{u}_{1A}^\ast, \tilde{u}_{1B}, \tilde{u}_{1B}^\ast \}$, we have the mode corresponding to $\tilde{u}_1 = \tilde{u}_{1A} + \tilde{u}_{1B}$.
We can find a second mode in this subspace as follows.
The dual vector of $\tilde{u}_1$ is $\tilde{v}_1 = i \omega \tilde{u}_1^\ast = \alpha_1 \tilde{v}_{1A} + \beta_1 \tilde{v}_{1B}$.
We want to find a second normalized mode, $\tilde{u}_1'$, in this four-dimensional subspace such that $\tilde{v}_1^T \tilde{u}_1' = 0$ and $\tilde{v}_1^{\ast T} \tilde{u}_1' = 0$.
If $\alpha_1 \beta_1 > 0$ we use $\tilde{u}_1' := \frac{1}{\sqrt{\alpha_1 \beta_1}} ( \beta_1 \tilde{u}_{1A} - \alpha_1 \tilde{u}_{1B} )$, and if $\alpha_1 \beta_1 < 0$ we use $\tilde{u}_1' := \frac{1}{\sqrt{-\alpha_1 \beta_1}} ( \beta_1 \tilde{u}_{1A}^\ast - \alpha_1 \tilde{u}_{1B}^\ast )$.
We also have the corresponding dual vector, $\tilde{v}_1' = i \omega \tilde{u}_1'^\ast$, for which we have $\tilde{v}_1'^T \tilde{u}_1 = \tilde{v}_1'^T \tilde{u}_1^\ast = 0$.
By construction, we then have a pair of modes spanning this subspace which can be used to express $P_1^\Gamma$ as
\begin{align}
  P_1^\Gamma = \tilde{u}_1 \tilde{v}_1^T + \tilde{u}_1^\ast \tilde{v}_1^{\ast T} + \tilde{u}_1' \tilde{v}_1'^T + \tilde{u}_1'^\ast \tilde{v}_1'^{\ast T}.
\end{align}
Due to the symplectic orthonormality of the symplectic eigenvectors of the full PT covariance matrix $G^\Gamma$, and $\tilde{v}_1'^T \tilde{u}_1 = \tilde{v}_1'^T \tilde{u}_1^\ast = 0$, we then see that when applying $P_1^\Gamma$ to the expansion for $G^\Gamma$, we have
\begin{align}
  G_1^\Gamma &= \tilde{\nu}_1 ( \tilde{u}_1^\ast \tilde{u}_1^T + \tilde{u}_1 \tilde{u}_1^{\ast T} ) \nonumber \\
  &\qquad + (\tilde{u}_1' \tilde{v}_1'^T + \tilde{u}_1'^\ast \tilde{v}_1'^{\ast T}) G^\Gamma (\tilde{v}_1' \tilde{u}_1'^T + \tilde{v}_1'^\ast \tilde{u}_1'^{\ast T}).
\end{align}
Hence, we see that $G_1^\Gamma$ splits into a direct sum over the two modes $(\tilde{u}_1,\tilde{u}_1^\ast)$ and $(\tilde{u}_1',\tilde{u}_1'^\ast)$, i.e., they are uncorrelated.
Therefore, $\tilde{\nu}_1$ is indeed a symplectic eigenvalue of $G_1^\Gamma$.
Note that, since $G_1$ is the covariance matrix for a single pair of modes, the second symplectic eigenvalue from the mode $(\tilde{u}_1',\tilde{u}_1'^\ast)$ must be larger than one, hence does not contribute to the logarithmic negativity.

To summarize, we have shown that the pair of modes $(u_{1A},u_{1A}^\ast)$ in $A$ and $(u_{1B},u_{1B}^\ast)$ in $B$, which are constructed by normalizing the components of $u_1 = \Gamma \tilde{u}_1 = \tilde{u}_{1A} + T \tilde{u}_{1B}$, achieve the optimal bound for the logarithmic negativity from the previous section.
Further, this pair of modes is unique, up to single-mode transformations.

\subsection{Accessing the total negativity}

The total logarithmic negativity between $A$ and $B$ is determined from the collection of PT modes $(\tilde{u}_n,\tilde{u}_n^\ast)$ with $\tilde{\nu}_n < 1$.
Here we will show that (under the assumption of local symplectic orthogonality) we can repeat the above procedure to construct a set of pairs of modes in $A$ and $B$, each pair with logarithmic negativity $-\log_2 \tilde{\nu}_n$.

Similar to before, for each $n$ with $\tilde{\nu}_n < 1$, we decompose $\tilde{u}_n = \tilde{u}_{nA} + \tilde{u}_{nB}$.
This suggests restricting to the subspaces $\text{span} \{ \tilde{u}_{nA}, \tilde{u}_{nA}^\ast \}_n$ in $A$ and $\text{span} \{ T \tilde{u}_{nB}, T \tilde{u}_{nB}^\ast \}_n$ in $B$.
Note that it is not guaranteed that we have $\tilde{u}_{nA}^{\ast T} \omega \tilde{u}_{n'A} \sim i \delta_{nn'}$, $\tilde{u}_{nB}^{\ast T} \omega \tilde{u}_{n'B} \sim i \delta_{nn'}$, etc., from  $\tilde{u}_n^{\ast T} \omega \tilde{u}_{n'} = i \delta_{nn'}$ alone.
\footnote{In the Hilbert space, such a symplectic non-orthogonality would correspond to operators which do not commute, hence they would not be associated with separate modes of the subsystems $A$ and $B$.}
However, it was argued in \cite{klco_entanglement_2023,gao_partial_2024} that the vacuum of the scalar field has special structure which guarantees that one does indeed have local symplectic orthogonality, i.e.,
\begin{align}
  \tilde{u}_{nA}^{\ast T} \omega \tilde{u}_{n'A} = i \alpha_n \delta_{nn'}, \quad & \quad \tilde{u}_{nB}^{\ast T} \omega \tilde{u}_{n'B} = i \beta_n \delta_{nn'}, \nonumber\\
  \tilde{u}_{nA}^{T} \omega \tilde{u}_{n'A} = 0, \qquad \; & \; \qquad \tilde{u}_{nB}^{T} \omega \tilde{u}_{n'B} = 0,
\end{align}
for some $\alpha_n, \beta_n \neq 0$ (with $\alpha_n + \beta_n = 1$).
(We verify this explicitly in our case in \cref{sec:smearings}, where we find $\alpha_n = \beta_n = \tfrac12$ $\forall n$.)
This implies that we can construct a set of pairs of modes in $A$ and $B$, each of which corresponds to one of the $\tilde{\nu}_n < 1$.
In \cite{klco_entanglement_2023,gao_partial_2024}, these pairs were referred to as the ``negativity cores'' of the system $AB$.
One can decompose the entire subsystem $A$ (similarly $B$) into these core modes along with the symplectic complement of the core subspace in $A$, i.e., the subspace $\{ u \in V_A : \tilde{u}_{nA}^T \omega u = \tilde{u}_{nA}^{\ast T} \omega u = 0 \, \forall \, n \}$.\footnote{One can construct a basis for this subspace using a symplectic analogue of the Gram-Schmidt process (e.g., as done in \cite{klco_entanglement_2023}).}
In \cite{klco_entanglement_2023,gao_partial_2024}, these complementary subspaces of $A$ and $B$ were referred to as the ``halo.''
This structure implies that the state space factorizes into a tensor product of Hilbert spaces associated with each of the core modes, along with the halo.

We can follow a similar procedure as the previous section to explicitly construct the core modes, as well as show that the restriction to these core modes has the same logarithmic negativity as the full $AB$ system.
First we write the normalized modes as
\begin{align}
  u_{nA}& := \begin{cases}
    \frac{1}{\sqrt{\alpha_n}} \tilde{u}_{nA}, & \text{if $\alpha_n > 0$} \\
    \frac{1}{\sqrt{-\alpha_n}} \tilde{u}_{nA}^\ast, & \text{if $\alpha_n < 0$}
  \end{cases}\,,
  \nonumber\\
  u_{nB} &:= \begin{cases}
    \frac{1}{\sqrt{\beta_n}} T \tilde{u}_{nB}^\ast, & \text{if $\beta_n > 0$} \\
    \frac{1}{\sqrt{-\beta_n}} T \tilde{u}_{nB}, & \text{if $\beta_n < 0$}
  \end{cases}\,,
\end{align}
as well as the corresponding covectors $v_{nA} = i \omega u_{nA}^\ast$ and $v_{nB} = i \omega u_{nB}^\ast$.
Using these, we can construct the projection onto the core subspace as $P_C = \sum_n ( u_{nA} v_{nA}^T + u_{nB} v_{nB}^T + \text{c.c.} )$.
The reduced covariance matrix on this subspace is then $G_C = P_C G P_C^T$, and after partial transposition, we have $G_C^\Gamma = P_C^\Gamma G^\Gamma (P_C^\Gamma)^T$, with
\begin{align}
  P_C^\Gamma = \sum_n ( \tilde{u}_{nA} \tilde{v}_{nA}^T + \tilde{u}_{nB} \tilde{v}_{nB}^T + \text{c.c.} ),
\end{align}
where $\tilde{v}_{nA} = \frac{i}{\alpha_n} \omega \tilde{u}_{nA}$ and $\tilde{v}_{nB} = \frac{i}{\beta_n} \omega \tilde{u}_{nB}$.
In each of the subspaces indexed by $n$, we can change bases to the pair of normalized modes $(\tilde{u}_n,\tilde{u}_n^\ast)$ and $(\tilde{u}_n',\tilde{u}_n'^\ast)$, where
\begin{align}
  \tilde{u}_n' := \begin{cases}
    \frac{1}{\sqrt{\alpha_n \beta_n}} ( \beta_n \tilde{u}_{nA} - \alpha_n \tilde{u}_{nB} ), & \text{if $\alpha_n \beta_n > 0$} \\
    \frac{1}{\sqrt{-\alpha_n \beta_n}} ( \beta_n \tilde{u}_{nA}^\ast - \alpha_n \tilde{u}_{nB}^\ast ), & \text{if $\alpha_n \beta_n < 0$}.
  \end{cases}
\end{align}
Along with $\tilde{v}_n' = i \omega \tilde{u}_n'^\ast$, we have
\begin{align}
  P_C^\Gamma = \sum_{\tilde{\nu}_n < 1} ( \tilde{u}_n \tilde{v}_n^T + \tilde{u}_n' \tilde{v}_n'^T + \text{c.c.} ).
\end{align}
Now let us examine the decomposition of the full PT covariance matrix,
\begin{align}
  G^\Gamma &= \sum_{\tilde{\nu}_n < 1} \tilde{\nu}_n ( \tilde{u}_n^\ast \tilde{u}_n^T + \tilde{u}_n \tilde{u}_n^{\ast T} ) + \sum_{\tilde{\nu}_n \geq 1} \tilde{\nu}_n ( \tilde{u}_n^\ast \tilde{u}_n^T + \tilde{u}_n \tilde{u}_n^{\ast T} ) \nonumber\\
  &=: G^\Gamma_\mathcal{N} + G^\Gamma_{\cancel{\mathcal{N}}}.
\end{align}
If we denote part of the above projection as $P' := \sum_{\tilde{\nu}_n < 1} ( \tilde{u}_n' \tilde{v}_n'^T + \text{c.c.} )$, then by our construction we have
\begin{align}
  G_C^\Gamma = P_C^\Gamma G^\Gamma P_C^{\Gamma T} = G^\Gamma_\mathcal{N} \oplus P' G^\Gamma_{\cancel{\mathcal{N}}} P'^T.
\end{align}
Therefore, the symplectic eigenvalues of $G_C^\Gamma$ consists of those of $G^\Gamma_\mathcal{N}$, which are those of $G^\Gamma$ that contribute to the negativity, as well as those of $P' G^\Gamma_{\cancel{\mathcal{N}}} P'^T$.
One can show that the symplectic eigenvalues of $P' G^\Gamma_{\cancel{\mathcal{N}}} P'^T$ are all bounded below by one, using a similar argument to our previous optimality proof.
Namely, if we let $\tilde{\nu}''$ be the smallest symplectic eigenvalue of $P' G^\Gamma_{\cancel{\mathcal{N}}} P'^T$, and $\tilde{v}''$ the normalized dual vector of the corresponding symplectic eigenvector, then $\tilde{\nu}'' = \tilde{v}''^{\ast T} G^\Gamma_{\cancel{\mathcal{N}}} \tilde{v}'' = \tilde{v}''^{\ast T} G^\Gamma \tilde{v}''$ (note that $P'^T \tilde{v}'' = \tilde{v}''$ since it is in the subspace spanned by $\{ \tilde{v}_n' \}_{\tilde{\nu}_n < 1}$).
Since $\tilde{v}''$ is normalized as $\tilde{v}''^{\ast T} \Omega \tilde{v}'' = i$,
\begin{align}
  1 &= \sum_n ( |\tilde{v}''^T \tilde{u}_n|^2 - |\tilde{v}''^T \tilde{u}_n^\ast |^2 ) \nonumber\\
  &= \sum_{\tilde{\nu}_n \geq 1} ( |\tilde{v}''^T \tilde{u}_n|^2 - |\tilde{v}''^T \tilde{u}_n^\ast |^2 ).
\end{align}
Therefore,
\begin{align}
  \tilde{\nu}'' &= \tilde{v}''^{\ast T} G^\Gamma \tilde{v}''\nonumber \\
  &= \sum_{\tilde{\nu}_n \geq 1} \tilde{\nu}_n ( |\tilde{v}''^T \tilde{u}_n |^2 + | \tilde{v}''^T \tilde{u}_n^\ast |^2 ) \nonumber\\
  &\geq \sum_{\tilde{\nu}_n \geq 1} ( |\tilde{v}''^T \tilde{u}_n |^2 + | \tilde{v}''^T \tilde{u}_n^\ast |^2 )\nonumber \\
  &= 1 + 2 \sum_{\tilde{\nu}_n \geq 1} | \tilde{v}''^T \tilde{u}_n^\ast |^2 \geq 1.
\end{align}

\section{Phase space of the 1+1d massless scalar field}
\label{sec:qft_phase_space}

We will now apply the framework in the previous sections to the (1+1)-dimensional massless scalar field.
In this section, we will develop a precise definition of the phase space and co-phase space, as well as the kernels representing the bilinear forms in the right-left mover representation.
We will then use these in~\cref{sec:analytical} to calculate the logarithmic negativity and core modes for two intervals $A$ and $B$.
Throughout, we will use the Fourier convention
\begin{align}
  \phi(x) = \int \frac{dk}{2\pi} \, \tilde{\phi}(k) \, e^{ikx}.
\end{align}

\subsection{Phase space of canonical fields}

First, we will need to precisely specify the phase space and the space of linear observables (co-phase space) in terms of the canonical $\phi$ and $\pi$ fields.
Let us begin with the space of observables, which correspond to smearing functions in the quantum field theory.
We will assume these to be contained within the space of smooth functions with compact support, $(V^\ast)_{\phi\pi} \subset C_0^\infty(\mathbb{R}) \oplus C_0^\infty(\mathbb{R})$, with element $(f,h)_{\phi\pi}$ representing the observable $(f,h)_{\phi\pi}^T (\phi,\pi)_{\phi\pi} = \int dx \, [ f(x) \phi(x) + h(x) \pi(x) ]$.
Note that the $\phi\pi$ subscripts are used to indicate that we are using the representation in terms of the canonical fields $\phi$ and $\pi$.

The canonical commutation relations $[ \phi(x), \pi(x') ] = i \delta(x-x')$ and $[ \phi(x), \phi(x') ] = [ \pi(x), \pi(x') ] = 0$ are encoded in $\Omega : V^\ast \times V^\ast \to \mathbb{R}$ as
\begin{align}
  &\Omega((f_1,h_1)_{\phi\pi},(f_2,h_2)_{\phi\pi}) := \nonumber\\
  &\qquad \qquad \qquad \int dx \, [ f_1(x) h_2(x) - f_2(x) h_1(x) ].
\end{align}
The equal-time vacuum correlation functions of the quantum field theory are
\begin{align}
  G_\phi(x,x') := \bra{0} \{ \phi(x), \phi(x') \} \ket{0} &= \int \frac{dk}{2\pi} \frac{1}{|k|} e^{ik(x-x')}, \\
  G_\pi(x,x') := \bra{0} \{ \pi(x), \pi(x') \} \ket{0} &= \int \frac{dk}{2\pi} |k| e^{ik(x-x')},
\end{align}
and $\bra{0} \{ \phi(x), \pi(x') \} \ket{0} = 0$.
Therefore, the bilinear form $G : V^\ast \times V^\ast \to \mathbb{R}$ is represented by $G_\phi \oplus G_\pi$, i.e.,
\begin{align}
  &G((f_1,h_1)_{\phi\pi},(f_2,h_2)_{\phi\pi}) := 
  \nonumber\\
  &\int dx dx' \, [ f_1(x) G_\phi(x,x') f_2(x') + h_1(x) G_\pi(x,x') h_2(x') ].
\end{align}
Because $G_\phi$ has an infrared divergence at $k = 0$, in order for this quantity to be finite, we need to restrict its inputs to some subspace of $C_0^\infty(\mathbb{R})$.
We see that the input functions to $G_\phi$ must at least vanish at $k = 0$, which implies that they have vanishing integrals, $\int dx \, f(x) = 0$.
Combined with the fact that these functions have compact support, hence are analytic in Fourier space, means that they behave as $\mathcal{O}(k)$ as $k \to 0$.
Hence, the vanishing integral condition is also sufficient for the action of $G_\phi$ to be finite.
Therefore, our co-phase space will be
\begin{align}
  (V^\ast)_{\phi\pi} = \mathcal{F}(\mathbb{R}) \oplus C_0^\infty(\mathbb{R}),
\end{align}
where
\begin{align}
  \mathcal{F}(\mathbb{R}) := \{ f \in C_0^\infty(\mathbb{R}) : \int dx \, f(x) = 0 \}.
\end{align}
The vanishing integral condition for $f \in \mathcal{F}(\mathbb{R})$ is equivalent to $f(x) = F'(x)$ for some $F \in C_0^\infty(\mathbb{R})$.\footnote{Let $f \in C_0^\infty(\mathbb{R})$ with $\text{supp} (f) \subset (a,b)$. Then $F(x) := \int_a^x dx' \, f(x')$ is smooth and has support in $(a,b)$ since the integral of $f$ vanishes. Hence $f(x) = F'(x)$ with $F \in C_0^\infty(\mathbb{R})$.}
Thus, $\int dx \, f(x) \phi(x) = - \int dx \, F(x) \phi'(x)$, which can be interpreted physically as the field $\phi'(x)$ being observable, rather than $\phi(x)$.

The classical phase space of the massless scalar field in 1+1 dimensions can be identified with a pair of functions $\phi(x)$ and $\pi(x)$ used as initial data at $t=0$.
Here we will assume that these are both contained within the space of smooth functions with compact support, $(V)_{\phi\pi} \subset C_0^\infty(\mathbb{R}) \oplus C_0^\infty(\mathbb{R})$, with elements denoted $(\phi,\pi)_{\phi\pi}$.
The symplectic form $\omega = \Omega^{-1}$ is given by
\begin{align}
  &\omega((\phi_1,\pi_1)_{\phi\pi},(\phi_2,\pi_2)_{\phi\pi}) \nonumber\\
  &= - \int dx \, [ \phi_1(x) \pi_2(x) - \phi_2(x) \pi_1(x) ].
\end{align}
To determine $g = G^{-1}$, notice that as integral kernels we formally have $G_\pi = G_\phi^{-1}$, so $g = G_\pi \oplus G_\phi$,
\begin{align}
 & g((\phi_1,\pi_1)_{\phi\pi},(\phi_2,\pi_2)_{\phi\pi}) := 
  \nonumber\\
  &\int dx dx' \, \left[ \phi_1(x) G_\pi(x,x') \phi_2(x') 
  + \pi_1(x) G_\phi(x,x') \pi_2(x') \right].
\end{align}
We see that for this quantity to be well defined, we need to impose the vanishing integral condition on the $\pi$ fields.
Therefore, our phase space is
\begin{align}
  (V)_{\phi\pi} = C_0^\infty(\mathbb{R}) \oplus \mathcal{F}(\mathbb{R}).
\end{align}

\subsection{Right- and left-moving representation}

We will find it convenient to rewrite these quantities in terms of right- and left-moving fields, which corresponds to a change of basis in the phase space.
Solutions to the wave equation in 1+1 dimensions can be written as $\phi(t,x) = \phi_R(t-x) + \phi_L(t+x)$.
Therefore, at $t=0$ we have
\begin{align}
  \phi(x) &= \phi_R(x) + \phi_L(x), \\
  \pi(x) &= -\phi_R'(x) + \phi_L'(x).
\end{align}
The fields $\phi_{R/L}(x)$ are determined nonuniquely by
\begin{align}
  \phi_R(x) &= \frac12 \phi(x) - \frac12 \int_{-\infty}^x dx' \, \pi(x') + C, \\
  \phi_L(x) &= \frac12 \phi(x) + \frac12 \int_{-\infty}^x dx' \, \pi(x') - C,
\end{align}
where $C$ is an arbitrary constant.
Since $C$ is not observable, we are free to fix $C = 0$, which makes the decomposition unique.
It further implies that the fields $\phi_{R/L}$ will be compactly supported.
We can thus equivalently represent elements in phase space by $(\phi,\pi)_{\phi\pi} = (\phi_R,\phi_L)_{RL}$, where the $RL$ subscript denotes the basis of $\phi_{R/L}$ fields, which are related to $\phi$ and $\pi$ by the above transformations.
The phase space in this representation is
\begin{align}
  (V)_{RL} = C_0^\infty(\mathbb{R}) \oplus C_0^\infty(\mathbb{R}).
\end{align}

We can similarly write linear observables in the $RL$ representation.
To show the corresponding transformation, we write
\begin{align}
  (f,0)_{\phi\pi}^T (\phi,\pi)_{\phi\pi} &= \int dx \, f(x) \phi(x) \nonumber\\
  &= \int dx \, f(x) \phi_R(x) + \int dx \, f(x) \phi_L(x) 
  \nonumber\\  &
  = (f,f)_{RL}^T \, (\phi_R,\phi_L)_{RL},
\end{align}
and
\begin{align}
  (0,h)_{\phi\pi}^T (\phi,\pi)_{\phi\pi} &= \int dx \, h(x) \pi(x) \nonumber\\
  &= -\int dx \, h(x) \phi_R'(x) + \int dx \, h(x) \phi_L'(x) \nonumber\\
  &= \int dx \, h'(x) \phi_R(x) - \int dx \, h'(x) \phi_L(x) \nonumber\\
  &= (h',-h')_{RL}^T \, (\phi_R,\phi_L)_{RL}.
\end{align}
Extending these by linearity, we get
\begin{align}
  (f,h)_{\phi\pi} = (f + h', f - h')_{RL}.
\end{align}
Hence, the co-phase space in this representation is
\begin{align}
  (V^\ast)_{RL} = \mathcal{F}(\mathbb{R}) \oplus \mathcal{F}(\mathbb{R}).
\end{align}
The inverse of the above transformation is
\begin{align}
  (f_R,f_L)_{RL} = ( \tfrac12 (f_R + f_L), \tfrac12 (F_R - F_L) )_{\phi\pi},
\end{align}
where $F_{R/L}(x) = \int_{-\infty}^x dx' \, f_{R/L}(x')$.
Note that $F_{R/L}$ is compactly supported due to the vanishing integrals of $f_{R/L}$.

Using the above transformations, we can easily compute the $RL$ representations of $\Omega$ and $G$.
First, note that we can write the transformation from $(\phi,\pi)_{\phi\pi}$ to $(\phi_R,\phi_L)_{RL}$ in Fourier space as
\begin{align}
  \tilde{\phi}_{R/L}(k) = \tfrac12 [ \tilde{\phi}(k) \mp \tfrac{1}{ik} \tilde{\pi}(k) ].
\end{align}
We then find $(\Omega)_{RL} = \Omega_R \oplus \Omega_L$ and $(G)_{RL} = G_R \oplus G_L$, with
\begin{align}
  \Omega_R(x,x') = -\Omega_L(x,x') &= \int \frac{dk}{2\pi} \frac{1}{2ik} e^{ik(x-x')} \nonumber \\
  &= \frac14 \sgn(x-x') \\[2ex]
  G_R(x,x') = G_L(x,x') &= \int \frac{dk}{2\pi} \frac{1}{2|k|} e^{ik(x-x')}.
\end{align}
Hence the right- and left-moving fields correspond to two subsystems which are uncorrelated.
Note that the vanishing integral condition for smearing functions implies that two right- (or left-)moving observables supported in spacelike-separated regions will commute, even though the kernel $\Omega_{R/L}$ is nonlocal.
Similarly, we have $(\omega)_{RL} = \omega_R \oplus \omega_L$ and $(g)_{RL} = g_R \oplus g_L$, with
\begin{align}
  \omega_R(x,x') = -\omega_L(x,x') &= \int \frac{dk}{2\pi} (2ik) e^{ik(x-x')} \nonumber \\
  &= 2 \frac{d}{dx} \delta(x-x'), \\[2ex]
  g_R(x-x') = g_L(x-x') &= \int \frac{dk}{2\pi} (2|k|) e^{ik(x-x')}.
\end{align}
Therefore, $(J)_{RL} = J_R \oplus J_L$ with $J_R = - G_R \omega_R = G_L \omega_L = - J_L$, and
\begin{align}
  J_R(x,x') = -J_L(x,x') &= -i \int \frac{dk}{2\pi} \sgn(k) e^{ik(x-x')} \nonumber\\
  &= \frac{1}{\pi} \text{P.V.} \frac{1}{x-x'}.
\end{align}
We recognize $J_R$ as the Hilbert transform.

\subsection{Subsystems of (multiple) intervals}

Now we will discuss the restriction to local subsystems.
A local observer in the quantum field theory only has access to some region of space, for example an interval $A = (a_1,a_2)$, and thus their linear observables should be supported only in $A$.
In the case of a single interval, we identify the subspaces of $V$ and $V^\ast$ corresponding to the subsystem $A$ by $(V_A)_{RL} = C_0^\infty(A) \oplus C_0^\infty(A)$ and $(V^\ast_A)_{RL} = \mathcal{F}(A) \oplus \mathcal{F}(A)$, respectively, where
\begin{align}
  \mathcal{F}(A) := \{ f \in C_0^\infty(A) : \int_{a_1}^{a_2} dx \, f(x) = 0 \}.
\end{align}

In this paper, we are also interested in subsystems of two intervals $A = (a_1,a_2)$ and $B = (b_1,b_2)$, with $a_1 < a_2 < b_1 < b_2$.
In this case, if we were to define $\mathcal{F}(A \cup B)$ by merely imposing that the total integral over $A \cup B$ vanishes, then for $f \in \mathcal{F}(A \cup B)$ we would have $F(x) = \int_{a_1}^x dx' \, f(x')$ compactly supported on $(a_1,b_2)$, but it would be allowed to have a nonzero constant value between the two intervals.
This would imply, for example, that we would allow for observables $\int dx \, f(x) \phi(x) = - \int dx \, F(x) \phi'(x)$ where $F(x)$ may be nonzero between $A$ and $B$.
It would also imply that generally $\phi_{R/L}$ would be supported on $(a_1,b_2)$ rather than $A \cup B$, even if $\phi$ and $\pi$ are supported on the latter.

In our physical setup, we consider $A \cup B$ to correspond to the joint system of two observers, each with access to one of the intervals, rather than a single nonlocal observer with access to the whole of $A \cup B$.
Therefore, we conclude that the joint system should be described by the direct sum of these two subsystems.
We then define
\begin{align}
  \mathcal{F}(A \cup B) &:= \mathcal{F}(A) \oplus \mathcal{F}(B) \nonumber \\[1ex]
  &= \{ f \in C_0^\infty(A \cup B) : \nonumber \\
  &\qquad \quad \int_{a_1}^{a_2} dx \, f(x) = \int_{b_1}^{b_2} dx \, f(x) = 0 \},
\end{align}
along with $(V_{AB})_{RL} = C_0^\infty(A \cup B) \oplus C_0^\infty(A \cup B)$ and $(V^\ast_{AB})_{RL} = \mathcal{F}(A \cup B) \oplus \mathcal{F}(A \cup B)$.
One could extend the definition analogously to multiple intervals.

As a consistency check, one can follow a similar procedure to that in \cref{sec:analytical} to construct eigenfunctions of $J^T$ (i.e., without partial transposition, so that the eigenfunctions correspond to the normal modes of the reduced state of the field in $A \cup B$).
It is possible to construct eigenfunctions which have a vanishing total integral over $A \cup B$ and symplectic eigenvalues with $|\nu| < 1$ (which violates the uncertainty principle).
These unphysical modes can be removed by imposing the stronger condition above, where the integral must vanish over the individual intervals, since one can show it cannot be satisfied by any of the eigenfunctions with $|\nu| < 1$.
This demonstrates the necessity of this stronger condition.\footnote{These conditions are also imposed in \cite{arias_entropy_2018} in calculating the eigenfunctions of $J^T$. Here we show how they arise from the phase space structure of the field theory.}

\subsection{Phase space formalism in infinite dimensions}

To conclude this section, we highlight
some of the subtleties which arise in the phase space formalism in this infinite-dimensional setting.
While this is not essential reading to follow the later calculations, we include this discussion since we have modified some of the formal structure presented in \cref{sec:gaussian}.
For simplicity, we  focus only on the right-moving sector.

First, unlike the finite-dimensional setting, the co-phase space $(V^\ast)_R = \mathcal{F}(\mathbb{R})$ is only a proper subspace of the (continuous) dual space of the global phase space $(V)_R = C_0^\infty(\mathbb{R})$, while the full dual space consists of a space of distributions acting on the latter.
Recall that we made these restrictions to obtain appropriate domains for the bilinear forms $\Omega$, $G$, $\omega$, and $g$.
In finite dimensions, we also had that $\Omega$ induced a bijection $V^\ast\to V$.
Despite the fact that $V^\ast$ is no longer the dual space of $V$, by the construction of $V$ and $V^\ast$, $\Omega$ is still a bijection with $\omega$ acting as its inverse.

The main subtlety we address in this subsection is how to define the restriction of $J$ (and the bilinear forms $G$, $\omega$, etc.) to a subsystem, for example, corresponding to an interval $A = (a_1,a_2)$.
In finite dimensions, one can simply use a projection $P_A$ onto the subspace $V_A$ to obtain compressions of the form $G_A = P_A G P_A^T$ or $J_A = P_A J P_A$.
In the quantum field theory, we can formally project the kernel $J(x,x')$ to the interval $A$ using a characteristic function $\chi_A(x)$ by $\chi_A(x) J(x,x') \chi_A(x')$.
However, suppose it then acts on a function $u(x) \in C_0^\infty(A)$.
The first $\chi_A(x')$ simply acts as the identity, followed by $J(x,x')$ which multiplies by $(-i) \, \sgn(k)$ in Fourier space.
Functions with compact support, such as $u(x)$, must be an entire function in Fourier space (by a Paley-Wiener theorem \cite{rudin_real_1987}).
However, after $J$ acts on $u$, the resulting function is no longer analytic in Fourier space, hence $(Ju)(x)$ cannot have compact support (although it is still smooth).
Multiplying this function by another $\chi_A(x)$ will create discontinuities at the boundary of $A$.
Therefore, although the resulting function will be supported in $A$, it is not an element of $(V_A)_R = C_0^\infty(A)$ (nor even the full space $(V)_R = C_0^\infty(\mathbb{R})$).
Similarly, suppose we reduce $G$ by $\chi_A(x) G(x,x') \chi_A(x')$.
If we then consider a general $u(x) \in \mathcal{F}(\mathbb{R})$, the Fourier transform of $\chi_A(x) u(x)$ may not vanish at $k=0$, hence inserting this into $G$ will cause it to diverge.

The basic problem with this way of reducing these quantities to $A$ is that the subspace $V_A$ is not the image of the projection $\chi_A(x)$ on the global phase space $V$.
Rather, the physical motivation behind our definition of $V_A^\ast$, for example, is an observer in $A$ who has access to a limited set of linear observables from $V^\ast$ (i.e., those supported in $A$).
Hence, we should identify the bilinear forms reduced to $A$ by restricting the domains of the global forms.
For example, we restrict $G$ by $G_A := G |_A : \mathcal{F}(A) \times \mathcal{F}(A) \to \mathbb{R}$, and similar for the other bilinear forms.
As for the induced map $\Omega : \mathcal{F}(\mathbb{R}) \to C_0^\infty(\mathbb{R})$, note that the manner in which we defined $\mathcal{F}(A)$ ensures that the the image of $\mathcal{F}(A)$ lies in $C_0^\infty(A)$.
This is because, for $v(x) \in \mathcal{F}(A)$, we have $(\Omega v)(x) = \tfrac12 \int_{a_1}^x dx' \, v(x')$, which is smooth and vanishes for $x \geq a_2$.
Further, it is bijective, with the inverse of $\Omega |_A$ being simply the restriction $\omega |_A$.

As for $J_A$, we saw above that the restriction of the domain of the linear operator $J$ to $C_0^\infty(A)$ maps into $C^\infty(\mathbb{R})$ (which is outside of $V$).
Therefore, instead of identifying $J$ as a linear map, we will instead consider it as a bilinear map $J : \mathcal{F}(\mathbb{R}) \times C_0^\infty(\mathbb{R}) \to \mathbb{R}$.
That is, we define $J$ by its ``matrix elements'' $v^T J u$ for $v \in \mathcal{F}(\mathbb{R})$ and $u \in C_0^\infty(\mathbb{R})$.
It is then straightforward to define the reduction to subsystem $A$ by restricting its domain, i.e., $J_A := J |_A : \mathcal{F}(A) \times C_0^\infty(A) \to \mathbb{R}$.
Further, this perspective helps to formulate the eigenvalue problem of $J_A$.
It is known from prior work that one should not expect the eigenfunctions $J_A$ to be in $C_0^\infty(A)$.
Indeed, for the interval $A = (a_1,a_2)$, its eigenfunctions are of the form $(x-a_1)^{-is} (a_2-x)^{is}$ (see, e.g., \cite{arias_entropy_2018}).
However, typically one employs these mode functions only as a basis in which to expand the field operators, and they are not required to reside in the domain of $J_A$ (e.g., in a way analogous to a rigged Hilbert space construction~\cite{madrid_role_2005}).
Defining $J_A$ in terms of its matrix elements then suggests to formulate the eigenvalue problem for $J_A$ in a weak sense,
\begin{align}
  (J_A^T v)^T u_\lambda = \lambda \, v^T u_\lambda, \quad \forall v \in V_A^\ast.
\end{align}
That is, we define it in a way such that $J_A^T$ acts only on functions in $V_A^\ast$, and it allows for $u_\lambda$ to be a distribution (instead of being constrained to $V_A$).
This is analogous to formulations of the eigenvalue problem for the position and momentum operators in quantum mechanics, with eigenfunctions as distributions rather than elements in the Hilbert space.
One can similarly define an eigenvalue problem for $J_A^T$ by
\begin{align}
  v_\lambda^T (J_A u) = \lambda \, v_\lambda^T u, \quad \forall u \in V_A.
\end{align}
Thus, we allow for the eigenfunctions $v_\lambda$ of $J_A^T$ (i.e., left eigenfunctions of $J_A$) to be distributions over $V_A$.
Note that, because our linear observables in $V_A^\ast$ were defined so that they annihilate the zero mode of the phase space elements under the pairing $v^T u$ (i.e., the vanishing integral condition of $\mathcal{F}(A)$), we also require the eigenfunctions of $J_A^T$ to have this property.

Therefore, when we write the eigenvalue problems $J_A u_\lambda = \lambda \, u_\lambda$ and $J_A^T v_\lambda = \lambda \, v_\lambda$, they should be understood in the above sense.
Note that this extends straightforwardly to restrictions to $\mathcal{F}(A \cup B)$, as well as to the partially-transposed version we consider in \cref{sec:analytical}.

\section{Analytical calculation of logarithmic negativity and core modes}
\label{sec:analytical}

Here we will provide the details of our calculation of the logarithmic negativity and eigenfunctions of the operator $J^\ammaG := (J^\Gamma)^T$.
Recall, in our setup we have two open intervals $A = (a_1,a_2)$ and $B = (b_1,b_2)$, with $a_1 < a_2 < b_1 < b_2$, which define the subsystems of two local observers.
The corresponding phase space is $(V)_{RL} = C_0^\infty(A \cup B) \oplus C_0^\infty(A \cup B)$ and co-phase space is $(V^\ast)_{RL} = \mathcal{F}(A \cup B) \oplus \mathcal{F}(A \cup B)$, where
\begin{align}
  \mathcal{F}(A \cup B) &:= \{ f \in C_0^\infty(A \cup B) : \nonumber \\
  &\qquad \quad \int_{a_1}^{a_2} dx \, f(x) = \int_{b_1}^{b_2} dx \, f(x) = 0 \}.
\end{align}
We also have the kernels
\begin{align}
  \Omega_R(x,x') = -\Omega_L(x,x') &= \frac14 \text{sgn}(x-x'), \\[1ex]
  \omega_R(x,x') = -\omega_L(x,x') &= 2 \frac{d}{dx} \delta(x-x'), \\[1ex]
  G_R(x,x') = G_L(x,x') &= \int \frac{dk}{2\pi} \frac{1}{2|k|} e^{ik(x-x')}, \\[1ex]
  g_R(x,x') = g_L(x,x') &= \int \frac{dk}{2\pi} (2|k|) e^{ik(x-x')}, \\[1ex]
  J_R(x,x') = -J_L(x,x') &= -i \int \frac{dk}{2\pi} \sgn(k) e^{ik(x-x')} \nonumber \\
  &= \frac{1}{\pi} \text{P.V.} \frac{1}{x-x'}.
\end{align}

\subsection{Partial transpose}

First, we must find an implementation of partial transposition, which acts as the identity on $A$ and applies a transposition to $B$.
In \cref{sec:transpose}, we showed that a transpose on $B$, $T$, is any involution $T^2 = 1$ such that $T \Omega_B T^T = -\Omega_B$, where $\Omega_B$ is $\Omega$ restricted to the subsystem $B$.
If we focus on only the right-moving modes, then since $\Omega_R(x,x') = \tfrac14 \sgn(x-x')$, we see that we can apply a transpose simply by inverting the interval $B=(b_1,b_2)$ with $x \mapsto b_1 + b_2 - x$.
Thus, partial transposition on the joint $AB$ system can be achieved using the following kernel
\begin{align}
  \Gamma(x,x') = \begin{cases}
    \delta(x - x'), & \text{if $x,x' \in A$} \\
    \delta(x + x' - b_1 - b_2), & \text{if $x,x' \in B$} \\
    0, & \text{otherwise}.
  \end{cases}
\end{align}
Since $\Omega_L = -\Omega_R$, we can apply the same map to the left-moving modes.
Therefore, our phase space map implementing partial transposition will be $\Gamma \oplus \Gamma$.

Note that since this acts independently on right- and left-moving modes, we have
\begin{align}
  G^\Gamma := \Gamma G_R \Gamma^T \oplus \Gamma G_L \Gamma^T,
\end{align}
and therefore,
\begin{align}
  J^\Gamma := - G^\Gamma \omega = J_R^\Gamma \oplus J_L^\Gamma,
\end{align}
with $J_L^\Gamma = -J_R^\Gamma$.
The eigenvalues of $J_R^\Gamma := - G_R^\Gamma \omega_R$ come in pairs $\pm i \tilde{\nu}$, hence $J_R^\Gamma$ and $J_L^\Gamma$ have the same spectrum.
Therefore, both sectors contribute equally to the logarithmic negativity and $J^\Gamma$ has (at least) a twofold degeneracy in all of its eigenvalues.
Further, we also see that the degenerate eigenvectors are easily related, since if $J_R^\Gamma \tilde{u} = i \tilde{\nu} \tilde{u}$, then $J_L^\Gamma \tilde{u}^\ast = i \tilde{\nu} \tilde{u}^\ast$.

Therefore, in the following, we will focus only on the right-moving fields, $J_R^\Gamma$, as we can easily obtain the corresponding eigenvalues and eigenvectors of $J_L^\Gamma$.
We will also drop the subscripts $R$ for simplicity.

\subsection{$J^\ammaG$ kernel operator}

For the purpose of finding the eigenvalues of $J^\Gamma$, we will find it easier to work with its transpose $J^\ammaG := (J^\Gamma)^T$.
These have the same spectrum, and since $J^\ammaG$ acts on $(V^\ast)_R = \mathcal{F}(A \cup B)$, this will allow us to solve directly for the smearing functions associated with the eigenspaces contributing to the logarithmic negativity.

\begin{widetext}
First we write the partial transpose of the covariance matrix, $G^\Gamma := \Gamma G \Gamma^T$,
\begin{align}
G^\Gamma(x,x') = \begin{cases}
  \int \frac{dk}{2\pi} \frac{1}{2|k|} e^{ik(x-x')}, & \text{if $x,x' \in A$ or $x,x' \in B$} \\
  \int \frac{dk}{2\pi} \frac{1}{2|k|} e^{ik(x+x'-b_1-b_2)}, & \text{if $x \in A$, $x' \in B$ or $x \in B$, $x' \in A$}.
\end{cases}
\end{align}
Next we compute an expression for $J^\ammaG := (J^\Gamma)^T = \omega G^\Gamma$,
\begin{align}
J^\ammaG(x,x') = 2 \frac{d}{dx} G^\Gamma(x,x') 
&= \begin{cases}
  i \int \frac{dk}{2\pi} \sgn(k) e^{ik(x-x')}, & \text{if $x,x' \in A$ or $x,x' \in B$} \\
  i \int \frac{dk}{2\pi} \sgn(k) e^{ik(x+x'-b_1-b_2)}, & \text{if $x \in A$, $x' \in B$ or $x \in B$, $x' \in A$}
\end{cases}\nonumber \\
&= \begin{cases}
  -\frac{1}{\pi} \text{P.V.} \frac{1}{x-x'}, & \text{if $x,x' \in A$ or $x,x' \in B$} \\
  -\frac{1}{\pi} \text{P.V.} \frac{1}{x+x'-b_1-b_2}, & \text{if $x \in A$, $x' \in B$ or $x \in B$, $x' \in A$}.
\end{cases}
\end{align}
We rewrite the action of $J^\ammaG$ in a form which will be convenient for computing its eigenfunctions.
For $x \in A$,
\begin{align}
(J^\ammaG f)(x) &= -\frac{1}{\pi} \fint_{a_1}^{a_2} \frac{dx'}{x-x'} f(x') - \frac{1}{\pi} \fint_{b_1}^{b_2} \frac{dx'}{x+x'-b_1-b_2} f(x') \nonumber\\
&= -\frac{1}{\pi} \fint_{a_1}^{a_2} \frac{dx'}{x-x'} f(x') - \frac{1}{\pi} \fint_{b_1}^{b_2} \frac{dx'}{x-x'} f(-x'+b_1+b_2), \label{eq:JammaG_A}
\end{align}
where $\fint$ denotes the principal value of the integral.
And for $x \in B$,
\begin{align}
(J^\ammaG f)(x) &= -\frac{1}{\pi} \fint_{a_1}^{a_2} \frac{dx'}{x+x'-b_1-b_2} f(x') - \frac{1}{\pi} \fint_{b_1}^{b_2} \frac{dx'}{x-x'} f(x') \nonumber\\
&= -\frac{1}{\pi} \fint_{a_1}^{a_2} \frac{dx'}{x+x'-b_1-b_2} f(x') - \frac{1}{\pi} \fint_{b_1}^{b_2} \frac{dx'}{x+x'-b_1-b_2} f(-x'+b_1+b_2) \nonumber\\
(J^\ammaG f)(-x+b_1+b_2) &= +\frac{1}{\pi} \fint_{a_1}^{a_2} \frac{dx'}{x-x'} f(x') + \frac{1}{\pi} \fint_{b_1}^{b_2} \frac{dx'}{x-x'} f(-x'+b_1+b_2) \label{eq:JammaG_B}.
\end{align}
\end{widetext}
Now consider the expressions \eqref{eq:JammaG_A} and \eqref{eq:JammaG_B} for an eigenfunction $(J^\ammaG f)(x) = i \tilde{\nu} f(x)$.
Let us write $f = f_A + f_B$, where $f_A$ and $f_B$ have support in $A$ and $B$, respectively.
We will also define $\tilde{f}_B(x) := -f_B(-x+b_1+b_2)$.
Then we see from \eqref{eq:JammaG_A} and \eqref{eq:JammaG_B},
\begin{align}
i \tilde{\nu} f_A(x) &= -\frac{1}{\pi} \fint_{a_1}^{a_2} \frac{dx'}{x-x'} f_A(x') + \frac{1}{\pi} \fint_{b_1}^{b_2} \frac{dx'}{x-x'} \tilde{f}_B(x'), \\
i \tilde{\nu} \tilde{f}_B(x) &= -\frac{1}{\pi} \fint_{a_1}^{a_2} \frac{dx'}{x-x'} f_A(x') + \frac{1}{\pi} \fint_{b_1}^{b_2} \frac{dx'}{x-x'} \tilde{f}_B(x').
\end{align}
Or, if we write $\tilde{f} = f_A + \tilde{f}_B$, we can combine these for $x \in A \cup B$,
\begin{align}
i \nu \tilde{f}(x) = -\frac{1}{\pi} \fint_{a_1}^{a_2} \frac{dx'}{x-x'} \tilde{f}(x') + \frac{1}{\pi} \fint_{b_1}^{b_2} \frac{dx'}{x-x'} \tilde{f}(x'). \label{eq:JammaG_eigval_eqn}
\end{align}
Therefore, we see that a function $\tilde{f} = f_A + \tilde{f}_B$ satisfying \eqref{eq:JammaG_eigval_eqn} corresponds to an eigenfunction $f = f_A + f_B$ of $J^\ammaG$, where $f_B(x) = -\tilde{f}_B(-x+b_1+b_2)$.

\subsection{Boundary value problem in $\mathbb{C}$}

Here we introduce a modification of the method of~\cite{arias_entropy_2018} to reformulate the eigenvalue problem of $J^\ammaG$ as a boundary value problem in the complex plane.
The key difference is the second condition below, which flips the boundary condition in the interval $B$ to account for the partial transposition.

Let us consider the intervals $A \cup B$ as segments of the real axis in $\mathbb{C}$.
For some $0 \neq \lambda \in \mathbb{C}$, consider the problem of finding a function $S(z)$ on $\mathbb{C}$ which is cut along $A \cup B$ with boundary values $S^\pm(x) := S(x + i0^\pm)$, where $x \in A \cup B$, and which satisfies:
\begin{enumerate}[i)]
  \item $S(z)$ analytic on $\mathbb{C} \setminus \overline{A \cup B}$ (where $\overline{A \cup B}$ denotes the closure of $A \cup B$), \label{cond:BVP_1}
  
  \item $S^+(x) = \begin{cases}
    \lambda S^-(x), & x \in A, \\ \lambda^{-1} S^-(x), & x \in B,
  \end{cases}$ \label{cond:BVP_2}
  
  \item $\lim_{|z| \to \infty} |z^2 S(z)| < \infty$, \label{cond:BVP_3}

  \item $\lim_{z \to p} | (z-p) S(z) | = 0$, where $p$ is any of the boundary points $\{ a_1, a_2, b_1, b_2 \}$. \label{cond:BVP_4}
\end{enumerate}
We will demonstrate that solutions to this problem are in one-to-one correspondence with formal eigenfunctions of $J^\ammaG$.
Among these, those which satisfy the vanishing integral condition of $\mathcal{F}(A \cup B)$ will comprise the solutions to our eigenvalue problem.

Suppose we have a function $S(z)$ satisfying conditions \eqref{cond:BVP_1}--\eqref{cond:BVP_4}.
Let $z \in \mathbb{C} \setminus \overline{A \cup B}$, then the following contour integral vanishes when taken along a closed counterclockwise path encircling $z$ and $\overline{A \cup B}$,
\begin{align}
\oint \frac{dz'}{z-z'} S(z') = 0. \label{eq:contour_intgl}
\end{align}
To see this, deform the contour to a large circle of radius $R$ centered at the origin, then
\begin{align}
\left| \oint \frac{dz'}{z-z'} S(z') \right| &\leq 2 \pi R \max_{|z'| = R} \frac{1}{|z-z'|} |S(z')| \nonumber\\
&= \frac{2\pi}{R} \max_{|z'| = R} \frac{1}{|z-z'|} |z'^2 S(z')| \nonumber\\
&\to 0 \quad \text{as $R \to \infty$, using \eqref{cond:BVP_3}.}
\end{align}
Note that this integral also vanishes under weaker conditions than \eqref{cond:BVP_3}, but we will see below that this stronger decay is enforced by the vanishing integral conditions for the eigenfunctions.

Now we take the vanishing integral \eqref{eq:contour_intgl} and decompose it into contributions from the point $z$, the upper and lower sides of the cuts $A \cup B$, as well as small semi-circles around the endpoints $\{ a_1, a_2, b_1, b_2 \}$.
The latter vanish due to condition \eqref{cond:BVP_4}.
Hence we obtain
\begin{align}
S(z) &= \frac{1}{2 \pi i} \int_{A \cup B} \frac{dx'}{z-x'} \left[ S^-(x') - S^+(x') \right].
\label{eq:analytic_extension}
\end{align}
Now we take limits as $z$ approaches the upper and lower sides of the cuts, $z = x + i0^\pm$ with $x \in A \cup B$.
Using the Sokhotski-Plemelj theorem,
\begin{align}
\frac{1}{x-x'+i0^\pm} = \mp \pi i \delta(x-x') + \text{P.V.} \frac{1}{x-x'},
\end{align}
we get
\begin{align}
&S^\pm(x) = \mp \frac12 \left[ S^-(x) - S^+(x) \right] \nonumber \\
&\qquad \qquad + \frac{1}{2 \pi i} \fint_{A \cup B} \frac{dx'}{x-x'} \left[ S^-(x') - S^+(x') \right] \\[1ex]
&S^+(x) + S^-(x) = \frac{1}{\pi i} \fint_{A \cup B} \frac{dx'}{x-x'} \left[ S^-(x') - S^+(x') \right].
\end{align}
Then we make use of condition \eqref{cond:BVP_2}.
Let us write $S^\pm_A(x)$ and $S^\pm_B(x)$ to denote the functions $S^\pm(x)$ with support restricted to the intervals $A$ and $B$, respectively.
\begin{widetext}
For $x \in A$,
\begin{align}
(1 + \lambda^{-1}) S^+_A(x) &= \frac{1}{\pi i} \fint_A \frac{dx'}{x-x'} (\lambda^{-1} - 1) S^+_A(x') + \frac{1}{\pi i} \fint_B \frac{dx'}{x-x'} (1 - \lambda^{-1}) S^-_B(x') \\
i \left( \frac{1+\lambda^{-1}}{1-\lambda^{-1}} \right) S_A^+(x) &= - \frac{1}{\pi} \fint_A \frac{dx'}{x-x'} S^+_A(x') + \frac{1}{\pi} \fint_B \frac{dx'}{x-x'} S^-_B(x'),
\end{align}
and for $x \in B$,
\begin{align}
(\lambda^{-1} + 1) S^-_B(x) &= \frac{1}{\pi i} \fint_A \frac{dx'}{x-x'} (\lambda^{-1} - 1) S^+_A(x') + \frac{1}{\pi i} \fint_B \frac{dx'}{x-x'} (1 - \lambda^{-1}) S^-_B(x') \\
i \left( \frac{1+\lambda^{-1}}{1-\lambda^{-1}} \right) S^-_B(x) &= - \frac{1}{\pi} \fint_A \frac{dx'}{x-x'} S^+_A(x') + \frac{1}{\pi} \fint_B \frac{dx'}{x-x'} S^-_B(x').
\end{align}
\end{widetext}
Therefore, we see that the boundary values $S^+_A(x) + S^-_B(x)$ solve \eqref{eq:JammaG_eigval_eqn} for $\tilde{f}(x)$ with eigenvalue $i\tilde{\nu} = i(1+\lambda^{-1})/(1-\lambda^{-1})$.
The corresponding eigenfunction is then $S^+_A(x) - S^-_B(-x+b_1+b_2)$.

Since we want the eigenfunctions of $J^\ammaG$ to serve as a basis for $\mathcal{F}(A \cup B)$, they should also satisfy the vanishing integral conditions in the definition of $\mathcal{F}(A \cup B)$, i.e.,
\begin{align}
\int_{a_1}^{a_2} dx \, S^+_A(x) &= 0, \quad \text{and} \\
\int_{b_1}^{b_2} dx \, [- S^-_B(-x+b_1+b_2)] &= 0 \iff \int_{b_1}^{b_2} dx \, S^-_B(x) = 0.
\end{align}
We mentioned above that this motivates the decay requirement in \eqref{cond:BVP_3}; we will now show this.
From \eqref{eq:analytic_extension}, we see that $S(z)$ at any point in the complex plane is determined by its boundary values $S^\pm(x)$.
Consider \eqref{eq:analytic_extension} for large $z$ (i.e., far from $A \cup B$),
\begin{align}
S(z) &= \frac{1}{2 \pi i} \int_{A \cup B} \frac{dx'}{z-x'} \left[ S^-(x') - S^+(x') \right] \\[1ex]
&= \frac{(\lambda^{-1}-1)}{2 \pi i} \int_{a_1}^{a_2} \frac{dx'}{z-x'} S^+_A(x') \nonumber \\
&\qquad + \frac{(1-\lambda^{-1})}{2 \pi i} \int_{b_1}^{b_2} \frac{dx'}{z-x'} S^-_B(x') \\[1ex]
&= \frac{(\lambda^{-1}-1)}{2 \pi i} \int_{a_1}^{a_2} dx' \, \frac1z \left[ 1 - \frac{x'}{z} + \cdots \right] S^+_A(x') \nonumber \\
&\qquad + \frac{(1-\lambda^{-1})}{2 \pi i} \int_{b_1}^{b_2} dx' \, \frac1z \left[ 1 - \frac{x'}{z} + \cdots \right] S^-_B(x').
\end{align}
We then see that the vanishing integral conditions on $S^+_A$ and $S^-_B$ imply that $S(z) \sim z^{-2}$ for large $|z|$.
However, the converse is not true; $S(z) \sim z^{-2}$ merely implies $\int_{a_1}^{a_2} dx \, S^+_A(x) = \int_{b_1}^{b_2} dx \, S^-_B(x)$.
Nevertheless, we will find it useful to apply \eqref{cond:BVP_3} as an intermediate step in constructing $S(z)$, and thereafter we will impose the stronger conditions $\int_{a_1}^{a_2} dx \, S^+_A(x) = \int_{b_1}^{b_2} dx \, S^-_B(x) = 0$.

To summarize, if we find a function $S(z)$ in the complex plane satisfying conditions \eqref{cond:BVP_1}--\eqref{cond:BVP_4} for some $0 \neq \lambda \in \mathbb{C}$, and further impose $\int_{a_1}^{a_2} dx \, S^+_A(x) = \int_{b_1}^{b_2} dx \, S^-_B(x) = 0$, then we get an eigenfunction $S^+_A(x) - S^-_B(-x+b_1+b_2)$ of $J^\ammaG$ with eigenvalue $i(1+\lambda^{-1})/(1-\lambda^{-1})$.
Conversely, an eigenfunction of $J^\ammaG$ defines a function in the complex plane, via \eqref{eq:analytic_extension}, which satisfies \eqref{cond:BVP_1}--\eqref{cond:BVP_4}.\footnote{In principle, one should relax the requirement \eqref{cond:BVP_4} to $\lim_{z \to p} | (z-p) S(z) | < \infty$, in order to allow for simple poles $(x-p)^{-1}$. For instance, the eigenfunctions of $J^T$ exhibit such behavior at the boundary points \cite{arias_entropy_2018}. In that case, one must interpret the action of $J^T$ in terms of a regularized integral. However, we will find below that the eigenfunctions of $J^\ammaG$ with $|\tilde{\nu}| < 1$ only diverge proportional to $(x-p)^{-1/2}$ near the boundaries, and relaxing condition \eqref{cond:BVP_4} as such does not produce any additional eigenfunctions. This allows us to use the stronger condition \eqref{cond:BVP_4}, which simplifies the previous analysis. It is possible that one may need to consider relaxing this constraint for the eigenfunctions of $J^\ammaG$ with $|\tilde{\nu}| \geq 1$.}

A general value $0 \neq \lambda \in \mathbb{C}$ will give us a formal eigenfunction of $J^\ammaG$, but not all of these values will be relevant for our problem.
For instance, we should have $\tilde{\nu} = (1+\lambda^{-1})/(1-\lambda^{-1}) \in \mathbb{R}$, which implies that $\lambda \in \mathbb{R}$.
Further, we will focus on finding eigenfunctions with $|\tilde{\nu}| < 1$, since these correspond to modes which will contribute to the negativity.
We can rewrite $\lambda = (\tilde{\nu}+1)/(\tilde{\nu}-1)$, hence $\tilde{\nu} \in (-1,1)$ corresponds to $\lambda \in (-\infty,0)$.
Therefore, in the following, we will only consider $\lambda \in (-\infty,0)$.

\subsection{Solution to boundary value problem}

We will proceed to solve the boundary value problem described in the previous section for $\lambda < 0$.
This follows a similar procedure to \cite{arias_entropy_2018}, with appropriate modifications which account for the partial transposition.

First, we will find a particular solution to conditions \eqref{cond:BVP_1} and \eqref{cond:BVP_2} (leaving \eqref{cond:BVP_3} and \eqref{cond:BVP_4} aside, for now).
We begin with the ansatz,
\begin{align}
(z-a_1)^{\alpha_1} (z-a_2)^{\alpha_2} (z-b_1)^{\beta_1} (z-b_2)^{\beta_2}. \label{eq:BVP_ansatz}
\end{align}
Consider the factor $(z-a_1)^{\alpha_1} = e^{\alpha_1 \ln(z-a_1)}$.
It has a branch point at $z=a_1$, and we choose the branch cut along the negative real axis.
Now if we take limits approaching the upper and lower sides of the branch cut, with $x = \Re(z) < a_1$, we have
\begin{align}
&(x + i0^\pm - a_1)^{\alpha_1} = e^{\pm \pi i \alpha_1} |x - a_1|^{\alpha_1}, \nonumber \\
&(x + i0^+ - a_1)^{\alpha_1} = e^{2 \pi i \alpha_1} (x + i0^- - a_1)^{\alpha_1}.
\end{align}
By choosing the exponents in \eqref{eq:BVP_ansatz} appropriately, we can use this behavior to find a solution to \eqref{cond:BVP_1} and \eqref{cond:BVP_2}.
Since we want the function to be continuous as we cross the real axis when $a_2 < x < b_1$, we can cancel the factor $e^{2 \pi i (\beta_1 + \beta_2)}$ by choosing $\beta_2 = -\beta_1$.
Note that this choice is not unique, but for now we seek only some particular solution.
Similarly, \eqref{eq:BVP_ansatz} will be continuous across $x < a_1$ if we choose $\alpha_2 = -\alpha_1$.
With these choices, \eqref{eq:BVP_ansatz} will satisfy condition \eqref{cond:BVP_1}.
For \eqref{cond:BVP_2}, we want to have $\lambda = e^{2 \pi i \alpha_2} = e^{-2 \pi i \beta_2}$, or $\beta_2 = -\alpha_2$.
Thus, \eqref{eq:BVP_ansatz} solves \eqref{cond:BVP_1} and \eqref{cond:BVP_2} if we pick $\alpha := \alpha_1 = -\alpha_2 = -\beta_1 = \beta_2$, with $\alpha$ chosen so that $\lambda = e^{-2 \pi i \alpha}$.

If we further impose that $\lambda \in (-\infty,0)$, then $\lambda = e^{-2 \pi i \alpha} = e^{-2 \pi i \Re(\alpha)} e^{2 \pi \Im(\alpha)}$ gives $\Re(\alpha) \in \tfrac12 + \mathbb{Z}$ and $\Im(\alpha) \in \mathbb{R}$.
Again, we will simply pick one solution for now, so set $\alpha = \tfrac12 + i s$ with $s \in \mathbb{R}$.
Hence, we have that for any $\lambda = - e^{2 \pi s} \in (-\infty,0)$ with $s \in \mathbb{R}$, the following function satisfies conditions \eqref{cond:BVP_1} and \eqref{cond:BVP_2},
\begin{align}
&(z-a_1)^{\tfrac12 + is} (z-a_2)^{-\tfrac12 - is} (z-b_1)^{-\tfrac12 - is} (z-b_2)^{\tfrac12 + is} \nonumber \\
&\qquad \qquad \qquad \quad =: e^{\left(\tfrac12 + is\right) w(z)},
\end{align}
where
\begin{align}
w(z) := \ln \left( \frac{z-a_1}{z-a_2} \frac{z-b_2}{z-b_1} \right).
\end{align}
  
Now, let $S(z)$ be any solution to the full boundary value problem \eqref{cond:BVP_1}--\eqref{cond:BVP_4} with some fixed $\lambda = - e^{2 \pi s} \in (-\infty,0)$.
Using the function we found above, consider the function
\begin{align}
M(z) := S(z) e^{-\left(\tfrac12 + is\right) w(z)},
\end{align}
where notice that we flipped the sign in the exponent.
We will now proceed to show that $M(z)$ is uniquely determined (up to a constant factor), which in turn fixes $S(z)$ uniquely.

Note that $M(z)$ is analytic on $\mathbb{C} \setminus \overline{A \cup B}$.
Further, the multiplicative factors from $S(z)$ and $e^{-(1/2 + is) w(z)}$ will cancel as we cross the cut $A \cup B$ (excluding the endpoints).
Thus, $M(z)$ is continuous as we take limits from either side of $A \cup B$.
A consequence of Morera's theorem implies that $M(z)$ can be uniquely extended to an analytic function on $\mathbb{C} \setminus \{a_1,a_2,b_1,b_2\}$, if we define $M(z)$ on $A \cup B$ by these limits (see, e.g., Section 27 of \cite{dennery_mathematics_1996}).

The function $M(z)$ is thus analytic with isolated singularities at $\{a_1,a_2,b_1,b_2\}$.
Now we will show that $M(z)$ is meromorphic, i.e., that these are not essential singularities.
The Laurent expansion of $M(z)$ around $z=b_1$ is
\begin{align}
M(z) = \sum_{n=0}^\infty c_n (z-b_1)^n + \sum_{n=1}^\infty \frac{d_n}{(z-b_1)^n},
\end{align}
with
\begin{align}
d_n = \frac{1}{2 \pi i} \int_{\gamma} dz \, (z-b_1)^{n-1} M(z),
\end{align}
where $\gamma$ is a small circle of radius $r$ centered at $z=b_1$.
Then,
\begin{widetext}
\begin{align}
|d_n| &\leq \frac{1}{2\pi} 2\pi r \max_{|z-b_1| = r} r^{n-1} |M(z)| \qquad \text{(Darboux inequality)} \nonumber\\
&= r^{n-1} \max_{|z-b_1| = r} |r S(z)| \left| e^{-\left(\tfrac12 + is\right) w(z)} \right| \nonumber\\
&= r^{n-1} \max_{|z-b_1| = r} |r S(z)| \left| \sqrt{\frac{|z-a_2||z-b_1|}{|z-a_1||z-b_2|}} e^{s \text{Arg}[(z-a_1)(z-b_2)/(z-a_2)(z-b_1)]} \right|
\leq C(r) \, r^{n-\tfrac12},
\end{align}
\end{widetext}
where $C(r) \to 0$ as $r \to 0$.
In the last step we used the fact that $|r S(z)|$ vanishes as $r \to 0$, due to condition \eqref{cond:BVP_4}.
Note that $| e^{-(1/2 + is) w(z)} |$ contributes a factor of $r^{1/2}$, and the remaining factors from this are bounded by constants for small $r$, where we also note that $\text{Arg}[...] \in (-\pi,\pi]$.
By taking the limit $r \to 0$, we see that $d_n = 0$ for $n \geq \frac12$.
This implies that $z=b_1$ is a removable singularity of $M(z)$.
A similar argument can be made to show that $z=a_2$ is a removable singularity.
For $a_1$ and $b_2$, the square root from $| e^{-(1/2 + is) w(z)} |$ contributes a factor of $r^{-1/2}$, which gives $|d_n| \leq C(r) \, r^{n-3/2}$.
This vanishes in the limit $r \to 0$ for $n \geq \frac32$.
Hence, $M(z)$ does not have an essential singularity, but may have a simple pole at these points.

Our function $M(z)$ is thus meromorphic with (possibly) simple poles at $a_1$ and $b_2$.
Using the Mittag-Leffler expansion, one can conclude that $M(z)$ must be of the form,
\begin{align}
M(z) = \frac{p(z)}{(z-a_1)(z-b_2)},
\end{align}
for some entire function $p(z)$.
Now we constrain the function $p(z)$ by examining its behavior for large $|z|$.
First, we write
\begin{align}
|p(z)| &= \left| (z-a_1) (z-b_2) S(z) e^{-\left(\tfrac12 + is\right) w(z)} \right| \\
&= \left| 1 - \frac{a_1}{z} \right| \left| 1 - \frac{b_2}{z} \right| \left| z^2 S(z) \right| \nonumber \\
&\quad \times \left| \sqrt{\frac{|z-a_2||z-b_1|}{|z-a_1||z-b_2|}} e^{s \, \text{Arg}\left[\frac{(z-a_1)(z-b_2)}{(z-a_2)(z-b_1)}\right]} \right|.
\end{align}
Notice that all of these factors are bounded as $|z| \to \infty$, including $|z^2 S(z)|$ due to condition \eqref{cond:BVP_3}.
Therefore, $p(z)$ must be a constant, by Liouville's theorem.

We have therefore shown that the unique solution (up to a constant factor) to the boundary value problem \eqref{cond:BVP_1}--\eqref{cond:BVP_4}, with $\lambda = -e^{2 \pi s}$ and $s \in \mathbb{R}$, is
\begin{align}
S(z) &= C (z-a_1)^{-\tfrac12+is} (z-a_2)^{-\tfrac12-is} \nonumber \\
&\qquad \quad \times (z-b_1)^{-\tfrac12-is} (z-b_2)^{-\tfrac12+is},
\end{align}
where $C \in \mathbb{C}$ is a normalization constant.
\cref{fig:BVP_solution_complex_plane} shows a plot of the magnitude and phase of the function $S(z)$ in the complex plane.

\begin{figure*}
\centering
\includegraphics[width=0.9\linewidth]{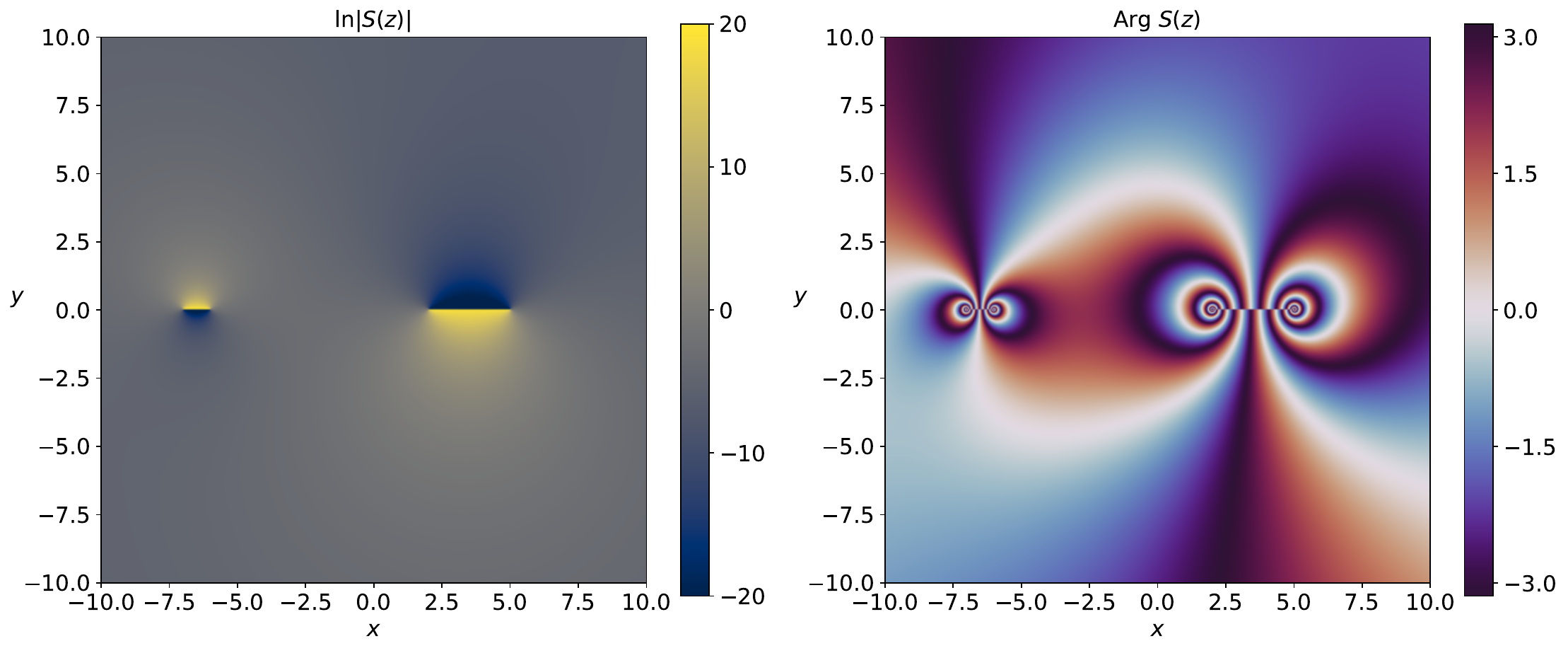}
\caption{Visualization of the solution $S(z)$ in the complex plane $z = x + iy$. The parameters are $A = (-7,-6)$, $B = (2,5)$, and $s \approx 6.837$ (corresponding to the smallest eigenvalue $\tilde{\nu}$, as explained below). Note that we plot the magnitude as $\ln |S(z)|$.}
\label{fig:BVP_solution_complex_plane}
\end{figure*}

Now that we have the solution to the boundary value problem, we can find the associated boundary values by taking the limits $z \to x + i0^+$ for $x \in A$,
\begin{align}
S^+_A(x) &= C' (x-a_1)^{-\tfrac12+is} (a_2-x)^{-\tfrac12-is} \nonumber \\
&\qquad \quad \times (b_1-x)^{-\tfrac12-is} (b_2-x)^{-\tfrac12+is} \nonumber\\[1ex]
&= \frac{C' e^{is \omega(x)}}{\sqrt{-(x-a_1)(x-a_2)(x-b_1)(x-b_2)}},
\end{align}
and $z \to x + i0^-$ for $x \in B$,
\begin{align}
S^-_B(x) &= C' (x-a_1)^{-\tfrac12+is} (x-a_2)^{-\tfrac12-is} \nonumber \\[1ex]
&\qquad \quad \times (x-b_1)^{-\tfrac12-is} (b_2-x)^{-\tfrac12+is} \nonumber\\
&= \frac{C' e^{is \omega(x)}}{\sqrt{-(x-a_1)(x-a_2)(x-b_1)(x-b_2)}},
\end{align}
where $C' := i e^{\pi s} C$ and, for $x \in \mathbb{R}$, we define
\begin{align}
\omega(x) := \ln \left( - \frac{x-a_1}{x-a_2} \frac{x-b_2}{x-b_1} \right).
\end{align}

Therefore, this corresponds to a (formal) eigenfunction $S^+_A(x) - S^-_B(-x+b_1+b_2)$ with eigenvalue $i\tilde{\nu}$, where
\begin{align}
\tilde{\nu} = \frac{1+\lambda^{-1}}{1-\lambda^{-1}} = \tanh(\pi s).
\end{align}
In the next section, we will examine the implications of the requirement that the integrals of these eigenfunctions must vanish over each of the intervals $A$ and $B$.

\subsection{Discrete spectrum and logarithmic negativity}

We have obtained the unique solution to the boundary value problem for any given choice of $\lambda = -e^{2 \pi s} \in (-\infty,0)$.
However, we also require $\int_{a_1}^{a_2} dx \, S^+_A(x) = \int_{b_1}^{b_2} dx \, S^-_B(x) = 0$ for the corresponding eigenfunctions to be valid solutions to our problem.
Here, we will show that these conditions lead to a discrete spectrum of $J^\ammaG$ for $|\tilde{\nu}| < 1$.

We first evaluate the integral
\begin{widetext}
\begin{align}
&\int_{a_1}^{a_2} dx \, S^+_A(x)
= C' \int_{a_1}^{a_2} dx \, (x-a_1)^{-\tfrac12+is} (a_2-x)^{-\tfrac12-is} (b_1-x)^{-\tfrac12-is} (b_2-x)^{-\tfrac12+is} \nonumber\\
&= C' (b_1-a_1)^{-\tfrac12-is} (b_2-a_1)^{-\tfrac12+is} \int_0^1 dt \, t^{-\tfrac12+is} (1-t)^{-\tfrac12-is} \left[ 1 - \left( \frac{a_2-a_1}{b_1-a_1} \right) t \right]^{-\tfrac12-is} \left[ 1 - \left( \frac{a_2-a_1}{b_2-a_1} \right) t \right]^{-\tfrac12+is}.
\end{align}
Now we perform the following change of variables, $t \mapsto (1-t)/\left[ 1 - \left(\frac{a_2-a_1}{b_2-a_1}\right) t \right]$, to obtain
\begin{align}
\int_{a_1}^{a_2} dx \, S^+_A(x) = C' (b_1-a_2)^{-\tfrac12-is} (b_2-a_1)^{-\tfrac12+is} \int_0^1 dt \, t^{-\tfrac12-is} (1-t)^{-\tfrac12+is} \left[ 1 - \left( \frac{\eta}{\eta-1} \right) t \right]^{-\tfrac12-is},
\end{align}
where $\eta$ is the cross ratio,
\begin{align}
\eta := \frac{(a_2-a_1)(b_2-b_1)}{(b_1-a_1)(b_2-a_2)} = \frac{\ell_A \ell_B}{(\ell_A + d)(\ell_B + d)}, \label{eq:cross_ratio}
\end{align}
and where $\ell_A = a_2 - a_1$ and $\ell_B = b_2 - b_1$ are the lengths of the two intervals, and $d = b_1 - a_2$ is their separation.
The expression above is an integral representation of a hypergometric function,
\begin{align}
\int_{a_1}^{a_2} dx \, S^+_A(x) = C' (b_1-a_2)^{-\tfrac12-is} (b_2-a_1)^{-\tfrac12+is} B\left(\tfrac12-is,\tfrac12+is\right) {}_2 F_1 \left( \tfrac12+is, \tfrac12-is; 1; \frac{\eta}{\eta-1} \right),
\end{align}
where $B$ is the beta function, which in this case is $B\left(\tfrac12-is,\tfrac12+is\right) = \pi / \cosh(\pi s)$.
For these parameter values in ${}_2 F_1$, it can be identified with a Legendre function (see 8.820.1 of \cite{gradshteyn_table_2007}),
\begin{align}
P_{-is-\tfrac12} \left( \frac{1+\eta}{1-\eta} \right) = {}_2 F_1 \left( \tfrac12+is, \tfrac12-is; 1; \frac{\eta}{\eta-1} \right).
\end{align}
We then obtain our final expression,
\begin{align}
\int_{a_1}^{a_2} dx \, S^+_A(x) = C' \frac{\pi}{\cosh(\pi s)} (b_1-a_2)^{-\tfrac12-is} (b_2-a_1)^{-\tfrac12+is} P_{-is-\tfrac12} \left( \frac{1+\eta}{1-\eta} \right).
\end{align}
\end{widetext}
One can follow similar steps for the integral in $B$ to find the same expression, i.e., $\int_{b_1}^{b_2} dx \, S^-_B(x) = \int_{a_1}^{a_2} dx \, S^+_A(x)$.
Recall that this is expected from the $S(z) \sim z^{-2}$ decay as $|z| \to \infty$.

The condition $\int_{a_1}^{a_2} dx \, S^+_A(x) = \int_{b_1}^{b_2} dx \, S^-_B(x) = 0$ is then equivalent to
\begin{align}
P_{-is-\tfrac12} \left( \frac{1+\eta}{1-\eta} \right) = 0. \label{eq:conical_zeros}
\end{align}
This will only vanish for particular values of $s$.
Thus, the roots of the Legendre function (in $s \in \mathbb{R}$, for fixed $\eta$) determine the allowed values of $s$, hence the spectrum of $J^\ammaG$.

Note that this condition only depends on the geometry through the cross ratio $\eta \in (0,1)$ defined in \eqref{eq:cross_ratio}.
The expressions below involving $\eta$ will be general, but we will often find it instructive to consider the special case where the sizes of the two intervals are equal $\ell = a_2 - a_1 = b_2 - b_1$ with separation $d = b_1 - a_2$.
In this case,
\begin{align}
\eta = \frac{1}{(1+d/\ell)^2},
\end{align}
which is a function only of the ratio $d/\ell$.
Two regimes of interest will be $d/\ell \to \infty$, which corresponds to $\eta \to 0$, and $d/\ell \to 0$, corresponding to $\eta \to 1$.

The Legendre function in \eqref{eq:conical_zeros} with the particular order $-is - \tfrac12$ is also known as a conical (or Mehler) function \cite{zhurina_tables_1966}.
We will use known properties of this function to infer general properties of the spectrum of $J^\ammaG$, and then proceed to find analytical approximations of the zeros in the two regimes $\eta \to 0$ and $\eta \to 1$.

First, the argument $\tfrac{1+\eta}{1-\eta}$ lies in the range $(1,\infty)$, and we will find it convenient to reparametrize it by
\begin{align}
\cosh \rho := \frac{1+\eta}{1-\eta}, \quad \text{with $\rho \in (0,\infty)$}.
\end{align}
On this domain, $P_{-is-1/2}(\cosh\rho)$ is real-valued \cite{zhurina_tables_1966}.
Further, we have that $P_{-is-1/2}(\cosh\rho) = P_{is-1/2}(\cosh\rho)$ \cite{zhurina_tables_1966}, which means that the zeros come in pairs $\pm s$.
Thus, the eigenvalues of $J^\ammaG$ come in pairs $\pm i \tilde{\nu} = \pm i \tanh(\pi s)$, as expected.
Because of this symmetry, we can restrict our attention to $s \geq 0$.
Importantly, a general Legendre function $P_\mu$ is an entire function of $\mu$ \cite{zhurina_tables_1966}.
This implies that the zeros of $P_{-is-1/2}$ (as a function of $s$) must be isolated, i.e., they form a discrete set $\{ s_n \}$, and hence the spectrum of $J^\ammaG$ (in the range $|\tilde{\nu}| < 1$) is discrete.

\cref{fig:conical_functions} shows a plot of $P_{-is-1/2}(\tfrac{1+\eta}{1-\eta})$ as a function of $s$ for different values of $\eta$.
From this plot, we see that the value of the first (smallest) zero decreases as $\eta$ increases (and similar for the second zero, etc.).
The first zero corresponds to the eigenfunction of $J^\ammaG$ which contributes the most to the logarithmic negativity, since the contribution of each eigenvalue is $E_{\mathcal{N},n} := -\log_2 \tilde{\nu}_n = -\log_2(\tanh (\pi s_n))$, which increases for smaller values of $s_n$.
Therefore, the observation that the first zero decreases with increasing $\eta$ is consistent with the expectation that the negativity should increase as the intervals get closer (e.g., if we consider two intervals of length $\ell$ and decreasing separation $d$).
\begin{figure}[h]
\centering
\includegraphics[width=\linewidth]{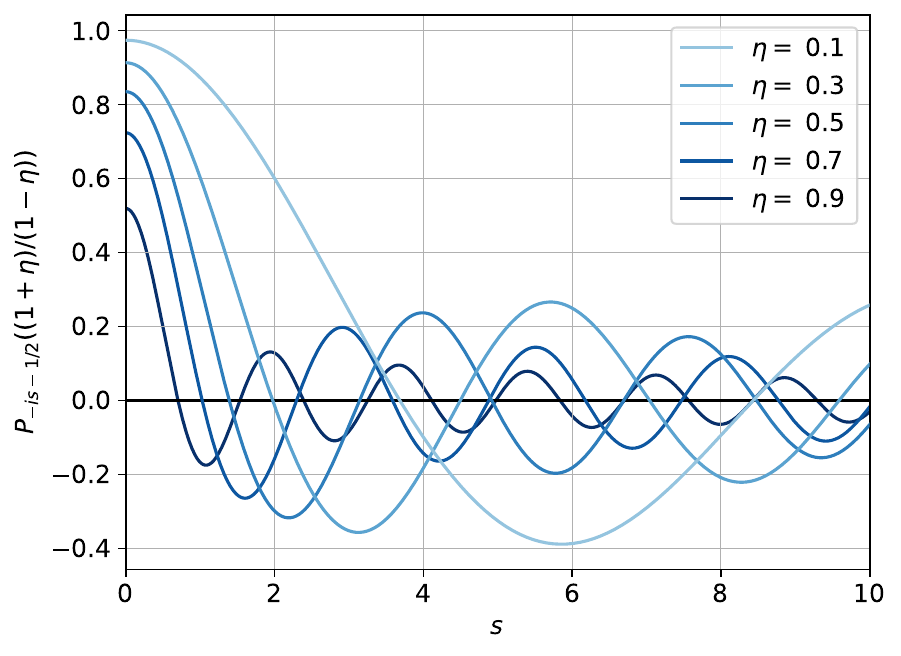}
\caption{Plot of $P_{-is-1/2}(\tfrac{1+\eta}{1-\eta})$ as a function of $s$ for different values of $\eta$. Computed using the \texttt{legenp} function from the Python \texttt{mpmath} library \cite{mpmath}.}
\label{fig:conical_functions}
\end{figure}

\subsubsection{Large separation regime}

We will now find analytical approximations of the zeros of $P_{-is-1/2}(\cosh\rho)$ in the regime where the intervals are widely separated compared to their sizes, i.e., $\eta \to 0$ or $\rho \to 0$.
Based on the observations from \cref{fig:conical_functions}, we seek to approximate $P_{-is-1/2}(\cosh\rho)$ when $\rho$ is small and $s$ is large.
We will use the leading order term in a series expansion in terms of Bessel functions \cite{zhurina_tables_1966},
\begin{align}
P_{-is-1/2}(\cosh\rho) = \sqrt{\frac{\rho}{\sinh\rho}} J_0(s\rho) + \mathcal{O}(s^{-3/2}).
\end{align}
The full series is convergent for $\rho < 2\pi\sqrt{2+\sqrt{5}}$.
Note that the convergence criterion is equivalent to $\eta \lessapprox 0.99999032$ or, when the sizes of the intervals are equal, $d/\ell \gtrapprox 4.84 \times 10^{-6}$.
Beyond this, the expansion can still be used as an asymptotic series as $s \to \infty$.

We will approximate the (non-negative) zeros of $P_{-is-1/2}(\cosh\rho)$ using the zeros of the first term in the series.
If we denote the zeros of $J_0$ by $j_{0,n}$ with $n \in \mathbb{N}$, this gives
\begin{align}
s_n \approx \frac{j_{0,n}}{\rho} = \frac{j_{0,n}}{\text{arccosh}\left( \frac{1+\eta}{1-\eta} \right)}, \quad n \in \mathbb{N}.
\label{eq:conical_zeros__large_separation}
\end{align}
To assess the accuracy of this approximation, in \cref{fig:conical_zeros__large_separation}, we plot the relative error between the first ten of these approximated zeros with the zeros of $P_{-is-1/2}$ computed numerically using the \texttt{root\_scalar} method from \texttt{scipy.optimize} in Python \cite{scipy} (with the approximations \eqref{eq:conical_zeros__large_separation} as initial guesses).

\begin{figure*}
\centering
\begin{subfigure}[t]{0.47\textwidth}
\centering
\includegraphics[width=\linewidth]{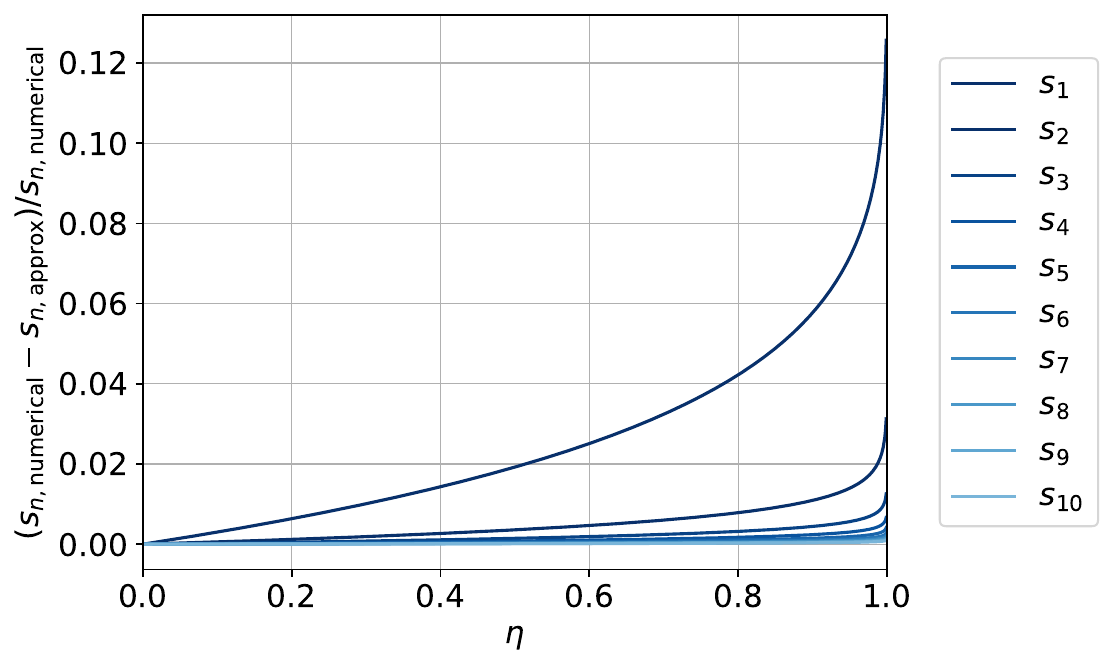}
\caption{}
\end{subfigure}
\hfill
\begin{subfigure}[t]{0.47\textwidth}
\centering
\includegraphics[width=\linewidth]{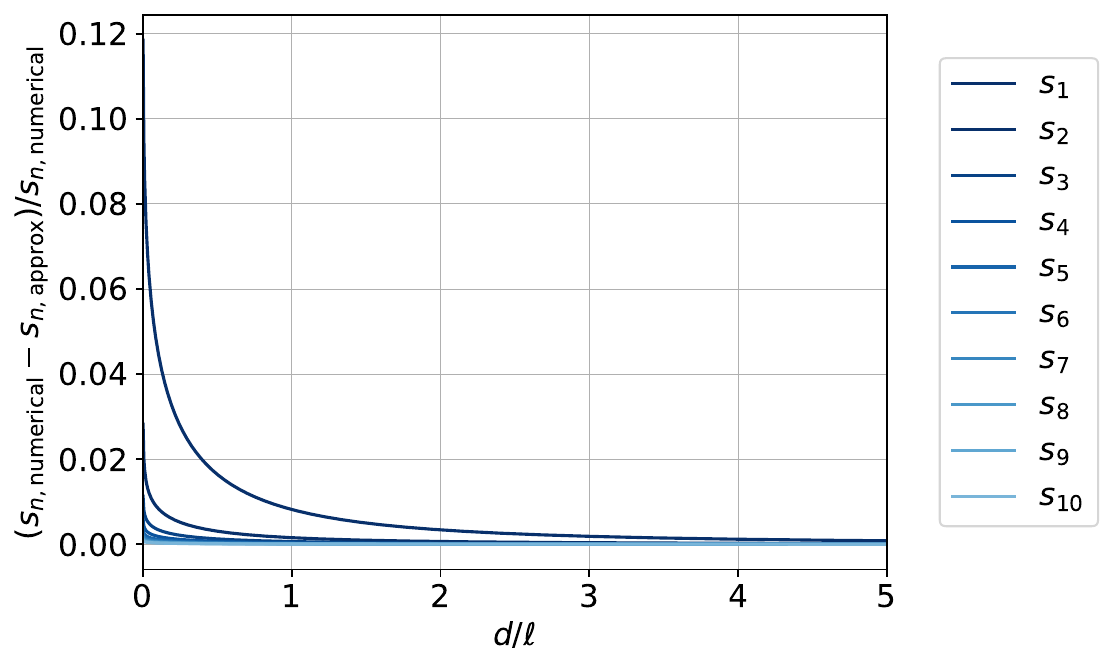}
\caption{}
\end{subfigure}
\caption{Relative error in the analytical approximation of the zeros $s_{n,\text{approx}}$ in the large separation regime using \eqref{eq:conical_zeros__large_separation}, compared to the numerical zeros $s_{n,\text{numerical}}$ found using the \texttt{root\_scalar} method from \texttt{scipy.optimize} in Python \cite{scipy}. We plot the relative error as a function of (a) $\eta$ and (b) $d/\ell$.}
\label{fig:conical_zeros__large_separation}
\end{figure*}

Since these zeros determine the spectrum of $J^\ammaG$ by $\pm i \tilde{\nu}_n = \pm i \tanh(\pi s_n)$, we can calculate the logarithmic negativity associated with each of these modes, $E_{\mathcal{N},n} := -\log_2\tilde{\nu}_n = -\log_2(\tanh(\pi s_n))$.
We plot $E_{\mathcal{N},n}$ for the first ten modes in \cref{fig:ENn__large_separation}, and compare the numerically determined zeros with our approximation.
\begin{figure*}
\centering
\begin{subfigure}[t]{0.47\textwidth}
\centering
\includegraphics[width=\linewidth]{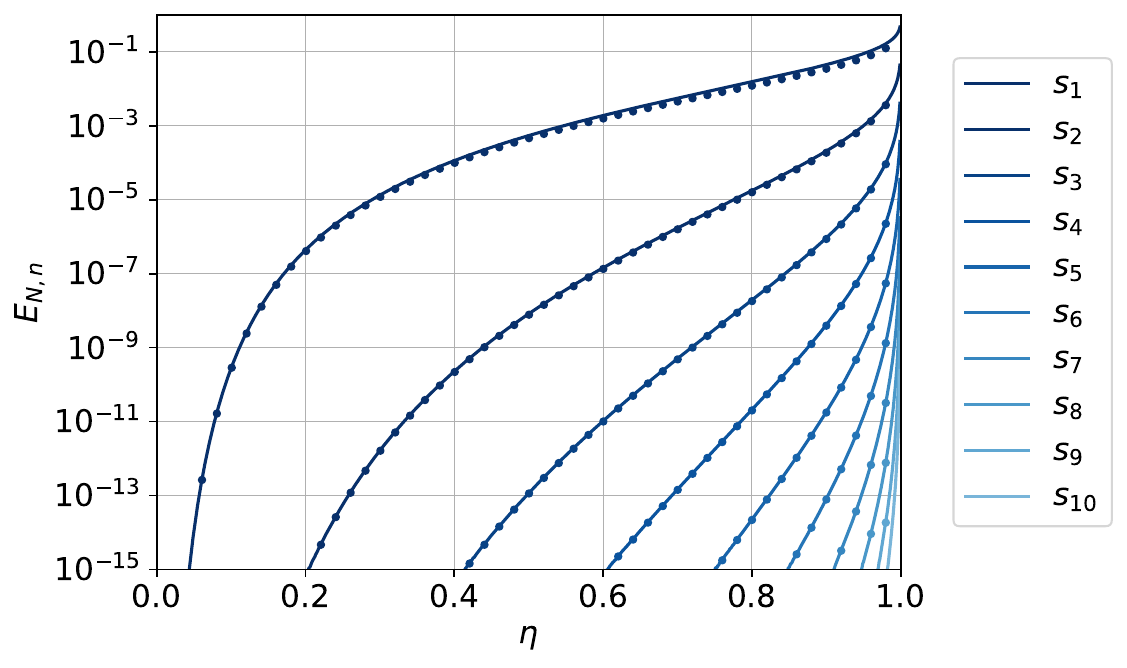}
\caption{}
\end{subfigure}
\hfill
\begin{subfigure}[t]{0.47\textwidth}
\centering
\includegraphics[width=\linewidth]{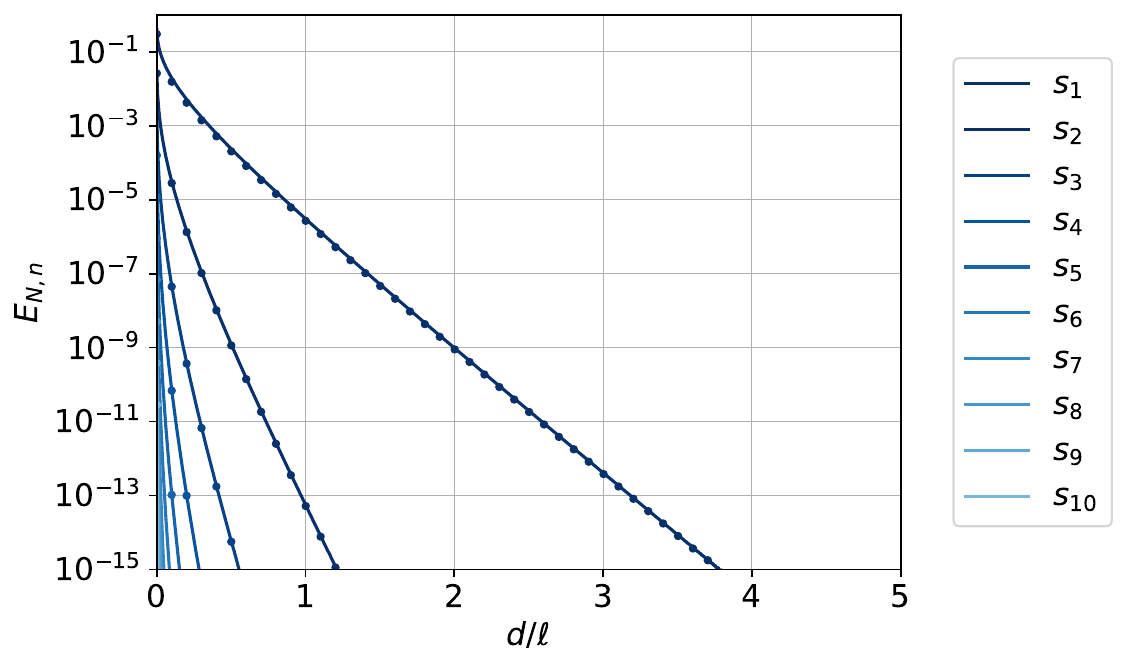}
\caption{}
\end{subfigure}
\caption{Logarithmic negativity associated with the smallest ten eigenvalues of $J^\ammaG$. The solid lines correspond to the zeros approximated in the large separation regime using \eqref{eq:conical_zeros__large_separation}, and the dots `.' correspond to the zeros computed numerically (as above). We show these quantities both as a function of (a) $\eta$ and (b) $d/\ell$.}
\label{fig:ENn__large_separation}
\end{figure*}
From the figure, we see that the logarithmic negativity of each of the modes appears to decay exponentially in $d/\ell$, with the decay rate increasing for larger $s_n$.
We can determine these decay rates analytically by examining the asymptotic behavior of \eqref{eq:conical_zeros__large_separation} as $\eta \to 0$.
We find that $s_n \sim \tfrac{j_{0,n}}{2\sqrt{\eta}}$, which leads to
\begin{align}
E_{\mathcal{N},n} \sim 2 e^{-\frac{\pi j_{0,n}}{\sqrt{\eta}}}.
\end{align}
When the lengths of the two intervals are equal, this is
\begin{align}
E_{\mathcal{N},n} \sim 2 e^{-\pi j_{0,n} \left( 1 + \frac{d}{\ell} \right)}.
\end{align}
Therefore, the logarithmic negativity associated with the $n^{\text{th}}$ mode decays exponentially in $d/\ell$ at a rate $\pi j_{0,n}$.
The total logarithmic negativity is thus dominated by the contribution of the first mode, which decays at a rate $\pi j_{0,1} \approx 7.555$, in agreement with the finding of~\cite{arias_entanglement_2026}.

\subsubsection{Small separation regime}

Now we will proceed to approximate the zeros of $P_{-is-1/2}(\cosh\rho)$ when the two intervals are close compared to their sizes, i.e., $\eta \to 1$ or $\rho \to \infty$.

For large values of $\rho$, the zeros of the conical functions can be approximated by \cite{hobson_theory_1931,zhurina_tables_1966},
\begin{align}
  \rho s_n \approx n \pi + \arctan \left( \frac{\Im [B(\tfrac12,\tfrac12+is_n)]}{\Re [B(\tfrac12,\tfrac12+is_n)]} \right), \quad n \in \mathbb{N},
\label{eq:conical_zeros__large_rho}
\end{align}
where $B$ is the beta function.
Note that the second term is simply the complex argument of $B(\tfrac12,\tfrac12+is_n)$, which is bounded between $[-\pi,\pi)$.
Thus, we see from \cref{eq:conical_zeros__large_rho} that $s_n = \mathcal{O}(\rho^{-1})$ for large $\rho$.
Using $B(\tfrac12,\tfrac12+is_n) \approx \pi - 2 \pi i s_n \ln2$ for small $s_n$, we have $\text{Arg} \, B(\tfrac12,\tfrac12+is_n) \approx - (2 \ln 2) s_n = \mathcal{O}(s_n) = \mathcal{O}(\rho^{-1})$.
This term then contributes a $\mathcal{O}(\rho^{-2})$ correction to the approximation of $s_n$, hence to leading order,
\begin{align}
  s_n \approx \frac{n \pi}{\rho} = \frac{n \pi}{\text{arccosh}\left( \frac{1+\eta}{1-\eta} \right)}, \quad n \in \mathbb{N}.
  \label{eq:conical_zeros__small_separation}
\end{align}

In \cref{fig:conical_zeros__small_separation}, we again plot the relative error between these approximations for the first ten zeros, compared with those computed numerically.
We see that this approximation will only be useful in capturing the asymptotic behavor as $\eta \to 1$.
\begin{figure}[h]
\centering
  \includegraphics[width=\linewidth]{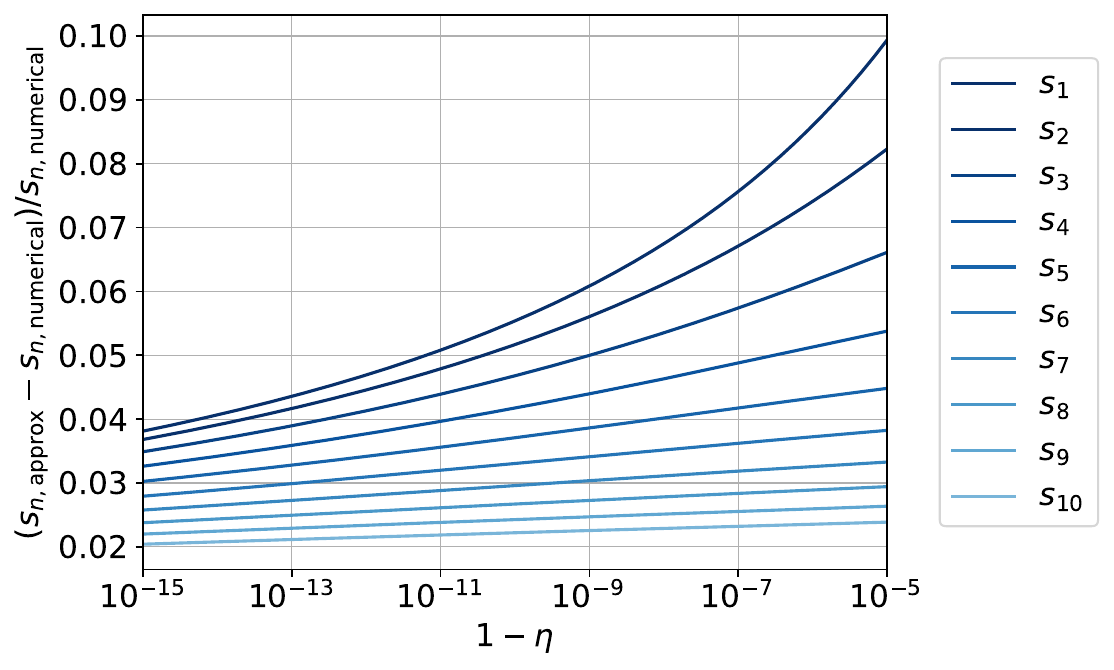}
  \caption{Relative error in the analytical approximation of the zeros $s_{n,\text{approx}}$ in the small separation regime using \eqref{eq:conical_zeros__small_separation}, compared to the numerical zeros $s_{n,\text{numerical}}$ found using the \texttt{root\_scalar} method from \texttt{scipy.optimize} in Python \cite{scipy}. We plot the relative error as a function of $1 - \eta$ (the plot against $d/\ell$ is similar, since $1-\eta \approx 2 (d/\ell)$ in this regime).}
  \label{fig:conical_zeros__small_separation}
\end{figure}
\cref{fig:ENn__small_separation} shows the corresponding logarithmic negativity of the first ten modes in this regime.
\begin{figure}[h]
\centering
  \includegraphics[width=\linewidth]{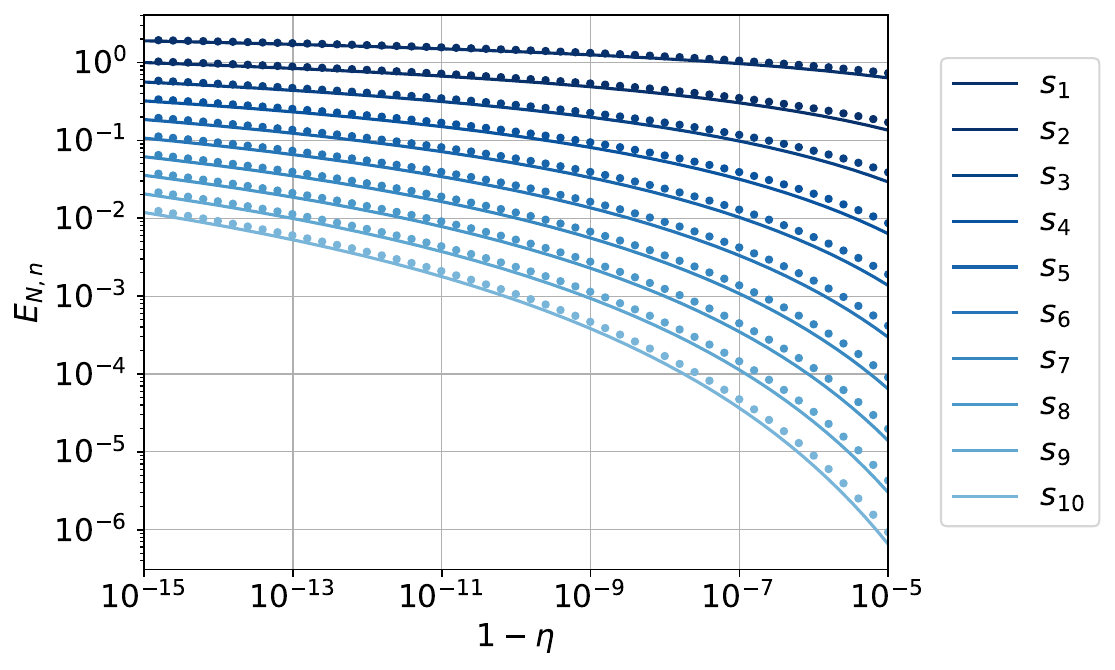}
  \caption{Logarithmic negativity associated with the smallest ten eigenvalues of $J^\ammaG$. The solid lines correspond to the zeros approximated in the small separation regime using \eqref{eq:conical_zeros__small_separation}, and the dots `.' correspond to the zeros computed numerically (as above). We show these quantities as a function of $1 - \eta$ (the plot against $d/\ell$ is similar, since $1-\eta \approx 2 (d/\ell)$ in this regime).}
  \label{fig:ENn__small_separation}
\end{figure}

Using the approximation \eqref{eq:conical_zeros__small_separation}, we can determine the asymptotic behavior of the logarithmic negativity of each of these modes when $\eta \to 1$.
We find $\rho = \text{arccosh}\left( \frac{1+\eta}{1-\eta} \right) = -\ln(1-\eta) + \ln4 + o(1)$, which yields
\begin{align}
  E_{\mathcal{N},n} \sim \log_2 ( - \ln ( 1 -\eta ) ) - \log_2(n \pi^2).
\end{align}
Or, when the lengths of the two intervals are equal,
\begin{align}
  E_{\mathcal{N},n} \sim \log_2 \left( - \ln \left( \frac{d}{\ell} \right) \right) - \log_2(n \pi^2).
\end{align}
Note that the logarithmic negativity of each mode diverges in the limit $\eta \to 1$ or $d/\ell \to 0$.

\subsubsection{Universal CFT leading order divergence}

We can also reproduce the universal 1+1D conformal field theory prediction for the leading order divergence of the total logarithmic negativity as the two intervals approach one another \cite{calabrese_entanglement_2012,calabrese_entanglement_2013,arias_entanglement_2026}.
The total logarithmic negativity is
\begin{align}
  E_\mathcal{N} = - \sum_{n=1}^\infty \log_2 ( \tanh( \pi s_n ) ).
  \label{eq:EN_sum__small_separation}
\end{align}
The sum over $n$ can be approximated by an integral in this regime, since the spacing between the zeros decreases as $\rho$ increases: $\Delta s = s_{n+1} - s_n = \frac{\pi}{\rho} + \mathcal{O}(\rho^{-2})$.
Let us write $W(s) := \text{Arg} \, B(\tfrac12,\tfrac12+is)$, then the spectral density can be written
\begin{align}
  1 = \Delta n = \frac{\Delta s}{\pi} \left( \rho - \frac{1}{\Delta s} \Delta W(s_n) \right).
\end{align}
Inserting this into \cref{eq:EN_sum__small_separation} and taking the limit $\Delta s \to 0$, we have
\begin{align}
  E_\mathcal{N} &\approx - \frac{\rho}{\pi} \frac{1}{\ln2} \int_{s_1}^\infty ds \, \ln ( \tanh(\pi s) ) \nonumber \\
  &\qquad + \frac{1}{\pi} \frac{1}{\ln2} \int_{s_1}^\infty ds \left[ \frac{d}{ds} W(s) \right] \ln ( \tanh(\pi s) ).
  \label{eq:EN_intgl__small_separation}
\end{align}
Note that $\frac{d}{ds} W(s)$ is analytic at $s=0$, so if we take the integrand of the second term and integrate over $[0,s_1]$, we find $\int_0^{s_1} ds \left[ \frac{d}{ds} W(s) \right] \ln ( \tanh(\pi s) ) = \mathcal{O}(s_1 \ln s_1)$.
This vanishes as $s_1 \to 0$ (i.e., $\rho \to \infty$), hence extending the second integral in \cref{eq:EN_intgl__small_separation} from $[s_1,\infty)$ to $[0,\infty)$ only differs by terms vanishing in this limit.
The second term in \cref{eq:EN_intgl__small_separation} is then independent of $\rho$ and only contributes a constant term to $E_\mathcal{N}$.
In contrast, the first integral in \cref{eq:EN_intgl__small_separation} is multiplied by $\rho$, hence we can find the divergent behavior of $E_\mathcal{N}$ by isolating the $\mathcal{O}(1)$ and $\mathcal{O}(s_1 \ln s_1)$ terms of this integral.
We then compute
\begin{align}
  &\int_{s_1}^\infty ds \, \ln ( \tanh(\pi s) ) \nonumber \\
  &= \int_{s_1}^\infty ds \, \left[ \ln(1 - e^{-2 \pi s}) - \ln(1 + e^{-2 \pi s}) \right]\nonumber \\
  &= \frac{1}{2\pi} \int_0^{e^{-2 \pi s_1}} \frac{dt}{t} \left[ \ln(1-t) - \ln(1+t) \right] \nonumber\\
  &= -\frac{1}{2\pi} \text{Li}_2(e^{-2 \pi s_1}) + \frac{1}{2\pi} \text{Li}_2(-e^{-2 \pi s_1}),
\end{align}
where $\text{Li}_2$ is the dilogarithm function.
We can expand the dilogarithm around $s_1=0$ by splitting the above integrals as $\int_0^{e^{-2 \pi s_1}} dt = \int_0^1 dt - \int_{e^{-2 \pi s_1}}^1 dt$.
In the second term, we then take $\frac{1}{t} =: \frac{1}{1-u} = \sum_{m=0}^\infty u^m$ and integrate termwise.
We find,
\begin{align}
  \text{Li}_2(e^{-2 \pi s_1}) &= \frac{\pi^2}{6} + 2 \pi s_1 [ \ln(2 \pi s_1) - 1 ] + \mathcal{O}(s_1^2 \ln s_1), \nonumber\\
  \text{Li}_2(-e^{-2 \pi s_1}) &= -\frac{\pi^2}{12} + (\ln2) 2 \pi s_1 + \mathcal{O}(s_1^2).
\end{align}
Putting these together, we have
\begin{align}
  E_\mathcal{N} &= \frac{1}{\ln2} \frac{\rho}{8} + \frac{\rho s_1}{\pi} \frac{1}{\ln2} [ \ln(\pi s_1) - 1 ] + \mathcal{O}(1) \nonumber\\
  &= -\frac18 \log_2(1-\eta) - \log_2(-\ln(1-\eta)) + \mathcal{O}(1).
\end{align}

The leading logarithmic divergence $-\tfrac18 \log_2(1-\eta)$ agrees with the universal 1+1D conformal field theory prediction in \cite{calabrese_entanglement_2012,calabrese_entanglement_2013,arias_entanglement_2026} (after multiplying our result by a factor of $2$ to account for the contributions from the left-moving sector).
There is a discrepancy between the coefficient of the subleading double-logarithmic term from that in \cite{calabrese_entanglement_2012,calabrese_entanglement_2013,arias_entanglement_2026}.
Subleading divergences of this kind in entanglement measures are known to arise from the infrared divergence of the (1+1)-dimensional massless field \cite{casini_entanglement_2009,mallayya_zero_2014,pye_locality_2015,yazdi_zero_2017,jain_log_2021}, suggesting that this discrepancy may be due to differences in infrared regularization.
We leave such an investigation as future work.

\subsection{Mode functions}

The eigenfunctions that we obtained above for $J^\ammaG$ are given by
\begin{widetext}
\begin{align}
  \tilde{v}_n(x) = \begin{cases}
    C_n (x-a_1)^{-\frac12-is_n} (a_2-x)^{-\frac12+is_n} (b_1-x)^{-\frac12+is_n} (b_2-x)^{-\frac12-is_n}, & \text{if $x \in (a_1,a_2)$}, \\
    - C_n (\tilde{x}-a_1)^{-\frac12-is_n} (\tilde{x}-a_2)^{-\frac12+is_n} (\tilde{x}-b_1)^{-\frac12+is_n} (b_2-\tilde{x})^{-\frac12-is_n}, & \text{if $x \in (b_1,b_2)$, where $\tilde{x} := b_1 + b_2 - x$},
  \end{cases}\label{eq:vntilde}
\end{align}
where $C_n$ is a normalization constant.
This function corresponds to the eigenvalue $-i \tilde{\nu}_n = - i \tanh(\pi s_n)$, where $s_n > 0$ is an $s$-root of the conical function $P_{-is-\frac12}((1+\eta)/(1-\eta))$.
The function $\tilde{v}_n^\ast(x)$ corresponds to eigenvalue $i \tilde{\nu}_n$.

The associated dual set of mode functions (i.e., eigenfunctions of $J^\Gamma$) can be obtained using
\begin{align}
  \tilde{u}_n(x) &= i (\Omega \tilde{v}_n^\ast)(x) = \frac{i}{4} \int_{A \cup B} dx' \, \sgn(x-x') \, \tilde{v}_n^\ast(x') \nonumber\\[1ex]
  &= \begin{cases}
    \frac{i}{2} C_n^\ast \int_{a_1}^x dx' \, (x'-a_1)^{-\frac12+is_n} (a_2-x')^{-\frac12-is_n} (b_1-x')^{-\frac12-is_n} (b_2-x')^{-\frac12+is_n}, & \text{if $x \in A$}, \\
    -\frac{i}{2} C_n^\ast \int_{b_1+b_2-x}^{b_2} dx' \, (x'-a_1)^{-\frac12+is_n} (x'-a_2)^{-\frac12-is_n} (x'-b_1)^{-\frac12-is_n} (b_2-x')^{-\frac12+is_n}, & \text{if $x \in B$}.
  \end{cases}
\end{align}
For the part in $A$, we find (after changing variables $x' = (x-a_1) t + a_1$, transforming $t \mapsto (1-t) / \left[ 1 - \left( \frac{x-a_1}{a_2-a_1} \right) t \right]$, and using 3.211 in \cite{gradshteyn_table_2007})
\begin{align}
  \tilde{u}_{nA}(x) = \frac{i}{2} (\tfrac12+is)^{-1} \frac{(a_2-x)(x-a_1)}{(a_2-a_1)} \tilde{v}_{nA}^\ast(x) F_1 \left( 1, \tfrac12 + is, \tfrac12 - is, \tfrac32 + is ; \frac{(b_1-a_2)(x-a_1)}{(a_2-a_1)(b_1-x)}, \frac{(b_2-a_2)(x-a_1)}{(a_2-a_1)(b_2-x)} \right),
\end{align}
where $F_1$ is the Appell hypergeometric function of two variables, and $\tilde{v}_{nA}^\ast(x)$ denotes the part of $\tilde{v}_n^\ast(x)$ supported in $A$.
Similarly, for $B$ we find
\begin{align}
  \tilde{u}_{nB}(x) = \frac{i}{2} (\tfrac12+is)^{-1} \frac{(b_2-\tilde{x})(\tilde{x}-b_1)}{(b_2-b_1)} \tilde{v}_{nB}^\ast(x) F_1 \left( 1, \tfrac12 + is, \tfrac12 - is, \tfrac32 + is ; \frac{(b_1-a_2)(b_2-\tilde{x})}{(b_2-b_1)(\tilde{x}-a_2)}, \frac{(b_1-a_1)(b_2-\tilde{x})}{(b_2-b_1)(\tilde{x}-a_1)} \right),
\end{align}
\end{widetext}
where $\tilde{x} = b_1 + b_2 - x$.

The magnitude of the constant $C_n \in \mathbb{C}$ is fixed by the normalization condition
\begin{align}
  \int_{A \cup B} dx \, \tilde{u}_n(x) \tilde{v}_{n'}(x) = \delta_{nn'}.
\end{align}
The complex phase of $C_n$ is arbitrary, and simply corresponds to a single-mode rotation.

\subsection{Core modes}
\label{sec:smearings}

In \cref{sec:optimal}, we explained how one can use the eigenfunctions of $J^\Gamma$ (symplectic eigenfunctions of $G^\Gamma$) to construct pairs of modes $(u_{nA},u_{nA}^\ast)$ in $A$ and $(u_{nB},u_{nB}^\ast)$ in $B$, where the logarithmic negativity of each pair corresponds to that of one of the symplectic eigenvalues of $G^\Gamma$.
This decomposition requires local symplectic orthogonality of the eigenfunctions of $J^\Gamma$ (corresponding to eigenvalues with $|\tilde{\nu}_n| < 1$) after restriction to each of the intervals, i.e., $\tilde{u}_{nA}^{\ast T} \omega \tilde{u}_{n'A} = i \alpha_n \delta_{nn'}$ and $\tilde{u}_{nB}^{\ast T} \omega \tilde{u}_{n'B} = i \beta_n \delta_{nn'}$ for some $\alpha_n, \beta_n \neq 0$.
We will now demonstrate this for our case.

Instead of working with the mode functions $\tilde{u}_n$, we can equivalently establish this condition from the covectors $\tilde{v}_n$.
Let us split these into parts supported on either $A$ or $B$ as $\tilde{v}_n = \tilde{v}_{nA} + \tilde{v}_{nB}$.
The above local symplectic orthogonality condition is then equivalent to
\begin{align}
  \tilde{v}_{nA}^{\ast T} \Omega \tilde{v}_{n'A} = i \alpha_n \delta_{nn'}, \quad & \quad \tilde{v}_{nB}^{\ast T} \Omega \tilde{v}_{n'B} = i \beta_n \delta_{nn'}, \nonumber\\
  \tilde{v}_{nA}^{T} \Omega \tilde{v}_{n'A} = 0, \qquad \; & \; \qquad \tilde{v}_{nB}^{T} \Omega \tilde{v}_{n'B} = 0.
\end{align}
Below we will show that $\tilde{v}_{nA}^{\ast T} \Omega \tilde{v}_{n'A} = \tilde{v}_{nB}^{\ast T} \Omega \tilde{v}_{n'B}$ and $\tilde{v}_{nA}^T \Omega \tilde{v}_{n'A} = \tilde{v}_{nB}^T \Omega \tilde{v}_{n'B}$.
Combined with the global symplectic orthonormality, $\tilde{v}_{nA}^{\ast T} \Omega \tilde{v}_{n'A} + \tilde{v}_{nB}^{\ast T} \Omega \tilde{v}_{n'B} = i\delta_{nn'}$ and $\tilde{v}_{nA}^T \Omega \tilde{v}_{n'A} + \tilde{v}_{nB}^T \Omega \tilde{v}_{n'B} = 0$, this is sufficient for local symplectic orthogonality with $\alpha_n = \beta_n = \tfrac12$.

\begin{figure*}
\centering
\begin{subfigure}[t]{\linewidth}
\centering
\includegraphics[width=\linewidth]{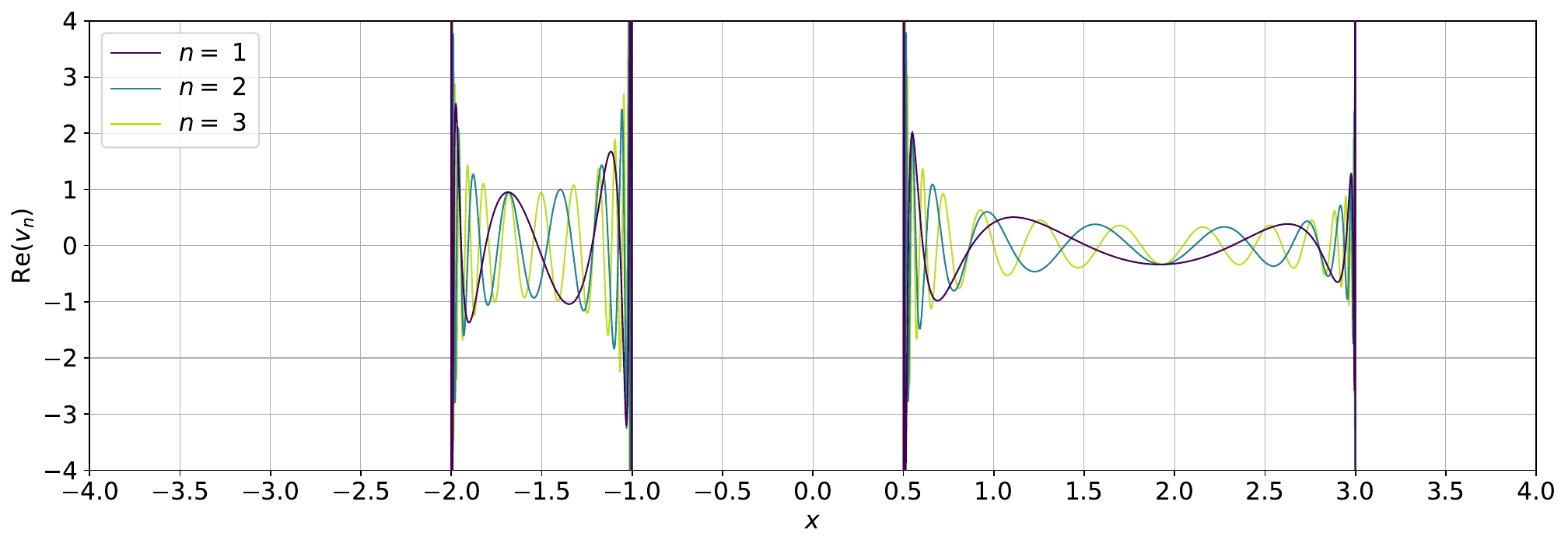}
\end{subfigure}
\begin{subfigure}[t]{\linewidth}
\centering
\includegraphics[width=\linewidth]{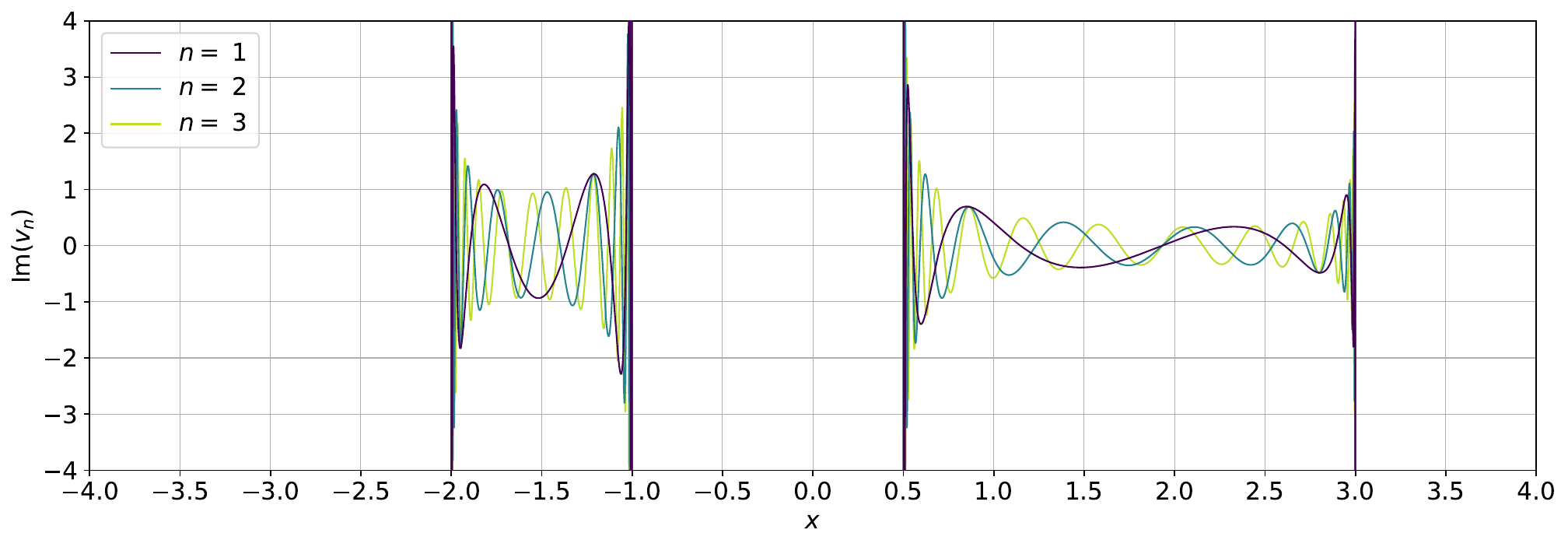}
\end{subfigure}
\caption{Co-phase space functions defining the three mode pairs with the largest contributions to the logarithmic negativity for $A = (-2,-1)$ and $B = (0.5,3)$. The two plots show the real and imaginary parts of $v_n$.}
\label{fig:mode_profiles_vn}
\end{figure*}

\begin{figure*}
\centering
\begin{subfigure}[t]{\linewidth}
\centering
\includegraphics[width=\linewidth]{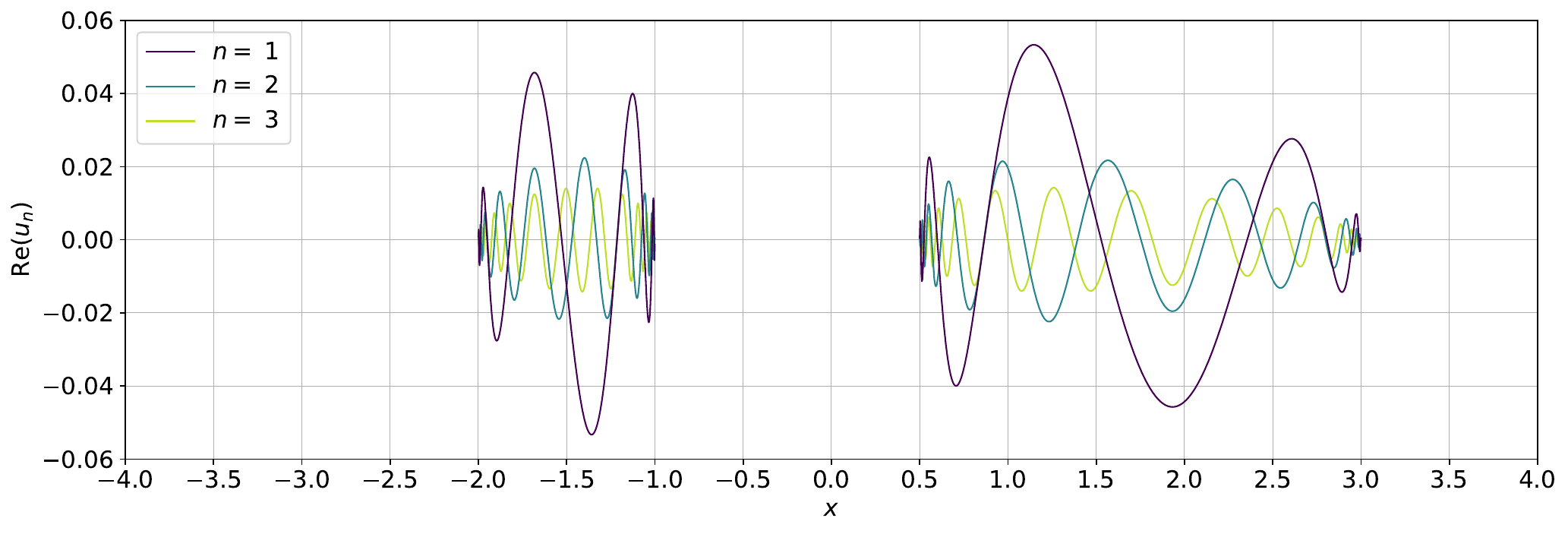}
\end{subfigure}
\begin{subfigure}[t]{\linewidth}
\centering
\includegraphics[width=\linewidth]{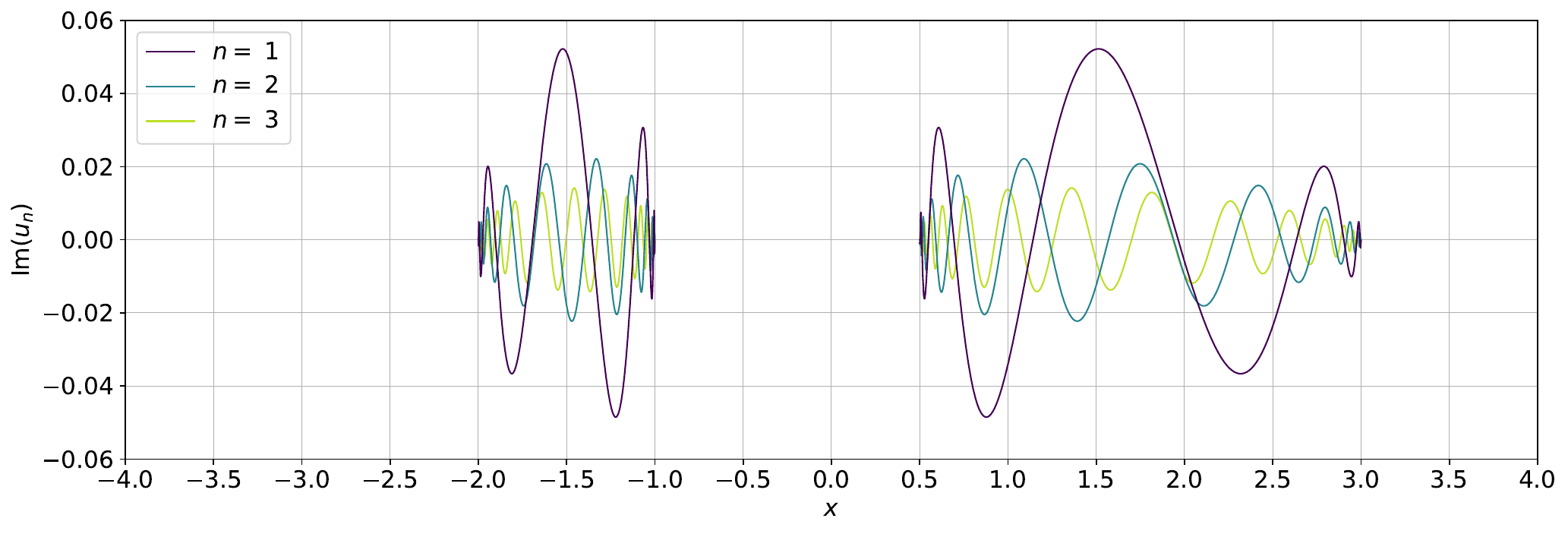}
\end{subfigure}
\caption{Mode functions defining the three mode pairs with the largest contributions to the logarithmic negativity for $A = (-2,-1)$ and $B = (0.5,3)$. The two plots show the real and imaginary parts of $u_n$. The Appell $F_1$ functions were computed using the \texttt{hyper2d} function from the Python \texttt{mpmath} library \cite{mpmath}.}
\label{fig:mode_profiles_un}
\end{figure*}

In \cref{sec:optimal}, we also showed how one can construct the desired pairs of modes in $A$ and $B$ from $\tilde{u}_{nA}$ and $\tilde{u}_{nB}$.
Here we will directly construct the corresponding pairs of smearing functions using (for $\alpha_n, \beta_n > 0$),
\begin{align}
  v_{nA} = \frac{1}{\sqrt{\alpha_n}} \tilde{v}_{nA} \qquad \text{and} \qquad v_{nB} = \frac{1}{\sqrt{\beta_n}} T^T \tilde{v}_{nB}^\ast.
\end{align}
One can then recover the corresponding mode functions using $\Omega$.

Recall, the eigenfunctions of $J^\ammaG$ that we obtained are given in \cref{eq:vntilde}.
To simplify notation, let us write $s_n$ as $s$, $v_n(x)$ as $v_s(x)$, and $\tilde{v}_s^\ast(x) = \tilde{v}_{-s}(x)$.
We can then use $\tilde{v}_s^T \Omega \tilde{v}_{s'}$ with $s \in \mathbb{R}$ to denote both $\tilde{v}_s^T \Omega \tilde{v}_{s'}$ and $\tilde{v}_s^{\ast T} \Omega \tilde{v}_{s'}$ with $s > 0$.
Starting with the interval $A$, we have
\begin{align}
  \tilde{v}_{sA}^T \Omega \tilde{v}_{s'A} = \frac14 \int_{a_1}^{a_2} dx \int_{a_1}^{a_2} dx' \, \sgn(x-x') \tilde{v}_s(x) \tilde{v}_{s'}(x').
\end{align}
After plugging in the expression for $\tilde{v}_s(x)$ in $A$, changing variables to $x = (a_2-a_1) t + a_1$, and applying the transformation $t \mapsto (1-t) / \left[ 1 - \left( \frac{a_2-a_1}{b_2-a_1} \right) t \right]$ (and the same for $x'$ and $t'$), we find
\begin{widetext}
\begin{align}
  \tilde{v}_{sA}^T \Omega \tilde{v}_{s'A} = & -\frac14 C_s C_{s'} (b_1-a_2)^{-1+is+is'} (b_2-a_1)^{-1-is-is'} \int_0^1 dt \int_0^1 dt' \, \sgn(t-t') \nonumber \\
  &\qquad \times t^{-\frac12+is} (1-t)^{-\frac12-is} \left[ 1 - \left( \frac{\eta}{1-\eta} \right) t \right]^{-\frac12+is} (t')^{-\frac12+is'} (1-t')^{-\frac12-is'} \left[ 1 - \left( \frac{\eta}{1-\eta} \right) t' \right]^{-\frac12+is'}.
  \label{eq:lso_integral}
\end{align}
\end{widetext}
For the interval $B$ we apply a similar procedure.
First we undo the transpose operation $x \to b_1 + b_2 - x$ and $x' \to b_1 + b_2 - x'$.
Then, similar to above, we write $x = (b_2-b_1) t + b_1$ and apply the transformation $t \mapsto t / \left[ 1 + \left( \frac{b_2-b_1}{b_1-a_1} \right) (1-t) \right]$ (and similar for $x'$ and $t'$).
We find that this yields the same expression as the integral for $\tilde{v}_{sA}^T \Omega \tilde{v}_{s'A}$, thus establishing local symplectic orthogonality with $\alpha_n = \beta_n = \frac12$.

Therefore, we find that the pairs of modes in $A$ and $B$ which have the same PT symplectic eigenvalues contributing to the negativity as $G^\Gamma$ are defined by the co-phase space functions
\begin{align}
  v_{nA}(x) &= \sqrt{2} C_n (x-a_1)^{-\frac12-is_n} (a_2-x)^{-\frac12+is_n} \nonumber \\
  &\qquad \quad \times (b_1-x)^{-\frac12+is_n} (b_2-x)^{-\frac12-is_n} \\[1ex]
  &= \frac{ \sqrt{2} C_n e^{-i s_n \omega(x)} }{ \sqrt{-(x-a_1)(x-a_2)(x-b_1)(x-b_2)} }, \label{eq:core_mode_A}
\end{align}
for $x \in A = (a_1,a_2)$, and
\begin{align}
  v_{nB}(x) &= -\sqrt{2} C_n^\ast (x-a_1)^{-\frac12+is_n} (x-a_2)^{-\frac12-is_n} \nonumber \\
  &\qquad \quad \times (x-b_1)^{-\frac12-is_n} (b_2-x)^{-\frac12+is_n} \\[1ex]
  &= \frac{ -\sqrt{2} C_n^\ast e^{i s_n \omega(x)} }{ \sqrt{-(x-a_1)(x-a_2)(x-b_1)(x-b_2)} }, \label{eq:core_mode_B}
\end{align}
for $x \in B = (b_1,b_2)$.
Recall $\omega(x) := \ln \left( - \frac{x-a_1}{x-a_2} \frac{x-b_2}{x-b_1} \right)$.
The corresponding mode functions are therefore $u_{nA}(x) = \sqrt{2} \tilde{u}_{nA}(x)$ and $u_{nB}(x) = \sqrt{2} \tilde{u}_{nB}^\ast(b_1 + b_2 - x)$.
We plot the real and imaginary parts of $v_n(x)$ and $u_n(x)$ in~\cref{fig:mode_profiles_vn,fig:mode_profiles_un}, respectively.

We notice from \cref{fig:mode_profiles_vn,fig:mode_profiles_un} that the mode pairs with smaller logarithmic negativity (recall $E_{\mathcal{N},1} > E_{\mathcal{N},2} > E_{\mathcal{N},3}$) exhibit higher frequency oscillations within the intervals.
This can also be seen in~\cref{eq:core_mode_A,eq:core_mode_B}, since $s_n$ increases with $n$.
Recall that the value of $s_n$ also increases as the separation between the two intervals increases (or, more generally, as $\eta$ decreases).
Thus, increasing the separation between the intervals also has the effect of increasing the frequencies of the oscillations of these functions (as well as decreasing the logarithmic negativity, since $E_{\mathcal{N},n} = -\log_2(\tanh(\pi s_n))$).
This effect was observed numerically in~\cite{klco_entanglement_2021}, where it was called a UV-IR connection, since it indicates that the long-range entanglement of the field is carried by high-frequency degrees of freedom within the intervals.
Here we have established this connection explicitly through our analytical formulas for the modes.

Recall from \cref{sec:gaussian} that the annihilation and creation operators associated with a set of normalized mode functions are given by $\hat{a}_n = v_{na} \hat{\xi}^a$ and $\hat{a}_n^\dagger = v_{na}^\ast \hat{\xi}^a$.
Those of our local mode pairs are therefore
\begin{align}
  \hat{a}_{nA} &= \int_{a_1}^{a_2} dx \, v_{nA}(x) \, \hat{\phi}_R(x), \\
  \hat{a}_{nB} &= \int_{b_1}^{b_2} dx \, v_{nB}(x) \, \hat{\phi}_R(x).
\end{align}
The corresponding quadratures in $A$ (similarly for $B$) are
\begin{align}
  \hat{Q}_{nA} &= \int_{a_1}^{a_2} dx \, \sqrt{2} \, \Re(v_{nA}(x)) \, \hat{\phi}_R(x), \\
  \hat{P}_{nA} &= \int_{a_1}^{a_2} dx \, \sqrt{2} \, \Im(v_{nA}(x)) \, \hat{\phi}_R(x).
\end{align}
Hence, $\sqrt{2} \, \Re(v_{nA})$ and $\sqrt{2} \, \Im(v_{nA})$ are the smearing functions of $\hat{\phi}_R(x)$ which can be used to couple to these modes.
These can also be written as smearing functions of $\hat{\phi}(x)$ and $\hat{\pi}(x)$ using the transformations from \cref{sec:qft_phase_space}. 
Note that, unlike the finite-dimensional case of \cref{sec:optimal}, these are distributions over $V_A$.
Therefore, in practice one would need to construct (e.g., smooth) approximations to these functions in order to approach the optimal negativity bounds in \cref{sec:optimal}.

\section{Numerical validation}
\label{sec:numerical}

Here we verify our analytical results for the logarithmic negativity in a numerical model.
We employ the following Hamiltonian as a discrete approximation to the Klein-Gordon field,
\begin{align}
  H=& \frac{\Delta x}2 \times \nonumber\\
  & \sum_{n=0}^{N-1} \left[  \pi(x_n)^2 
  + \left( \frac{\phi(x_n) - \phi(x_{n-1})}{\Delta x} \right)^2 
  \vphantom{\left( \frac{\phi(x_n) - \phi(x_{n-1})}{\Delta x} \right)^2 } 
  + \mu^2 \phi(x_n)^2 \right],
\end{align}
with $[ \phi(x_n), \pi(x_{n'}) ] =  \tfrac{i}{\Delta x}  \delta_{nn'}$ and periodic boundary conditions.
We also temporarily introduce a mass $\mu$ as an infrared regulator.
Let us rewrite the system in terms of the variables $q_n := \phi(x_n)$ and $p_n := \Delta x \, \pi(x_n)$, so that $[ q_n, p_{n'} ] = i\delta_{nn'}$ and
\begin{align}
  H = \frac{1}{2 \Delta x} \sum_{n=0}^{N-1} \left[ p_n^2 + (2 + \mu^2 \Delta x^2) q_n^2 - 2 q_n q_{n-1} \right].
\end{align}
We can diagonalize the Hamiltonian using the following expansions,
\begin{align}
  q_n &= \frac{1}{\sqrt{N}} \sum_{m=0}^{N-1} \frac{1}{\sqrt{2\omega_m}} ( a_m + a_{N-m}^\dagger ) e^{\frac{2 \pi i}{N} m n}, \\
  p_n &= \frac{1}{\sqrt{N}} \sum_{m=0}^{N-1} (-i) \sqrt{\frac{\omega_m}{2}} ( a_m - a_{N-m}^\dagger ) e^{\frac{2 \pi i}{N} m n},
\end{align}
where $[ a_m, a_{m'}^\dagger ] = i \delta_{mm'}$.
Then the Hamiltonian takes the form
\begin{align}
  H = \frac{1}{\Delta x} \sum_{m=0}^{N-1} \omega_m ( a_m^\dagger a_m + \tfrac12 ),
\end{align}
with
\begin{align}
  \omega_m = \sqrt{4 \sin^2 \left( \frac{\pi m}{N} \right) + \mu^2 \Delta x^2}.
\end{align}
The vacuum correlation functions are
\begin{align}
  G^{(q)}_{nn'} &:= \bra{0} \{ q_n, q_{n'} \} \ket{0} = \frac1N \sum_{m=0}^{N-1} \frac{1}{\omega_m} e^{\frac{2 \pi i}{N} m (n-n')}, \\
  G^{(p)}_{nn'} &:= \bra{0} \{ p_n, p_{n'} \} \ket{0} = \frac1N \sum_{m=0}^{N-1} \omega_m e^{\frac{2 \pi i}{N} m (n-n')}.
\end{align}

Now suppose we restrict to a local subsystem $A$ comprised of the lattice sites $n \in [0, N_A]$ (note that we can use $n=0$ as the first site without loss of generality because of translation invariance of the full system).
Recall that in the continuum, the subspace of phase space associated with a local subsystem is given by $C_0^\infty(A) \oplus \mathcal{F}(A)$, and the corresponding space of observables is $\mathcal{F}(A) \oplus C_0^\infty(A)$.
We saw that the vanishing integral condition of $\mathcal{F}(A)$ was essential in determining the eigenvalues of $J^\ammaG$.
Likewise, we will need to impose an analogous condition in the discrete model in order to emulate the continuum model.
Recall that the vanishing integral condition for $\mathcal{F}(A)$ arose because we consider the spatial derivatives $\phi'(x)$ to be observable, rather than $\phi(x)$.
Let us then rewrite the $q_n$'s in subsystem $A$ using the transformation $\vec{q'} = M \vec{q}$, defined by
\begin{align}
  M_{nn'} = \begin{cases}
    \delta_{0n'}, & \text{if $n = 0$}, \\
    \delta_{nn'} - \delta_{n-1,n'}, & \text{if $n \in [1,N_A]$},
  \end{cases}
\end{align}
so that $q_0' = q_0$ and $q_n' = q_n - q_{n-1}$ for $n \geq 1$.
The new $q'$ quadratures then correspond to discrete derivatives, except for the first site $q_0'$.
We can make this a symplectic transformation if we also apply $\vec{p'} = M^{-T} \vec{p}$, which is given by
\begin{align}
  (M^{-T})_{nn'} = \sum_{\ell=n}^{N_A} \delta_{\ell n'},
\end{align}
or $p_n' = \sum_{\ell=n}^{N_A} p_\ell$.
Now, notice that if we trace out the first site $(q_0',p_0')$, then for the $q'$ quadratures we are left with only derivatives, while the removal of $p_0' = \sum_{n=0}^{N_A} p_n = \sum_{n=0}^{N_A} \Delta x \, \pi(x_n)$ can be viewed as a discrete version of the vanishing integral condition in the continuum.
Further, for linear observables of the form $f^T q$ we have $f' = M^{-T} f$, so that $f^T q = f'^T q' = f_0' q_0' + f_1' q_1' + \cdots$.
We then see that removing $q_0'$ also removes $f_0' = \sum_{n=0}^{N_A} f_n$, which is analogous to the vanishing integral condition for the smearing functions of $\phi(x)$ in the continuum.

Therefore, we construct the covariance matrix for the local subsystem $A$ by first restricting $G^{(q)}$ and $G^{(p)}$ to the sites of $A$, transforming to $(q',p')$ variables to get
\begin{align}
  G \equiv \begin{bmatrix}
    M G^{(q)} M^T & 0 \\
    0 & M^{-T} G^{(p)} M^{-1}
  \end{bmatrix},
\end{align}
and then removing the rows and columns associated with $(q_0',p_0')$.
For two intervals $A$ and $B$, we apply the same transformation to each subsystem, i.e., $M_A \oplus M_B \oplus M_A^{-T} \oplus M_B^{-T}$, and remove the first primed site of each interval.
Note that removing these sites also increases the distance and decreases the sizes of the intervals by one site.
Also, for $n,n' \geq 1$, we have
\begin{align}
  (M G^{(q)} M^T)_{nn'} = \frac1N \sum_{m=0}^{N-1} \frac{1}{\omega_m} 4 \sin^2 \left( \frac{\pi m}{N} \right) e^{\frac{2 \pi i}{N} m (n-n')}.
\end{align}
The $m = 0$ drops out of the sum, hence after applying this transformation and removing $(q_0',p_0')$, we can safely take the massless limit $\mu \to 0$.

After constructing the covariance matrix for two intervals $A$ and $B$ as above, we can compute the logarithmic negativity.
Since $(q',p')$ form a canonical basis, we can apply the partial transpose by negating the $p'$ variables in $B$.
Further, $\omega = \Omega^{-1}$ takes its canonical form, so we can easily construct the operator $J^\Gamma = - G^\Gamma \omega$ and numerically compute its eigenvalues.

\begin{figure}[h]
  \centering
  \includegraphics[width=\linewidth]{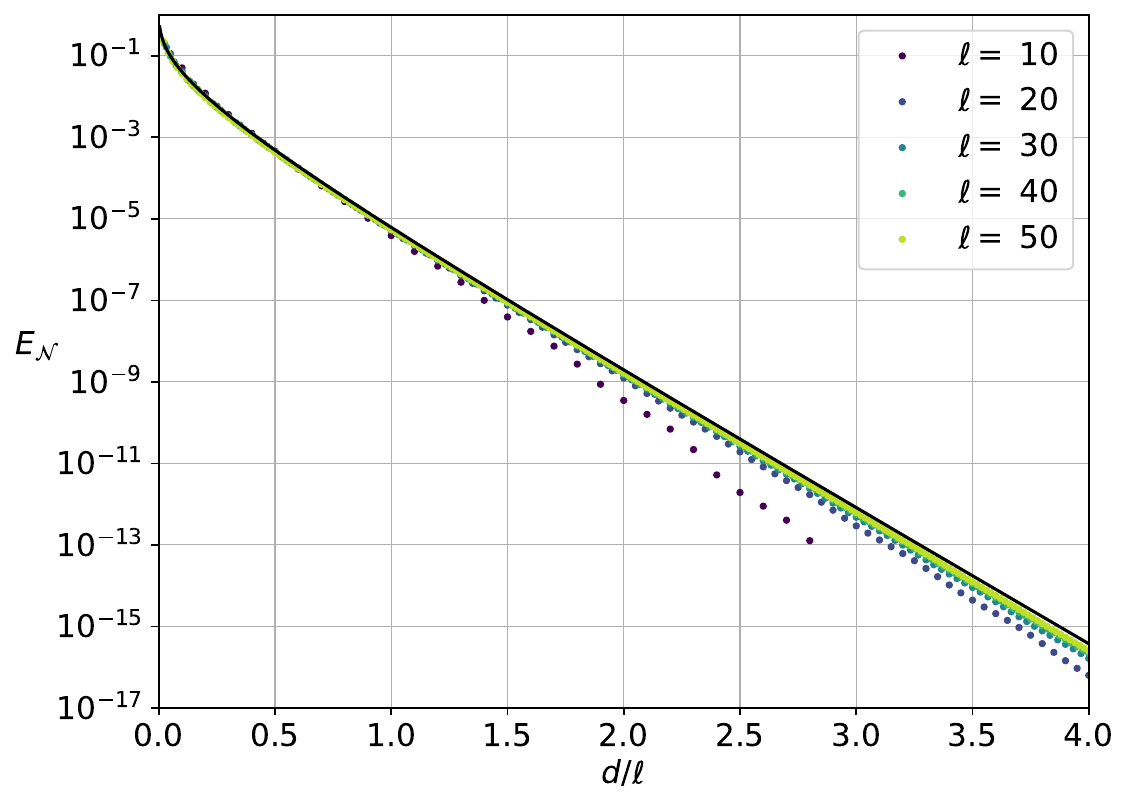}
  \caption{Numerical values of the logarithmic negativity between intervals $A$ and $B$ as a function of their separation. The numerical results are shown with dots. The parameters for this simulation were $N = 20000$, $\ell = N_A = N_B \in \{ 10, 20, 30, 40, 50 \}$, and the separation $d$ ranged from $1$ to $4 \ell$. Our analytical result is shown with a solid black line, which we approximate with $2 E_{\mathcal{N},1}$.}
  \label{fig:numerical_negativity}
\end{figure}
\cref{fig:numerical_negativity} shows the results for the total logarithmic negativity as a function of the separation between the intervals $A$ and $B$.
We see good agreement with our analytical result, which we approximate with the contribution from the mode pair with the largest logarithmic negativity, $E_{\mathcal{N},1}$ (counted twice due to the equal contributions from the right- and left-moving fields).
The numerical simulation was performed with interval sizes $\ell \in \{ 10, 20, 30, 40, 50 \}$ and $d$ ranging from $1$ to $4 \ell$ (for each $\ell$).
Larger values of $\ell$ are a better approximation of the continuum limit of the model, which we see as a better match to our analytical result as $\ell$ increases.
We also note that the numerical results deviate from the analytical result as the separation increases, and this begins to occur at smaller separations for smaller interval sizes $\ell$.
This is consistent with the fact that the mode functions with smaller negativity exhibit higher frequency oscillations (the UV-IR connection), which are less reliably captured by the discretized model for large separations or smaller values of $\ell$.

\begin{figure}[h]
  \centering
  \includegraphics[width=\linewidth]{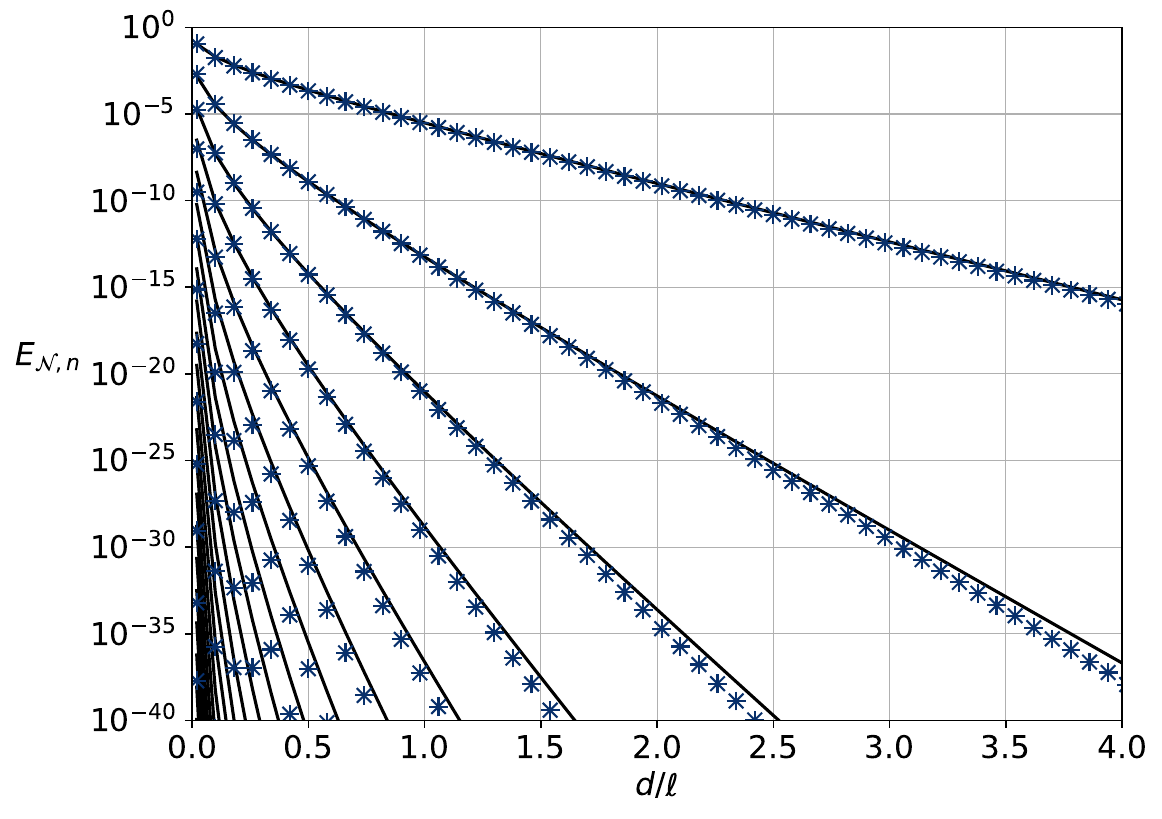}
  \caption{Numerical values of the contribution to logarithmic negativity of each eigenvalue of $J^\Gamma$. The numerical results are shown with `$\times$' or `$+$', and our analytical results with solid black lines. The parameters for this simulation were $N = 20000$, $\ell = N_A = N_B = 50$, and the separation $d$ ranged from $1$ to $4 \ell$.}
  \label{fig:numerical_negativity_permode}
\end{figure}
In~\cref{fig:numerical_negativity_permode}, we separate the contributions to the logarithmic negativity from each of the eigenvalues of $J^\Gamma$ for $\ell = 50$, and compare this to the contributions for our analytical mode pairs, $E_{\mathcal{N},n}$.
We see an accurate matching of each of our analytical modes with the numerical modes, including a twofold degeneracy of the numerical eigenvalues analogous to the equal contributions of the right- and left-moving sectors in the continuum.
We also see a deviation between the analytical and numerical results as the separation increases, which occurs at smaller separation for larger eigenvalues of $J^\Gamma$.
This is consistent with the observation that the mode pairs with smaller negativity exhibit higher frequency oscillations.

\section{Conclusion}
\label{sec:conclusion}

In this paper, we have provided a complete analytical characterization of the bipartite entanglement negativity between two disjoint intervals in a (1+1)-dimensional massless scalar field using Gaussian state methods.
We provided a calculation of the logarithmic negativity, constructed a decomposition of the two interval subsystem into pairs of modes carrying the negativity, as well as gave explicit analytical forms for the corresponding mode functions and associated smearing functions.
These smearing functions directly provide optimal detection profiles for extracting negativity from the field.
From these analytical formulas, we were also able to concretely establish the UV-IR connection from \cite{klco_entanglement_2021}.

The use of a coordinate-free framework for the phase space formalism allowed for the calculation to be performed in the right-left mover representation, which simplified our calculation.
This required us to develop a basis-independent definition of partial transposition in this framework, which we believe could be useful more generally.

Our work opens up a path for many further analytical studies on the structure of entanglement in quantum field theory.
A similar reformulation of the eigendecomposition of $J$ (the complex linear structure induced by the state on phase space) into a boundary value problem has been developed in higher dimensions \cite{arias_anisotropic_2017}, thus it would be natural to modify this approach in a similar manner as in this paper for the diagonalization of $J^\Gamma$ and $J^\ammaG$ (the partial transpose of $J$ and $J^T$, respectively) in higher dimensions.
Further, the framework of \cite{hackl_bosonic_2021} has also been developed for fermionic fields.
One could therefore perform a similar calculation for a (1+1)-dimensional massless fermionic field, which would provide another avenue for computing the negativity Hamiltonian of this system \cite{murciano_negativity_2022} using \cref{eq:negativity_hamiltonian}.
Since the expressions for the negativity Hamiltonian are much simpler for the fermionic field, comparisons could be done more explicitly.
Note that for fermions there are inequivalent definitions of partial transposition~\cite{shapourian_twisted_2019}, and thus one would need to carefully consider how the complex boundary value problem should be modified to implement either of these definitions.

We have analytically characterized the structure of the entanglement which is accessible to local observers (i.e., the negativity core).
In~\cite{klco_entanglement_2023,gao_partial_2024,gao_detecting_2025,gao_finite_2026}, it was shown that there is a rich structure of correlations between the core and halo (i.e., the remaining modes of the $AB$ system) as well as inaccessible entanglement in the halo.
It would therefore be of interest to also compute the halo modes to provide an analytical study of this structure for the (1+1)-dimensional massless scalar field in the continuum.
In principle, our approach can be used to calculate these modes, which correspond to the eigenfunctions of $J^\ammaG$ with $|\tilde{\nu}| \geq 1$.
One could also compare this decomposition to other mode decompositions, such as~\cite{wolfNotSoNormalModeDecomposition2008}.

It would also be interesting to further study the extraction of the entanglement from the negativity core modes and, for example, to address the energy cost of entanglement extraction~\cite{hackl_minimal_2019}.
Given the distributional and highly oscillatory nature of the smearing functions defining the core modes, it would be useful to examine the sensitivity of the extraction to approximations of these profiles.

For our calculation, we regulate the infrared divergence of the (1+1)-dimensional massless scalar field by assuming that the observers only couple to derivatives of the field, which are the quantities which contribute to the energy of the system.
It would be of interest to study the impact of alternative regularizations.
In particular, we would like to investigate the source of the discrepancy in the coefficient of the subleading double-logarithmic divergence from that of~\cite{{calabrese_entanglement_2012,calabrese_entanglement_2013,arias_entanglement_2026}} for the case of two adjacent intervals, which we expect arises from differing regularizations.
One approach to analyzing the effect of alternate regularization schemes would be to study different numerical models.
For example, in the work of~\cite{marcovitch_critical_2009}, studying the logarithmic negativity for a harmonic chain,
a small mass parameter is introduced to regulate the infrared (as opposed to our vanishing integral condition).
In reproducing their results, we found that in this system the spectrum of $J^\Gamma$ does not exhibit a twofold degeneracy in its eigenvalues, suggesting a persistent breaking of chiral symmetry even in the massless limit.
In contrast, our numerical model does exhibit this degeneracy.
(Recall that this is expected for the (1+1)-dimensional massless field in the continuum limit, since the right- and left-moving modes are decoupled.)
We plan to provide an in-depth comparison between different numerical models in follow-up work, studying how these different infrared regularizations of the (1+1)-dimensional massless scalar field impact the spectrum of $J$ and $J^\Gamma$, and particularly the effect on the degeneracy of the eigenvalues in the continuum limit.

\begin{acknowledgments}
JP and RHJ gratefully acknowledge support by the Wenner-Gren Foundations and the Wallenberg Initiative on Networks and Quantum Information (WINQ).
Nordita is supported in part by NordForsk.
\end{acknowledgments}

\bibliography{main.bib}

\end{document}